\documentclass[11pt]{article}

\usepackage{amsmath}
\usepackage{fullpage}
\usepackage{amsthm}
\usepackage{hyperref}

\usepackage{amssymb}

\usepackage{mathtools}
\usepackage{bbm}
\usepackage{bm}
\usepackage{graphicx}

\usepackage{xcolor}

\usepackage{algorithm}
\usepackage{algpseudocode}

\usepackage{comment}

\usepackage{authblk}

\usepackage[numbers]{natbib}

\newtheorem{assumption}{Assumption} 
\newtheorem{theorem}{Theorem}
\newtheorem{lemma}[theorem]{Lemma} 
\newtheorem{remark}[theorem]{Remark}
\newtheorem{corollary}[theorem]{Corollary}

\newcommand{\E}{\mathbb{E}}
\newcommand{\Prob}{\mathbb{P}}

\DeclareMathOperator{\diag}{diag}
\DeclareMathOperator{\tr}{tr}

\newcommand\Zijian[1]{{\color{magenta}Zijian: ``#1''}}
\newcommand\Peter[1]{{\color{blue}Peter: ``#1''}}
\newcommand\Cyrill[1]{{\color{teal}Cyrill: ``#1''}}
\newcommand\Revision[1]{#1}
\graphicspath{{Plots_EPS}}

\title{\huge Spectral Deconfounding for High-Dimensional Sparse Additive Models}

\author[1]{Cyrill Scheidegger}
\author[2]{Zijian Guo}
\author[1]{Peter B\"uhlmann}
\affil[1]{Seminar for Statistics, ETH Z\"urich}
\affil[2]{Department of Statistics, Rutgers University}

\begin{document}

\title{Spectral Deconfounding for High-Dimensional Sparse Additive Models}
\maketitle

\begin{abstract}
Many high-dimensional data sets suffer from hidden confounding \Revision{which affects both the predictors and the response of interest. In such situations, standard regression 
methods or algorithms lead to biased estimates.} This paper substantially extends previous work on \textit{spectral deconfounding} for high-dimensional linear models to the nonlinear setting and with this, establishes a proof of concept that spectral deconfounding is valid for general nonlinear models. Concretely, we propose an algorithm to estimate high-dimensional sparse additive models in the presence of hidden dense confounding: arguably, this is a simple yet practically useful nonlinear scope. We prove consistency and convergence rates for our method and evaluate it on synthetic data and a genetic data set.
\end{abstract}

\section{Introduction}
We consider estimation of nonlinear additive functions in the presence of dense unobserved confounding in the high-dimensional and sparse setting. \Revision{A regression problem is called confounded if there are variables that affect both the covariates and the outcome and the confounding is called unobserved or hidden if these variables are not observed.} Unobserved confounding is a severe problem in practice leading to large and asymptotically non-vanishing bias \Revision{and to spurious correlations. This is particularly severe if one aims for a causal interpretation of the functional form of the relationship between covariates and outcome.} While some progress on deconfounding and removing of bias has been achieved in the context of observational data for linear models, the current paper establishes the theory and methodology for nonlinear additive models with dense confounding. In particular, we build on spectral deconfounding introduced in \cite{CevidSpectralDeconfounding} which is simple and often more accurate than inferring hidden factor variables and then adjusting for them, \Revision{as also illustrated in Section \ref{sec_Experiments}.} The development of spectral deconfounding for nonlinear problems is new and requires careful theoretical analysis.
We believe it is important as it opens a path for addressing unobserved confounding in the context of nonlinear, high-dimensional regression in general. \Revision{Spectral deconfounding is based on the singular values of the design matrix, as suggested by its name. It is a simple procedure without any further tuning, and this implies a substantial advantage for practical data analysis.}

We focus in this paper on estimation, based on observational data, of high-dimensional sparse additive models in the presence of hidden confounding. More concretely, we look at the following model
\begin{equation}\label{eq_additiveIntro}
Y= f^0(X)+ H^T\psi+e \quad \text{and}\quad X= \Psi^T H +E,
\end{equation}
where $Y \in \mathbb R$ denotes the response or outcome variable, $X \in \mathbb R^p$ denotes the high-dimensional covariates, $H\in\mathbb R^q$ denotes the hidden confounders, $e\in\mathbb R$ and $E\in \mathbb R^p$ stand for random noises (which are ``suitably uncorrelated'' from $X$ and $H$, respectively, see Assumption \ref{ass_ConditionsModel0} later), and $f^0(X)=\beta_0^0+\sum_{j\in \mathcal T} f_j^0(X_j)$ is an unknown sparse additive function  with active set $\mathcal T\subset \{1,\ldots, p\}$ and $|\mathcal T|\ll p$. We assume that $H$ is low-dimensional ($q \ll p$) and that the confounding is dense \Revision{(i.e. $H$ affects many components of $X$, see Assumption \ref{ass_DimAndPsi} later)}. The goal is to accurately estimate $f^0$ and the individual component functions $f_j^0$.
Note that a naive (nonlinear) regression of $Y$ on $X$ yields an estimate of $\E[Y|X]=f^0(X)+\E[H|X]^T\psi$ (assuming $\E[e|X]=0$). Hence, an estimate of $f^0$ obtained in this naive way is biased. If the goal merely is prediction in the setting of model \eqref{eq_additiveIntro}, such a biased estimate may still appear useful at first sight. However, as argued in \cite{CevidDeconfoundingAndCausalRegularisation} for the linear case, estimating the function $f^0$ instead is desirable from the viewpoint of stability and replicability. For example, the effect of the confounder $H$ might be different for new data from another environment, such that an estimator of the form $\E[Y|X]$ fails to yield a reliable prediction. Moreover, if the confounding acts densely on $X$, $\E[Y|X]$ will not be sparse and algorithms tailored for sparsity will be the wrong choice. If one interprets \eqref{eq_additiveIntro} as a structural equations model (SEM), one can view $f^0$ as the direct causal effect of $X$ on $Y$ where the variables $X_{\mathcal T}$ are the causal parents of $Y$.

\subsection{Motivating Example}
We consider a motivating example. We fix $n=300$, $p=800$, and $q=5$ and simulate from model \eqref{eq_additiveIntro} for a nonlinear additive function $f^0(X)=\sum_{j\in \mathcal T} f_j^0(X_j)$ with $\mathcal T = \{1,2,3,4\}$. We refer to Section \ref{sec_SimResults} for the exact specification of the simulation scenario. We simulate 100 data sets and fit a high-dimensional additive model on each data set without deconfounding (``naive'') and with our deconfounded method (``deconfounded''). Histograms of the mean squared errors $\|\hat f-f^0\|_{L_2}^2$ and the size of the estimated active set are provided in Figure \ref{fig_ExampleIntroduction}.

\begin{figure}
\centering
\includegraphics[width=0.85\textwidth]{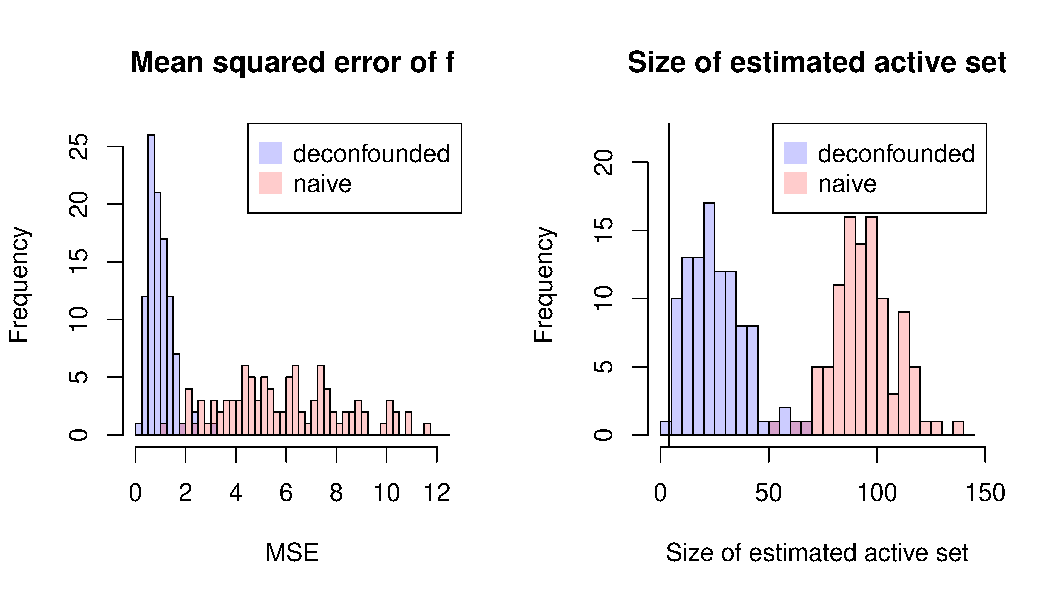}
\caption{MSE of estimated function for true $f^0$ (left) and size of estimated active set (right) for our proposed method (``deconfounded'') and the standard high-dimensional additive model fitting procedure (``naive''). The vertical bar in the right plot indicates the true size of the active set, which is $4$.} 
\label{fig_ExampleIntroduction}
\end{figure}

We see that our method clearly outperforms the standard ``naive'' approach both in terms of estimation error and also in terms of variable screening as the size of the estimated active set is much smaller, though both methods significantly overestimate the size of the active set. A more detailed simulation study with discussion can be found in Section \ref{sec_Experiments}.

\Revision{\subsection{Review of Spectral Deconfounding for Linear Models}\label{sec_ReviewLinear}}
\Revision{Spectral deconfounding has been introduced for high-dimensional sparse linear models in \cite{CevidSpectralDeconfounding}. The key new ingredients are spectral transformations which are linear transformations based on the data. Given such a transformation matrix $Q$, one simply applies $Q$ to the data and applies e.g. the Lasso to the transformed data. Constructing such a $Q$ is extremely simple: one just needs the singular value decomposition of the $n \times p$ design matrix $\mathbf X$. In its default version with the so-called trim transformation, one does \emph{not} need to specify a tuning parameter such as the dimensionality $q$ of $H$ or an upper bound of it.}

Spectral transformations have been shown to adjust (and remove) the effect of the hidden confounder $H$, under the assumption that $H$ acts densely on $X$, that is, many components of $X$ are affected by $H$. In such a scenario, one could alternatively estimate the matrix $\mathbf H\in \mathbb R^{n\times q}$ (with $n$ i.i.d. unobserved samples of $H$ as rows) by principal components of $\mathbf X$, say $\hat{\mathbf H}$, and then adjust with $\hat{\mathbf H}$. Such methodology and theory rely on fundamental results about high-dimensional latent factor models, see for example the review by \cite{BaiLargeDimensionalFactorAnalysis}. However, with such an approach, one needs to estimate an upper bound of the latent factor dimension $q$, \Revision{which can be a hard problem in practice (see also the discussion in Section \ref{sec_CompFactorModels} and the experiments in Section \ref{sec_Experiments})}. For estimating the unconfounded regression parameter, one does not necessarily need to have an accurate estimate of $\mathbf H$: spectral transformations avoid selecting an upper bound of $q$. Spectral transformations and corresponding deconfounding have been demonstrated to work very well in practice and theory in high-dimensional linear models with dense confounding \cite{CevidSpectralDeconfounding, GuoDoublyDebiasedLasso}. Even when the models are misspecified to a certain extent or when assumptions do not completely hold, extensive simulations have shown some robustness against dense (or at least fairly dense) confounding. 

These substantial practical, empirical, and theoretical advantages of spectral transformations for deconfounding remained unclear for nonlinear models. We establish here that the good properties of spectral transformations carry over to nonlinear additive models. The theoretical derivations are highly non-trivial, essentially because spectral transformations are based on $X$ but then applied to nonlinear (basis) functions $b_j(X_j) = b_j(\Psi_j^TH + E_j)$, where the hidden confounder $H$ is now in the argument of a nonlinear function $b_j(\cdot)$ but spectral deconfounding (and also PCA) are intrinsically based on linear operations. \Revision{We postpone a detailed discussion of the technical difficulties that arise from applying spectral deconfounding to nonlinear additive models to Section \ref{sec_OurContribution} and Section \ref{sec_RateDiscussion}.}

\subsection{\Revision{Additional} Related Work}
Our work is most related to the literature on \textit{spectral deconfounding}, introduced in \cite{CevidSpectralDeconfounding}\Revision{, as described in Section \ref{sec_ReviewLinear}. The idea of applying a spectral transformation to the data and using the Lasso on the transformed data turns out to be related to the Lava method for linear regression \cite{ChernozhukovLava} where the coefficient vector can be written as the sum of a sparse and a dense part.} As an extension of spectral deconfounding, a \textit{doubly debiased Lasso} estimator was proposed in \cite{GuoDoublyDebiasedLasso}, which allows to perform inference for individual components of the coefficient vector. The idea of spectral deconfounding has also been applied in \cite{BellotDeconfoundedScoreMethod} to the estimation of sparse linear Gaussian directed acylic graphs in the presence of hidden confounding.

There is an active area of research that considers variants of model \eqref{eq_additiveIntro}, mostly in the case where $f^0$ is linear, but does not use spectral transformations in the sense of \cite{CevidSpectralDeconfounding}. The following works all have in common that they, in some way explicitly, estimate the hidden confounder $H$ from $X$ or need to know or estimate the dimension $q$ of $H$ (although, in many cases, the methods can be rewritten using the PCA transformation defined in Appendix \ref{sec_ProofBoundCC}). For example, \cite{KneipFactorModels, FanFactorAdjustedRegularized, FanAreLatent} all consider regression problems, where the covariates $X$ come from a high-dimensional factor model. We refer to \cite{CevidSpectralDeconfounding} and \cite{GuoDoublyDebiasedLasso} for a more detailed discussion of related literature in the case of high-dimensional linear regression. More recently, also simultaneous inference for high-dimensional linear regression \cite{SunADecorrelatingAndDebiasingApproach} as well as estimation and inference for high-dimensional multivariate response regression \cite{BingInferenceInHDMultivariateResponse, BingAdaptiveEstimation} have been considered in the presence of hidden confounding.

There have also been some advances towards nonlinear models using this framework. In \cite{OuyangHDIGLM}, a debiased estimator is introduced for the high-dimensional generalized linear model with hidden confounding and consistency and asymptotic normality for the estimator is established.

Most recently and perhaps most related to our nonlinear setting, \cite{FanFactorAugmented}  consider a factor model $X=\Psi^T H+E$ for the covariates and a response $Y=m^\ast(H, E_{\mathcal J})+\epsilon_i$, where $\mathcal J$ is the active set. The goal is to estimate the function $m^\ast$, which is done by fitting a neural network. As a special case, this framework also allows to estimate additive models similar to \eqref{eq_additiveIntro}. However, the goal of \cite{FanFactorAugmented} is distinctively different from ours. The main goal of our paper is to consistently estimate the function $f^0$, which can be interpreted causally. For this, we implicitly filter out the factors using a spectral transformation. The goal of \cite{FanFactorAugmented} on the other hand, is to estimate the function $m^\ast$ which depends on the factors with the reason that including the factors helps to predict $Y$. \Revision{A more technical comparison of our work to high-dimensional factor models and in particular to \cite{FanFactorAugmented} can be found in Section \ref{sec_CompFactorModels}.}

Finally, for the case of unconfounded settings, high-dimensional additive models have been extensively studied as a more flexible alternative to the high-dimensional linear model while still avoiding the curse of dimensionality \citep{MeierHDAM, RavikumarSPAM, RaskuttiMinimaxOptimalRates, LinComponentSelection, YuanNonnegativeGarotte, KoltchinskiiSparseRecovery, KoltchinskiiSparsityInMultiple, TanDoublyPenalizedEstimation}.

\subsection{Our Contribution and Outline}\label{sec_OurContribution}
We propose a novel estimator for high-dimensional additive models in the presence of hidden confounding. For this, we expand the unknown functions $f_j^0$ into basis functions (e.g. B-splines) as done in \cite{MeierHDAM} and apply a spectral transformation as introduced in \cite{CevidSpectralDeconfounding} to the response and to the basis functions. On this transformed data, we apply an ordinary group lasso optimization to obtain the estimates $\hat f_j$. For this procedure, we prove consistency and provide both in-sample and $L_2$ convergence rates. Under suitable conditions, our method achieves a convergence rate of
\Revision{
$$\|\hat f-f^0\|_{L_2}=O_P\left(s^2\frac{(\log p)^{2/5}}{n^{2/5}}\right)$$
for the choice of $K\asymp (n/\log p)^{1/5}$ basis functions.
}
\Revision{The dependence on $n^{-2/5}$ is the standard dependence for fitting additive models, where the component functions are twice differentiable. However compared to the minimax optimal rate for high-dimensional additive models without confounding, the dependence on the sparsity $s$ and on $\log p$ is worse \cite{RaskuttiMinimaxOptimalRates, TanDoublyPenalizedEstimation}, see also Section \ref{sec_RateDiscussion}. We attribute this in part to the factor structure of $X$ and in part as being an artifact of the proof or our concrete estimation algorithm. We provide a more detailed discussion in Section \ref{sec_RateDiscussion}.}

\Revision{The extension of spectral deconfounding to nonlinear models is non-trivial. While some parts of the proof are similar to spectral deconfounding in the linear model \cite{CevidSpectralDeconfounding} and standard arguments for high-dimensional regression problems, there are new challenges that arise when considering nonlinear additive models. In addition to having to deal with approximation errors and centering issues when considering the approximation with basis functions, the main challenge is establishing that a group compatibility constant is bounded away from zero (also known as \textit{restricted eigenvalue condition}). This is achieved by reducing the sample compatibility constant to a population version. The population version can then be controlled using an extension of recent work on the eigenvalues of nonlinear correlation matrices to the confounded setting \cite{GuoExtremeEigenvalues}.}

\Revision{We perform a simulation study in Section \ref{sec_SimResults} comparing our method to standard additive model fitting ignoring the confounding and to an ad hoc method that tries to estimate the confounder and puts it as a linear term into the model. In conclusion, our method is shown to be the most robust against hidden confounding. In particular, it is more robust than the ad hoc method, when the components of the confounder affect $X$ not equally strongly. We complement the simulations by an application of our method to a genetic data set in Section \ref{sec_RealDataResults}.}

\Revision{The optimal rate for the high-dimensional additive model under hidden confounding is unknown. Even if our established rate might be sub-optimal, our rigorous technical analysis nevertheless establishes that spectral deconfounding can be applied to nonlinear models and this also may serve as motivation to apply spectral deconfounding to other machine learning methods.}

The rest of the paper is structured as follows. In Section \ref{sec_Method}, we introduce our setup and formulate the optimization problem. In Section \ref{sec_Theory}, we prove consistency and convergence rates for our method under suitable assumptions. We first present a general convergence result that holds under minimal assumptions (Theorem \ref{thm_BoundInSample}). This convergence rate depends on unknown quantities, namely a compatibility constant, the effect of the spectral transformation, and the best approximation of $f_j^0$ using the specified basis functions. These quantities are then subsequently controlled under some stronger assumptions. The experiments on simulated and real data can be found in Section \ref{sec_Experiments}. \Revision{All the proofs and some additional simulations }are presented in the appendix.

\subsection{Notation and Conventions}\label{sec_Notation}
We write $\lambda_j(A)$ for the $j$th largest singular value of the matrix $A$. If $A$ is symmetric and positive semi-definite, we also write $\lambda_{\max}(A)$ and $\lambda_{\min}(A)$ for the maximal and the minimal eigenvalue of $A$. We write $\|A\|_F$, $\|A\|_{op}$, and $\|A\|_\infty$ for the Frobenius norm, operator/spectral norm, and the element-wise maximum norm of the matrix $A$. For a sequence of random variables $X_n$ and a sequence of real numbers $a_n$, we write $X_n=o_P(a_n)$ if $X_n/a_n\to 0$ in probability and $X_n=O_P(a_n)$ if $\lim_{M\to\infty}\limsup_{n\to\infty}\Prob(|X_n|/a_n>M)=0$. For two sequences $a_n$ and $b_n$ of positive real numbers, we write $a_n\lesssim b_n$ if there exists a constant $C>0$ such that $a_n\leq C b_n$ for all $n\in \mathbb N$. We write $a_n\asymp b_n$ if $a_n\lesssim b_n$ and $b_n\lesssim a_n$ and $a_n\ll b_n$ if $\lim_{n\to\infty} a_n/b_n=0$. For a random variable $X$, $\|X\|_{\psi_2}=\inf\{t>0|\E[\exp(X^2/t^2)]\leq 2\}$ is the sub-Gaussian norm of $X$. We call $X$ a sub-Gaussian random variable if $\|X\|_{\psi_2}<\infty$. For a random vector $Z\in \mathbb R^d$, let $\|Z\|_{\psi_2}=\sup_{\|v\|_2=1}\|v^TZ\|_{\psi_2}$ and we call $Z$ a sub-Gaussian random vector if $\|Z\|_{\psi_2}<\infty$. We say that an event $\mathcal A$ occurs \textit{with high probability} if $\Prob(\mathcal A) = 1-o(1)$ for $n\to\infty$. For a real number $t\in \mathbb R$, we write $\lfloor t\rfloor$ for the floor function, i.e. the largest integer smaller or equal to $t$. We write $I_l$ for the $l\times l$ identity matrix and $\mathbf 1_l = (1,\ldots, 1)^T\in \mathbb R^l$ for the vector of $l$ ones. \Revision{For $p\in \mathbb N$, we also write $[p]$ for the set $\{1,\ldots, p\}$.}

\section{Model and Method}\label{sec_Method}
We consider the model
\begin{equation}\label{eq_additive}
Y= f^0(X)+ H^T\psi+e \quad \text{and}\quad X= \Psi^T H +E
\end{equation}
with random variables $H\in \mathbb R^q$, $X\in\mathbb R^p$ and $Y\in \mathbb R$, random errors $e\in \mathbb R$ and $E\in \mathbb R^p$ and fixed $\psi\in \mathbb R^q$ and $\Psi\in \mathbb R^{q\times p}$. We only observe $X$ and $Y$ and the confounder $H$ is unobserved. The goal is to estimate the unknown function $f^0$. In this work, we assume an additive and sparse structure of $f^0$, i.e. 
$$f^0(X)=\beta_0^0+\sum_{j=1}^p f_j^0(X_j)=\beta_0^0+\sum_{j\in \mathcal T}f_j^0(X_j),$$
with $\mathcal T\subset \{1,\ldots, p\}$ being the active set and $|\mathcal T| = s$. For identifiability, we assume that $\E[f_j^0(X_{j})]=0$ for all $j=1,\ldots, p$. To fix some notation, $x_1, \ldots, x_n\in\mathbb R^p$, $y_1,\ldots,y_n\in \mathbb R$ and $h_1,\ldots, h_n\in \mathbb R^q$ are i.i.d. samples from \eqref{eq_additive}. Let $\mathbf X\in \mathbb R^{n\times p}$ have rows $x_1, \ldots, x_n$, $\mathbf Y\in \mathbb R^n$ have entries $y_1, \ldots, y_n$ and $\mathbf H\in \mathbb R^{n\times q}$ have rows $h_1,\ldots, h_n$.

For each $j=1, \ldots, p$, we approximate $f_j^0$ using a set of basis functions, for example, a B-spline basis. The number of basis functions $K$ serves as a tuning parameter for smoothness. Define $b_j(\cdot)=b_j^{(n)}(\cdot)=(b_j^1(\cdot), \ldots, b_j^K(\cdot))^T$ to be the vector of basis functions for the $j$th component of $X$. The general idea of high-dimensional sparse additive models is to regress $Y$ on $(b_1(X_1)^T, \ldots, b_p(X_p)^T)^T$ using a group lasso scheme. We apply the trim transformation as in \cite{CevidSpectralDeconfounding} to deal with the hidden confounding. Let $r=\min(n, p)$ and $\mathbf X\mathbf X^T=UDU^T$ be the eigenvalue decomposition of $\mathbf X \mathbf X^T$ with matrices $U\in \mathbb R^{n\times n}$ having orthonormal columns and $D=\diag(d_1^2, \ldots, d_r^2, 0, \ldots, 0)$ with $d_1\geq \ldots \geq d_r> 0$ being the nonzero singular values of $\mathbf X$. For $l=1, \ldots, r$, define $\tilde d_l=\min(d_{\lfloor\rho r\rfloor}/d_l, 1)$ for some $\rho \in (0,1)$ and define 
 \begin{equation}\label{eq_DefQtrim}
 Q = Q^\text{trim}= U\diag(\tilde d_1, \ldots, \tilde d_r, 1, \ldots, 1) U^T.
 \end{equation}
 Usually, one takes $\rho = 0.5$, that is $Q$ shrinks the top half of the singular values of $\mathbf X$ to the median singular value of $\mathbf X$.

For $j=1,\ldots, p$, define the matrix
$$B^{(j)}= B^{(j)}(\mathbf X_{\cdot j})=\begin{pmatrix}
b_j^1(x_{1,j}) & \cdots & b_j^K(x_{1,j})\\
\vdots & \ddots & \vdots\\
b_j^1(x_{n,j}) & \cdots & b_j^K(x_{n,j})
\end{pmatrix} \in \mathbb R^{n\times K}.$$

Let $\mathbf 1_n=(1,\ldots, 1)^T\in \mathbb R^n$. We then use the group lasso estimator
\begin{equation}\label{eq_OptProblem}
\hat \beta = \arg\min_{\beta=(\beta_0,\beta_1^T,\ldots,\beta_p^T)^T\in \mathbb R^{Kp+1}} \left\{\frac{1}{n}\left\|Q(\mathbf Y-\beta_0\mathbf{1}_n-\sum_{j=1}^p B^{(j)}\beta_j)\right\|_2^2+\frac{\lambda}{\sqrt n}\sum_{j=1}^p\left \| B^{(j)}\beta_j\right\|_2\right\},
\end{equation}
and construct the estimators 
$\hat f_j(\cdot)=b_j(\cdot)^T\hat\beta_j$ and $\hat f(X)=\hat \beta_0+\sum_{j=1}^p \hat f_j(X_j)$. In the optimization problem \eqref{eq_OptProblem}, $\lambda$ serves as a tuning parameter for sparsity and $K$ as a tuning parameter for smoothness. Note that the matrices $B^{(j)}$ depend on $K$. Our method is summarized in Algorithm \ref{alg_HDAM}. Observe that we use the transformation $\tilde B^{(j)} = B^{(j)} R_j^{-1}$ and $\tilde\beta_j=R_j\beta_j$ with $R_j^T R_j = \frac{1}{n}(B^{(j)})^T B^{(j)}$  
to transform \eqref{eq_OptProblem} to an ordinary group lasso problem \citep{YuanGroupLasso} with the penalty $\lambda \sum_{j=1}^p \|\tilde{\beta}_j\|_2$.
\begin{algorithm}
\caption{Deconfounding for high-dimensional additive models}\label{alg_HDAM}
\begin{flushleft}
 \textbf{Input:} Data $\mathbf X\in \mathbb R^{n\times p}$, $\mathbf Y\in \mathbb R^n$, spectral transformation $Q\in \mathbb R^{n\times n}$, tuning parameters $\lambda$ and $K$, vectors $b_j(\cdot)$ of $K$ basis functions, $j=1,\ldots, p$.\\
 \textbf{Output:} Intercept $\hat\beta_0$ and functions $\hat f_j(\cdot)$, $j=1,\ldots, p$.
\end{flushleft}
\begin{algorithmic}
\State $B^{(j)}\gets (b_j(x_{1,j}),\ldots, b_j(x_{n,j}))^T\in \mathbb R^{n\times K}$
\State Find $R_j\in \mathbb R^{K\times K}$ such that $R_j^T R_j=\frac{1}{n} (B^{(j)})^T B^{(j)}$ \Comment{Cholesky decomposition}
\State $\tilde B^{(j)} \gets B^{(j)} R_j^{-1}$
\State $(\hat \beta_0, \hat{\tilde\beta}_1,\ldots, \hat{\tilde \beta}_p)=\arg\min\{\|Q(\mathbf Y-\beta_0 \mathbf 1_n -\sum_{j=1}^p \tilde B^{(j)}\tilde\beta_j)\|_2^2/n+\lambda\sum_{j=1}^p\|\tilde \beta_j\|_2\}$ \Comment{{Group lasso}}
\State $\hat\beta_j \gets R_j^{-1}\hat{\tilde\beta}_j$, $j=1,\ldots, p$
\State $\hat f_j(\cdot)\gets b_j(\cdot)^T\hat\beta_j$
\end{algorithmic}
\end{algorithm}

The estimator \eqref{eq_OptProblem} is similar to \cite{MeierHDAM} with the difference that we apply the spectral transformation to the first part of the objective and that we do not have an additional smoothness penalty term but regularize smoothness by the number of basis functions $K$. \Revision{Our method could easily be adapted to allow for some components of $X_j$ that only enter linearly into the model. More generally, from a theoretical perspective, it would also be possible to consider a different number of basis functions $K_j$ for each component $X_j$. However, in practice one needs to choose the number of basis functions by cross-validation, which is computationally not feasible if we allow for a different number $K_j$ for each component $X_j$.}

\subsection{Some Intuition}\label{sec_Intuition}
The intuition for the spectral deconfounding method \eqref{eq_OptProblem} is analogous to the linear case in \cite{CevidSpectralDeconfounding} and \cite{GuoDoublyDebiasedLasso}. Let $b\in \mathbb R^p$ be defined as 
\begin{equation}\label{eq_DefB}
b=\E[XX^T]^{-1}\Psi^T\psi,
\end{equation}
i.e. $X^T b$ is the best linear approximation of $H^T\psi$ by $X$ in the sense that $b=\arg\min_{b'}\E[(H^T\psi-X^Tb')^2]$. We can rewrite our model (\ref{eq_additive}) as
\begin{equation}\label{eq_ModelWithB}
Y=f^0(X) + X^T b+\epsilon, \quad \epsilon = e+ H^T \psi - X^T b.
\end{equation}
The heuristics is that -- in contrast to $\frac{1}{\sqrt n}\|\mathbf X b\|_2$ which is large due to the factor structure and large singular values of $\mathbf X$ -- the quantity $\frac{1}{\sqrt n}\|Q\mathbf X b\|_2$ converges to $0$ (see Lemma \ref{lem_QXB} below). If on the other hand, $Q$ does not shrink the vector $\mathbf f^0=(f^0(x_1), \ldots, f^0(x_n))^T\in \mathbb R^n$ too much, it seems reasonable that an $\hat f$ obtained by minimizing $\|Q\mathbf Y-Q\mathbf f\|_2$ should recover $f^0$ much better than an $\hat f$ obtained by minimizing $\|\mathbf Y-\mathbf f\|_2$.

\section{Theory}\label{sec_Theory}
In this section, we develop and describe the key mathematical results of the proposed procedure in Algorithm \ref{alg_HDAM}, and we give conditions under which our method is consistent and give rates for the convergence of $\hat f$ to $f$. We will show in Corollary \ref{cor_FinalRate} that under suitable assumptions and with the choices of $\lambda \asymp \left(\log p/n\right)^{2/5}$ and \Revision{$K\asymp (n/\log p)^{1/5}$, we obtain a rate of
\begin{equation}\label{eq_RateBeginning}
|\beta_0^0-\hat\beta_0|+\sum_{j=1}^p\|f_j^0-\hat f_j\|_{L_2}=O_P\left(s^2\frac{(\log p)^{2/5}}{n^{2/5}}\right).
\end{equation}

If instead, we allow $K$ to also depend on $s$, we obtain a convergence rate of $O_P\left(s^{11/10}\frac{(\log p)^{2/5}}{n^{2/5}}\right)$.}
However, our main results Theorem \ref{thm_BoundInSample} and Corollary \ref{cor_RateOutSample} hold under much more general conditions. The general convergence rate \eqref{eq_RateInSample} in these results depends on several general quantities like a compatibility constant and on how well the functions $f_j^0$ can be approximated by the basis functions $b_j(\cdot)$. These quantities are then subsequently controlled under stronger assumptions to arrive at the convergence rate given above. 

We start with the following assumptions on the model \eqref{eq_additive}.
\begin{assumption}\label{ass_ConditionsModel0}
\begin{enumerate}
	\item The random vectors $H$ and $E$ are centered, i.e.  $\E[H]=0\in \mathbb R^q$, $\E[E]=0\in \mathbb R^p$, and the entries of $E$ and $H$ have finite second moment. Moreover, $\E[E H^T]=0\in \mathbb R^{p\times q}$ and $\E[HH^T]= I_{q}$.
	\item 
    \Revision{Conditionally on $X$, the random variable $e$ has a sub-Gaussian distribution with $\E[e|X]=0$ a.s. and there exist constants $\sigma_e^2,C_0<\infty$ such that $\E[e^2|X]\leq\sigma_e^2$ a.s. and the sub-Gaussian norm of $e$ conditionally on $X$ is uniformly bounded by $C_0$, i.e. $\|e\|_{\psi_2|X} \coloneqq \inf\{t>0|\E[\exp(e^2/t^2)|X]\leq 2\}\leq C_0$ a.s.}
	\item $q\ll \min(n, p)$.
\end{enumerate}
\end{assumption}
The assumption $\E[E H^T]=0$ means that the random vectors $E$ and $H$ are uncorrelated. The assumption that $\E[HH^T]=I_{q}$ can be made without loss of generality. If $\E[HH^T]=\Sigma_H$, define $\tilde H=\Sigma_H^{-1/2} H$, $\tilde\Psi=\Sigma_H^{1/2}\Psi$ and $\tilde \psi=\Sigma_H^{1/2}\psi$. Then, $\E[\tilde H\tilde H^T]=I_{q}$ and we are again in the framework of model \eqref{eq_additive}. \Revision{Assertion (2) of Assumption \ref{ass_ConditionsModel0} allows for heteroscedastic errors and is more general than assuming $e$ being independent of $X$.}

For Theorem \ref{thm_BoundInSample} below, we need the following additional assumption.
\begin{assumption}\label{ass_ConditionsModel1}
	Let $\Sigma_E=\E[EE^T]$. There exist $C, c>0$ such that $c\leq\lambda_{\min}(\Sigma_E^{-1})\leq \lambda_{\max}(\Sigma_E^{-1})\leq C$.
\end{assumption}
Note that we only need a bound for the minimal eigenvalue of the precision matrix of the unconfounded part $E$ and not of $X$. This is crucial since because of the factor structure, the precision matrix of $X$ would not be nicely behaved.

For $j=1,\ldots, p$, let $f_j^\ast$ be an approximation of $f_j^0$ using the $K$ basis functions in $b_j(\cdot)$, that is
$$f_j^\ast(\cdot)= b_j(\cdot)^T\beta_j^\ast,$$
and let $f^\ast(X) = \beta_0^0+\sum_{j=1}^p f_j^\ast(X_j)$. Define the vectors $\mathbf f_j^0=(f_j^0(x_{1,j}), \ldots, f_j^0(x_{n,j}))^T\in \mathbb R^n$ and $\mathbf f^0=(f^0(x_{1}), \ldots, f^0(x_{n}))^T\in \mathbb R^n$ and similarly also $\hat{\mathbf f}_j$, $\hat{\mathbf f}$, $\mathbf f_j^\ast$ and $\mathbf f^\ast$. 

For technical reasons, we also need the following assumption on the basis functions, which is for example fulfilled for the B-spline basis (see Chapter 8 in \cite{FahrmeirRegression}).
\begin{assumption}[Partition of unity]\label{ass_BasisFunctions}
For all $j=1,\ldots, p$ and for all $x\in \textup{support}(X_j)$, we have that $b_j(x)^T\mathbf 1_K=1$.
\end{assumption}

We furthermore need to define the sample compatibility constant. For $w_0\in \mathbb R$ and $w_j\in \mathbb R^K$, $j=1,\ldots, p$, let us write $f_j^w(\cdot)=b_j(\cdot)^Tw_j$ and $f^w(x)=w_0+\sum_{j=1}^p f_j^w(x_j)$.
Moreover, for $M>0$ and $\mathcal T\subset \{1, \ldots, p\}$ define,
\begin{equation}\label{eq_DefFMTn}
\mathcal F_{M, \mathcal T}^n=\left\{f^w\colon w_0\in \mathbb R, w_j\in \mathbb R^K, \, \sum_{i=1}^n f_j^w(x_{i,j})=0,\text{and} \sum_{j\in \mathcal T^c}\frac{1}{\sqrt n}\|\mathbf f_j^w\|_2\leq M\left(|w_0|+\sum_{j\in \mathcal T}\frac{1}{\sqrt n}\|\mathbf f_j^w\|_2\right)\right\}.
\end{equation}
Note that the functions $f_j^w$ defining the functions $f^w$ in $\mathcal F_{M, \mathcal T}^n$ are empirically centered. We define the sample compatibility constant
\begin{equation}\label{eq_DefCC}
\tau_n=\inf_{\mathcal T \subset [p],\,|\mathcal T|\leq s}\inf_{f\in\mathcal F_{M,\mathcal T}^n}\frac{\frac{1}{n}\|Q\mathbf f^w\|_2^2}{w_0^2+\sum_{j=1}^p\frac{1}{n}\|\mathbf f_j^w\|_2^2}
\end{equation}
with $Q$ defined in \eqref{eq_DefQtrim}.

\begin{theorem}\label{thm_BoundInSample}
Suppose that Assumptions \ref{ass_ConditionsModel0}, \ref{ass_ConditionsModel1} and \ref{ass_BasisFunctions} hold and choose $\lambda$ as
\begin{equation}\label{eq_Lambda}
\lambda= AC_0\sqrt{\frac{K\log p}{n}}+\lambda_2 \text{ with } \lambda_2\gg \frac{\|\psi\|_2}{\sqrt{1+\lambda_q^2(\Psi)}}
\end{equation} 
for some constant $A>0$ large enough.
Then, with probability $1-o(1)$, we have that
\begin{equation}\label{eq_RateInSample}
|\beta_0^0-\hat\beta_0|+\sum_{j=1}^p\frac{1}{\sqrt n}\|\mathbf f_j^\ast-\hat{\mathbf f}_j\|_2 \lesssim r_n 
\end{equation}
with
\begin{align}
r_n=\frac{s\lambda}{\tau_n}+\frac{1}{\lambda}\frac{\|Q\mathbf Xb\|_2^2}{n} &+\sum_{j\in \mathcal T}\frac{1}{\sqrt n}\|\mathbf f_j^\ast-\mathbf f_j^0\|_2+\sum_{j\in \mathcal T}|\frac{1}{n}\sum_{i=1}^n f_j^0(x_{i,j})|\nonumber\\
&+\frac{1}{\lambda}\left(\sum_{j\in \mathcal T}\frac{1}{\sqrt n}\|\mathbf f_j^\ast-\mathbf f_j^0\|_2+\sum_{j\in \mathcal T}|\frac{1}{n}\sum_{i=1}^n f_j^0(x_{i,j})|\right)^2\label{eq_DefRn}
\end{align}
\end{theorem}
A proof can be found in Appendix \ref{sec_ProofBoundInSample}. The different components in the error term $r_n$ will be made more explicit below and in Corollary \ref{cor_FinalRate}. They have the following interpretations: \Revision{for the choice $K\asymp (n/\log p)^{1/5}$ we have $\lambda\asymp (\log p/n)^{2/5}$ (if the first term in the definition \eqref{eq_Lambda} of $\lambda$ dominates).} To control \Revision{the first} term, we thus need a lower bound on the compatibility constant $\tau_n$. The second term depends on $Q\mathbf X b$ and is due to the hidden confounding. This term is small by the properties of the trim transformation $Q$ (see also Section \ref{sec_Intuition}). The third term measures, how well we can approximate the target functions $f_j^0$ using the functions $f_j^\ast$ in the span of the $K$ basis functions $b_j(\cdot)$. \Revision{Because of the identifiability condition $\E[f_j^0(X_j)]=0$, $j=1,\ldots, p$,} the fourth term is a sum of $s$ means of centered random variables and will scale like $sn^{-1/2}\sup_j\|f_j^0\|_{L_2}$. The interpretation of the fifth term is analogous to the interpretation of the third and the fourth term. In the following sections, we will control the components of $r_n$ under stronger assumptions. 
\begin{remark}
Note that from Theorem \ref{thm_BoundInSample}, we immediately also get the same convergence rate for $|\beta_0^0-\hat\beta_0|+\sum_{j=1}^p\frac{1}{\sqrt n}\|\mathbf f_j^0-\hat{\mathbf f}_j\|_2$ (that is replacing $f_j^\ast$ by the true functions $f_j^0$). Moreover, from the additive form of the error rate, we also get the screening property \citep{BuehlmannHDStats}. If $\min_{j\in \mathcal T}\frac{1}{\sqrt n}\|\mathbf f_j^0\|_2\gg r_n$, the probability of selecting a superset of the true active set converges to $1$.
\end{remark}
The rate in Theorem \ref{thm_BoundInSample} is in-sample. To also obtain out-of-sample convergence rates, we need the following assumption on the basis functions.
\begin{assumption}\label{ass_BasisOutSample}
There exists $C>0$ such that on an event $\mathcal B$ with $\Prob(\mathcal B)=1-o(1)$, it holds that
$$\sup_{j=1, \ldots, p}\frac{\lambda_{\max}\left(\E[b_j(X_j)b_j(X_j)^T]\right)}{\lambda_{\min}\left(\frac{1}{n}(B^{(j)})^TB^{(j)}\right)}\leq C.$$
\end{assumption}
\Revision{Assumption \ref{ass_BasisOutSample} follows if both the population and the sample second moment of the basis functions evaluated at the covariates are sufficiently well-behaved.} A detailed discussion of Assumption \ref{ass_BasisOutSample} can be found in Section \ref{sec_EVDesign}.
\Revision{Let us define} the $L_2$-norm (with respect to the distribution of $X$) as $\|g\|_{L_2}=\E[g(X)^2]^{1/2}$.
\begin{corollary}\label{cor_RateOutSample}
Under Assumptions \ref{ass_ConditionsModel0}, \ref{ass_ConditionsModel1}, \ref{ass_BasisFunctions} and \ref{ass_BasisOutSample} and with $\lambda$ defined in \eqref{eq_Lambda}, we have that with probability larger than $1-o(1)$,
\begin{equation}\label{eq_RateCorOutSample}
|\beta_0^0-\hat\beta_0|+\sum_{j=1}^p\|f_j^\ast-\hat f_j\|_{L_2} \lesssim r_n
\end{equation}
with $r_n$ defined in \eqref{eq_DefRn}.
\end{corollary}
The proof can be found in Appendix \ref{sec_ProofCorOutSample}.
In the following, we focus on controlling the different components of the error term $r_n$ given in \eqref{eq_DefRn}. In Section \ref{sec_Compatibility}, we bound the compatibility constant $\tau_n$ from below. In Section \ref{sec_FurtherAnalysis}, we control the other components of $r_n$ and we show how the convergence rate \eqref{eq_RateBeginning} can be deduced.

\subsection{The Compatibility Constant}\label{sec_Compatibility}
In this section, we show that if $(H^T, E^ T)^T$ is a Gaussian random vector, the compatibility constant $\tau_n$ can be bounded from below. In a first step, we reduce the (sample) compatibility constant $\tau_n$ to a population version $\tau_0$ and in a second step, we bound the population compatibility constant $\tau_0$ from below. In addition to the Gaussianity assumption (Assumption \ref{ass_Gaussian}), we also need some more assumptions on the model (Assumption \ref{ass_DimAndPsi}) and some assumptions on the basis functions $b_j$ (Assumption \ref{ass_CondBasis}).
\begin{assumption}\label{ass_Gaussian}
    $(H^T, E^T)^T$ is a Gaussian random vector.
\end{assumption}
\begin{remark}\label{rmk_GaussWeak}
    For Theorem \ref{thm_BoundCC} (reduction of sample to population compatibility constant), the Gaussianity assumption can be weakened \Revision{to sub-Gaussian with additional constraints, most importantly $p/n\to c^\ast\in [0,\infty)$}, see also the proof in Appendix \ref{sec_ProofBoundCC}. However, the Gaussianity assumption is crucial for Theorem \ref{thm_BoundPopCC} (control of population compatibility constant).
\end{remark}
\begin{assumption}\label{ass_DimAndPsi}
Define $N=\max(p,n)$.
    \begin{enumerate}
        \item $\max(q, K)s\sqrt{\frac{\log(Kp)}{n}}=o(1)$.
        \item $\lambda_1(\Psi)/\lambda_q(\Psi)\lesssim 1$.
        \item $\lambda_q(\Psi)^2\gg s\sqrt p \max\left(\sqrt{q^3(\log N)^3)}, \sqrt{\frac{p}{n}}\sqrt{q(\log N)^2}\right)$.
        \item $\max_{l,j}|\Psi_{l,j}|\lesssim \sqrt{\log{(pq)}}$.
    \end{enumerate}
\end{assumption}
\Revision{Assertion (3) of Assumption \ref{ass_DimAndPsi} is the precise mathematical formulation of dense confounding. Intuitively it means that $H$ affects many components of $X$ \cite{CevidSpectralDeconfounding}.}
Assertions (2), (3), and (4) of Assumption \ref{ass_DimAndPsi} are motivated by similar assumptions in \cite{GuoDoublyDebiasedLasso}. In particular, note that assertion (3)  can be much less restrictive than the classical factor model assumption $\lambda_q(\Psi)^2\asymp p$ \citep{FanLargeCovarianceEstimation, FanAreLatent}. \Revision{As a simple example, assume that the sparsity $s$ and the number of confounders $q$ are fixed and that $p \asymp n$. Then, assertion (3) of Assumption \ref{ass_DimAndPsi} boils down to $\lambda_q(\Psi)^2\gg \sqrt p(\log p)^{3/2}$, which is much weaker than assuming $\lambda_q(\Psi)^2 \asymp p$.}

\Revision{Define the matrices $\hat\Sigma_j=\frac{1}{n}(B^{(j)})^TB^{(j)}$, $j = 1,\ldots p$, i.e. $\hat\Sigma_j$ is the sample second moment of the design matrix corresponding to the $j$th component of $X$.}

\begin{assumption}\label{ass_CondBasis}
\begin{enumerate}
    \item The random variables $\left(b_j^k(X_j)\right)_{k\in \mathbb N,\, j\in \mathbb N}$ are sub-Gaussian and there exists a constant $C>0$ such that for all $j,k\in \mathbb N$, we have $\|b_j^k(X_j)\|_{\psi_2}\leq C$.
    \item $\frac{1}{\min_{j=1,\ldots, p}\lambda_{\min}(\hat\Sigma_j)}=o_P\left(\sqrt{\frac{n}{\log(Kp)}}\frac{1}{Ks}\right)$.
    \item There exists $C>0$ and an event $\mathcal C$ with $\Prob(\mathcal C)=1-o(1)$ such that for all $j=1, \ldots, p$, we have $\frac{\lambda_{\max}(\hat \Sigma_j)}{\lambda_{\min}(\hat\Sigma_j)}\leq C$ on the event $\mathcal C$.
\end{enumerate}
\end{assumption}
Assertion (1) holds for example for the B-spline basis functions since they are uniformly bounded. Assertions (2) and (3) of Assumption \ref{ass_CondBasis} are related to Assumption \ref{ass_BasisOutSample} \Revision{and hold if the sample second moment of the basis functions evaluated at the covariates is sufficiently well behaved.} We postpone the detailed discussion of these assumptions to Section \ref{sec_EVDesign}, \Revision{but already note that under suitable conditions, assertion (2) can be replaced by (2') $sK^2\sqrt{\frac{\log(Kp)}{n}}=o(1)$. Note that this is a stronger requirement in terms of $s$ and $K$ than assertion (1) of Assumption \ref{ass_DimAndPsi}. In fact, this is a strong restriction on how fast the sparsity $s$ is allowed to grow, see also Section \ref{sec_RateDiscussion}.} 

We now define a population version of the compatibility constant. For this, define the set of additive functions
$$\mathcal F_{\text{add}}\coloneqq \left\{f(X)=w_0+\sum_{j=1}^p f_j(X_j)|\forall j=1,\ldots, p : \E[f_j(X_j)]=0,\, \E[f_j(X_j)^2] < \infty\right\}.$$
Note that the functions $f_j$ defining the functions $f$ in $\mathcal F_{\text{add}}$ are centered with respect to the distribution of $X_j$. Also, note that we do not have a cone-condition as for the sample version \eqref{eq_DefFMTn} anymore.
Define the population compatibility constant
$$\tau_0=\inf_{f\in \mathcal F_{\textup{add}}, a\in \mathbb R^q}\frac{\E[(f(X)-H^Ta)^2]}{w_0^2+\sum_{j=1}^p\E[f_j(X_j)^2]}.$$
\begin{theorem}\label{thm_BoundCC}
	Under Assumption \ref{ass_ConditionsModel0}, assertion (1), Assumption \ref{ass_ConditionsModel1}, Assumption \ref{ass_Gaussian}, Assumption \ref{ass_DimAndPsi} and Assumption \ref{ass_CondBasis}, assume that $\tau_0\gtrsim 1$. Then, with probability $1-o(1)$, we have that $\tau_n \gtrsim \tau_0$.
\end{theorem}
The proof is given in Appendix \ref{sec_ProofBoundCC}.
To bound the population compatibility constant $\tau_0$, we use methods from \cite{GuoExtremeEigenvalues}, but need to adapt them to our setting. Let $\Psi_j\in \mathbb R^q$ be the $j$th column of the matrix $\Psi$ and define the matrices
$$\Lambda=\Lambda_{\Psi, \Sigma_E}=\diag\left(\left(\|\Psi_j\|_2^2+(\Sigma_E)_{j,j}\right)^{-1/2}, j=1,\ldots, p\right)$$
and 
\begin{equation}\label{eq_DefAPsi}
A=A_{\Psi, \Sigma_E}=\Lambda\Sigma_E\Lambda\in \mathbb R^p,
\end{equation}
that is, the matrix $A$ has entries $A_{j,t}=\frac{(\Sigma_E)_{j,t}}{\sqrt{\|\Psi_j\|_2^2+(\Sigma_E)_{j,j}}\sqrt{\|\Psi_t\|_2^2+(\Sigma_E)_{t,t}}}$. The following result, which is a modification of Theorem 1 in \cite{GuoExtremeEigenvalues}, allows us to bound the population compatibility constant $\tau_0$ in the case of Gaussian random vectors.
\begin{theorem}\label{thm_BoundPopCC}
Under Assumption \ref{ass_ConditionsModel0}, assertion (1) and Assumption \ref{ass_Gaussian}, we have that for all $f\in \mathcal F_{\textup{add}}$ and all $a\in \mathbb R^q$,
$$\frac{\E[(f(X)-H^Ta)^2]}{w_0^2+\sum_{j=1}^p\E[f_j(X_j)^2]}\geq\lambda_{\min}(A_{\Psi,\Sigma_E}).$$
In particular, $\tau_0\geq\lambda_{\min}(A_{\Psi, \Sigma_E})$.
\end{theorem}
The proof is given in Appendix \ref{sec_ProofBoundPopCC}.
\begin{remark}
If the matrix $\Sigma_E$ is diagonal, the quantity $\lambda_{\min}(A_{\Psi, \Sigma_E})$ has a more explicit expression. If $\Sigma_E=\diag(\sigma_1^2,\ldots, \sigma_p^2)$, we have that $\lambda_{\min}(A_{\Psi, \Sigma_E})=\min_{j=1, \ldots,p}\frac{\sigma_j^2}{\|\Psi_j\|_2^2+\sigma_j^2}$. Hence, if the ratio of the confounding strength $\|\Psi_j\|_2^2$ compared to the unconfounded variance $\sigma_j^2$ is bounded uniformly in $j=1,\ldots, p$, we can bound the population compatibility constant away from zero.
\end{remark}

\subsection{Further Analysis of the Remainder Term $r_n$ and Overall Implications} \label{sec_FurtherAnalysis}
To control the second component of $r_n$ in \eqref{eq_DefRn}, \Revision{we use the following result.}
\begin{lemma}\label{lem_QXB}
Under Assumption \ref{ass_ConditionsModel0}, assertions (1) and (3), and Assumption \ref{ass_ConditionsModel1}, we have that with high probability
$$\frac{1}{n}\|Q\mathbf Xb\|_2^2\lesssim \frac{\|\psi\|_2^2}{\lambda_q(\Psi)^2}\max(1, p/n).$$
\end{lemma}
The proof can be found in Appendix \ref{sec_ProofQXB}. \Revision{A common assumption on $\lambda_q(\Psi)$ is the standard factor model assumption $\lambda_q(\Psi)\asymp \sqrt p$, which is verified in \cite{GuoDoublyDebiasedLasso} and \cite{CevidSpectralDeconfounding} for some concrete choices of $\Psi$. Under this standard factor model assumption, it follows from Lemma \ref{lem_QXB} that $\frac{1}{n}\|Q\mathbf X b\|_2^2\lesssim \|\psi\|_2^2\max(1/p, 1/n)$.}

For the third term in \eqref{eq_DefRn}, observe that by the Cauchy-Schwarz inequality and Markov's inequality, 
$\sum_{j\in\mathcal T}\frac{1}{\sqrt n}\|\mathbf f_j^\ast-\mathbf f_j^0\|_2=O_P\left(s\sup_{j\in \mathcal T}\|f_j^\ast-f_j^0\|_{L_2}\right)$. \Revision{To simplify the exposition,} we now make a concrete assumption on the size of $\sup_{j\in \mathcal T}\|f_j^\ast-f_j^0\|_{L_2}$ and verify it in Section \ref{sec_ApproxError} for some particular construction of basis functions. \Revision{The assumption essentially corresponds to all component functions $f_j^0$ being twice continuously differentiable with uniformly bounded second derivatives. However, one could also consider other levels of smoothness.}
\begin{assumption}\label{ass_ApproxError}
The approximation error of $f_j^0$ by $f_j^\ast$ satisfies
$$\sup_{j\in \mathcal T}\|f_j^\ast-f_j^0\|_{L_2}\lesssim K^{-2}.$$
\end{assumption}
For the fourth term in \eqref{eq_DefRn}, observe that by Markov's inequality, H\"older's inequality and using that $\E[f_j^0(X_j)]=0$,
 $\sum_{j\in \mathcal T}|\frac{1}{n}\sum_{i=1}^n f_j^0(x_{i,j})|=O_P\left(\frac{s}{\sqrt n}\sup_{j\in \mathcal T}\|f_j^0\|_{L_2}\right)$.
The fifth term is analogous to the third and the fourth term.


\Revision{
Putting things together, we obtain that (assuming $\sup_{j\in \mathcal T}\|f_j^0\|_{L_2}<\infty$)
\begin{equation}\label{eq_BoundRnNew}
    r_n = O_P\left( \frac{s\lambda}{\tau_n} + \frac{1}{\lambda}\frac{\|\psi\|_2^2\max(1, p/n)}{\lambda_q(\Psi)^2}+\frac{s}{K^2} + \frac{s}{\sqrt n} + \frac{1}{\lambda}\frac{s^2}{K^4} +\frac{1}{\lambda}\frac{s^2}{n}\right).
\end{equation}
Using \eqref{eq_BoundRnNew}, we get consistency of our method in a wide range of scenarios. Minimal requirements for consistency can be found in Appendix \ref{sec_AppendixConsistency}. We already note here that our method achieves consistency under weaker conditions on $\lambda_q(\Psi)$ than the standard factor model assumption $\lambda_q(\Psi) \asymp \sqrt p$.

Apart from consistency, \eqref{eq_BoundRnNew} can also be used to obtain convergence rates. To obtain a simple and comparable convergence rate, Corollary \ref{cor_FinalRate} below makes a set of concrete assumptions on $n$, $p$, $K$ and $\|\psi\|_2$ and the standard factor model assumption $\lambda_q(\Psi)\asymp \sqrt{p}$. Other assumptions are possible but lead to different convergence rates.

\begin{corollary}\label{cor_FinalRate}
Under Assumptions \ref{ass_ConditionsModel0}-\ref{ass_ApproxError}, assume that $n\lesssim p$, $\lambda_q(\Psi)\asymp \sqrt p$, $\|\psi\|_2\lesssim 1$ and that $\sup_j\|f_j^0\|_{L_2}<\infty$. Moreover, assume that the matrix $A_{\Psi, \Sigma_E}$ defined in \eqref{eq_DefAPsi} satisfies $\lambda_{\min}(A_{\Psi, \Sigma_E})\gtrsim 1$. Choose $K\asymp \left(n/\log p\right)^{2/5}$. Then, we can choose $\lambda_2$ in the definition \eqref{eq_Lambda} of $\lambda$ such that
\begin{equation}\label{eq_FinalRate}
|\beta_0^0-\hat\beta_0|+\sum_{j=1}^p\|f_j^0-\hat f_j\|_{L_2}=O_P\left(s^2\left(\frac{\log p}{n}\right)^{2/5}\right).
\end{equation}
\end{corollary}
The proof of Corollary \ref{cor_FinalRate} can be found in Appendix \ref{sec_ProofFinalRate}.
}

\Revision{
\subsubsection{Discussion of the Convergence Rate}\label{sec_RateDiscussion}
The rate in \eqref{eq_FinalRate} is just an example of the type of convergence rates that we can achieve from Theorem \ref{thm_BoundInSample} and Corollary \ref{cor_RateOutSample}. On one hand, consistency can also be obtained using relaxed assumptions on the confounding, i.e. one can relax the assumptions $\lambda_q(\Psi)\asymp \sqrt p$, $\|\psi\|\lesssim 1$ or $n\lesssim p$, but will obtain a different convergence rate (see also Appendix \ref{sec_AppendixConsistency}).
On the other hand, one can change Assumption \ref{ass_ApproxError}. Assumption \ref{ass_ApproxError} essentially corresponds to the functions $f_j^0$ being twice continuously differentiable with bounded second derivative (see Section \ref{sec_ApproxError}) but one could consider other levels of smoothness instead.

In the unconfounded setting, the minimax optimal $L_2$-error rate for the high-dimensional additive model with twice differentiable $f_j$ is known to be $\sqrt{\frac{\log(p/s)}{n} + s n^{-4/5}}$ \cite{RaskuttiMinimaxOptimalRates, TanDoublyPenalizedEstimation}. From that perspective, our result from equation \eqref{eq_FinalRate} has the correct dependence on $n$ but is slower with respect to the dependence on $s$ and $\log p$. In contrast to the unconfounded setting, we also have stronger restrictions on how fast $s$ is allowed to grow as a function of $n$, $p$, and $K$, the strongest of which comes from the second assertion of Assumption \ref{ass_CondBasis}, which is implied by $s \ll \sqrt{\frac{n}{\log(Kp) K^4}}$ under suitable conditions (see Lemma \ref{lem_ConditionBasis} below). Hence, if $K \asymp \left(n/\log p\right)^{2/5}$, we need that
$s\ll \left(n/\log p)\right)^{1/10}$. In total, both the dependence on $s^2$ in \eqref{eq_FinalRate} and the restriction $s\ll \left(n/\log p)\right)^{1/10}$ may on one hand be an artifact of our proof and also due to our rather simple algorithm. Instead of regularizing smoothness by the number of basis functions, one could use smoothness penalties as done for example in \cite{MeierHDAM, RaskuttiMinimaxOptimalRates, TanDoublyPenalizedEstimation}. On the other hand, it may also be that because of the confounding and the resulting factor structure, the dependence on the sparsity $s$ is indeed worse than what one can achieve in the standard high-dimensional additive model. In particular, the factor structure of $X$ does not allow to make assumptions like $\E[f(X)^2]\asymp\sum_{j=1}^p \E[f_j(X_j)^2]$. To infer such a condition for example from Corollary 1 in \cite{GuoExtremeEigenvalues}, we would need upper bounds on the maximum eigenvalue of the correlation matrix of $X$, which is not well-behaved due to the factor structure of $X$. Note that the spectral transformation $Q$ does not remove this factor structure since it is not applied to $\mathbf X$ itself but only to the nonlinear basis functions $B^{(j)}(\mathbf X_{.j})$. This is in contrast to spectral deconfounding for the high-dimensional linear model, where the spectral transformation is directly applied to $\mathbf X$ and essentially removes the factor structure. Consequently, spectral deconfounding achieves the same error rates as the standard Lasso in the unconfounded high-dimensional linear model \cite{CevidSpectralDeconfounding, GuoDoublyDebiasedLasso}.

\begin{remark}\label{rmk_RemarkFinalRate}
    If we instead allow $K$ to also depend on the (unknown) sparsity $s$, we can choose $K\asymp \left(\frac{ns}{\log p}\right)^{1/5}$, which yields a convergence rate of
$$O_P\left(s^{11/10}\frac{(\log p)^{2/5}}{n^{2/5}}\right),$$
see the proof of Corollary \ref{cor_FinalRate} in Section \ref{sec_ProofFinalRate}.

However, in that case, the restriction on how fast $s$ is allowed to grow becomes stronger, namely $s\ll \left(n/\log p\right)^{1/14}$.
\end{remark}

}

\subsection{Verifying Assumptions}
\subsubsection{On the Approximation Error $\|f_j^\ast-f_j^0\|_{L_2}$}\label{sec_ApproxError}
We now control the approximation error $\|f_j^\ast-f_j^0\|_{L_2}$ under some concrete assumptions. In practice, we would recommend to define $b_j(\cdot)$ as the B-spline basis with knots at the empirical quantiles of $X_j$. However, such a construction seems to be difficult to analyze theoretically (especially for the theory in Section \ref{sec_EVDesign}). For our theoretical considerations, we instead use the following construction.
\begin{assumption}\label{ass_ConstructionBasis}
Let $b_0(\cdot)\in \mathbb R^K$ be the $K$ B-spline basis functions with $K-4$ equally spaced knots in $[0,1]$, see for example \cite{ZhouLocalAsymptoticsForRegressionSplines}. Define $h=1/(K-3)$ to be the distance between two adjacent knots. For $j=1,\ldots, p$, let $F_j(\cdot)=\Prob(X_j\leq\cdot)$ be the distribution function of $X_j$.
\begin{enumerate}
	\item The basis functions $b_j(\cdot)\in \mathbb R^K$ are defined as $b_j(\cdot)=b_0(F_j(\cdot))$.
	\item The functions $F_j$ have continuous inverse $F_j^{-1}$. Moreover, the functions $f_j^0\circ F_j^{-1}:(0,1)\to \mathbb R$ are twice continuously differentiable, can be continuously extended to $[0,1]$ and there exists $C>0$ such that $\sup_{j=1, \ldots, p}\|(f_j^0\circ F_j^{-1})^{(2)}\|_\infty\leq C$.
	\item $K\sqrt{\frac{\log p +\log n}{n}}=o(1).$
\end{enumerate}
\end{assumption}
We expect that basis functions from Assumption \ref{ass_ConstructionBasis} have similar properties to the basis functions used in practice. Note that Assumption \ref{ass_BasisFunctions} (partition of unity) holds for the functions $b_j(\cdot)$ since it holds for the functions $b_0(\cdot)$. Note also that assertion (2) of Assumption \ref{ass_ConstructionBasis} is reasonable, if the functions $f_j^0(x_j)$ converge to a constant for large $|x_j|$.

\begin{lemma}\label{lem_ApproxError}
Under Assumption \ref{ass_ConstructionBasis}, assertions (1) and (2), there exist $(\beta_j^\ast )_{j=1, \ldots, p}$ in $\mathbb R^K$ such that the functions $f_j^\ast(\cdot)=b_j(\cdot)^T\beta_j^\ast$ satisfy $\sup_{j=1,\ldots, p}\|f_j^\ast-f_j^0\|_{L_2}\lesssim K^{-2}$.
\end{lemma}
The proof can be found in Appendix \ref{sec_ProofApproxError}.
\subsubsection{On the Eigenvalues of Second Moments of the Basis Functions}\label{sec_EVDesign}
We now justify Assumptions \ref{ass_BasisOutSample} and \ref{ass_CondBasis} on the minimal and maximal eigenvalues of the matrices $\Sigma_j=\E[b_j(X_j)b_j(X_j)^T]$ and $\hat\Sigma_j=\frac{1}{n}(B^{(j)})^TB^{(j)}$.
The following Lemma is a variant of Lemmas 6.1 and 6.2 in \cite{ZhouLocalAsymptoticsForRegressionSplines}.
\begin{lemma}\label{lem_ConditionBasis}
Under Assumption \ref{ass_ConstructionBasis}, assertions (1) and (3), there exist constants $0<M_1, M_2<\infty$ and a random variable $S_n\geq 0$ with $S_n=o_P(h)$ such that for all $j=1,\ldots, p$,
\begin{align}
\lambda_{\max}(\Sigma_j)&\leq M_1 h, &\lambda_{\min}(\Sigma_j)&\geq M_2 h\label{eq_EVBasisPop}\\
\lambda_{\max}(\hat\Sigma_j)&\leq M_1 h+ S_n, &\lambda_{\min}(\hat\Sigma_j)&\geq M_2h- S_n.\label{eq_EVBasisSam}
\end{align}
In particular, Assumption \ref{ass_BasisOutSample} and assertion (3) of Assumption \ref{ass_CondBasis} are fulfilled. Moreover, one can replace assertion (2) of Assumption \ref{ass_CondBasis} by $sK^2\sqrt{\frac{\log(Kp)}{n}}=o(1)$.
\end{lemma}
The proof can be found in Appendix \ref{sec_ProofConditionBasis}.

\subsection{Comparison with Factor Models}\label{sec_CompFactorModels}
\Revision{Various works have considered variants of model \ref{eq_additive}, where $f^0$ is a linear function, see for example \cite{KneipFactorModels, FanFactorAdjustedRegularized, FanAreLatent}. They assume a factor model for the covariates and include the estimated factors as additional predictors in the high-dimensional regression model. Although the motivation of those methods often is not mainly hidden confounding but obtaining better model selection and prediction, those methods solve a similar problem.

In this spirit, an alternative approach to estimate $f^0$ in model \eqref{eq_additive} would be to estimate the confounding $\mathbf H$ and fit a standard additive model for $\mathbf X$ with an additional linear term in the estimate $\hat{\mathbf H}$ of $\mathbf H$. Such a method is also implemented for the experiments in Section \ref{sec_Experiments} as a comparative method. There, we can observe that this method works very well as long as the covariates very clearly follow a factor model, but our deconfounding methodology proves to be more stable when the factor structure is less clear, {i.e. there is no clear gap in the spectrum of $\mathbf X$.}  From the theoretical side, a natural question is if our assumptions (in particular Assumption \ref{ass_DimAndPsi}) are weaker than what is required to get a consistent estimate of $\mathbf H$. When the confounding dimension $q$ is known, it follows from Lemma \ref{lem_HatH} in Appendix \ref{sec_ProofBoundCC} that it is indeed possible to get a consistent estimate $\hat{\mathbf H}$ of $\mathbf H$ up to rotation, even though our assumptions are weaker than the standard factor model assumptions $\lambda_q(\Psi)\asymp \sqrt{p}$. In practice, one does not know $q$, but needs to estimate it from the data. 
It was pointed out by a reviewer that one should expect that the factor dimension $q$ can be identified as soon as $\lambda_q(\Psi)\gg \sqrt{\max(1, p/n)}$\footnote{If we assume $\lambda_{\max}(\Sigma_E)$ and $\lambda_1(\Psi)/\lambda_q(\Psi)$ being bounded and $q \ll \min(n, p)$, it follows that $\|\mathbf E\|_{op}\lesssim \sqrt{\max(n, p)}$. By Weyl's inequality, $|\lambda_q(\Psi^T\mathbf H^T\mathbf H\Psi)-n\lambda_q(\Psi^T\Psi)|\leq  n\|\Psi\|_{op}^2 \|\mathbf H^T\mathbf H/n-I_q\|_{op} = C n\lambda_q(\Psi)^2 o_P(1)$. Hence, $|\lambda_q(\mathbf H \Psi)- \sqrt n \lambda_q(\Psi)|\lesssim \sqrt n \lambda_q(\Psi)^2 o_P(1)$  and $\lambda_q(\mathbf H\Psi)\asymp \lambda_q(\Psi)\sqrt n$. Again by Weyl's inequality for singular values, it holds that $|\lambda_q(\mathbf X)-\lambda_q(\mathbf H\Psi)|\leq \|\mathbf E\|_{op}$ and $\lambda_{q+1}(\mathbf X) = |\lambda_{q+1}(\mathbf X)-\lambda_{q+1}(\mathbf H\Psi)|\leq \|\mathbf E\|_{op}$. Hence, if $\sqrt{n}\lambda_q(\Psi)\gg \sqrt{\max(n, p)}$, we have that $\lambda_{q+1}(\mathbf X)$ is of strictly smaller order than $\lambda_q(\mathbf X)$, so asymptotically, we would expect that $q$ can be identified by thresholding the singular values of $\mathbf X$.}, which is implied by our Assumption \ref{ass_DimAndPsi}.
In practice, various methods have been proposed to estimate $q$, each coming with a slightly different set of assumptions. We refer to \cite{OwenBiCV} for a systematic review. In our simulations below, we will use the eigenvalue ratio method \cite{AhnEigenvalueRatioTest, LamFactorModeling}, that simply picks $\hat q$ that maximizes the ratio $\lambda_l(\mathbf X\mathbf X^T)/\lambda_{l+1}(\mathbf X\mathbf X^T)$ of adjacent eigenvalues of $\mathbf X\mathbf X^T$. In \cite{AhnEigenvalueRatioTest}, the consistency of this estimator is proved only under conditions similar to the standard factor model assumption.\footnote{To be precise, Assumption A (i) in \cite{AhnEigenvalueRatioTest}, requires that $\lambda_q(\Psi\Psi^T/p \mathbf H^T\mathbf H/n)$ converges to a finite value. Since $\mathbf H^T\mathbf H/n\to \E[HH^T]=I_q$, this implies that $\lambda_q(\Psi)\asymp \sqrt p$.} Other methods have weaker assumptions on the factor strength but assume that $\Sigma_E$ is diagonal, $q$ is fixed or $p/n$ converges to some finite constant \cite{OwenBiCV, DobribanPermutationMethods, DobribanDeterministicPA, FanEstimatingNumberOfFactors, OnatskiDeterminingTheNumberOfFactors}.
To summarize: while Assumption \ref{ass_DimAndPsi} is strong enough such that given the dimension, the factors can be estimated consistently, we are not aware of a method to estimate the factor dimension that is proved to work in every possible scenario covered by our assumptions. Nevertheless, there is good reason to believe that our assumptions are strong enough such that also the factor dimension $q$ can be consistently estimated. In Section \ref{sec_Experiments}, we will see that in practice it is advantageous to use our method, as soon as the factor structure is not so clear. The simple reason for this is that for finite samples, we can always find the median singular value, whereas it can be harder to find the right gap in the spectrum of $\mathbf X$.


Very recently, \citet{FanFactorAugmented} proposed a neural network estimator that also assumes a factor model $X = \Psi^T H + E$ for the covariates and a response $Y = m^\ast(H, E_{\mathcal J}) + \epsilon_i$, where $\mathcal J$ is the active set. As a special case, they also consider the high-dimensional additive model (see Appendix B there). However, they look at the problem from a different angle and their method is distinctively different from ours. }Their goal mainly is prediction of $Y$, whereas we want to estimate the functional dependence of $Y$ on $X$. On the more technical side, their asymptotic results are on the minimizer of an objective function (equation (B.2) in their work) over some space of deep ReLU networks, where it is not clear that a concrete implementation indeed finds those minimizers. In contrast, our method only relies on an ordinary group lasso optimization. Since the method in \cite{FanFactorAugmented} does not exactly estimate the function $f^0$, but a function depending on the factors $H$ and the errors $E$, we cannot directly compare the convergence rates. However, note that the squared $L_2$-rate given in Theorem 6 in their work (for $\gamma^\ast = 2$) is at least $O(s^2(\log^6 n/n)^{4/9})$ which is significantly slower \Revision{with respect to $n$} than our rate from Corollary \ref{cor_FinalRate}. This slower rate is attributed to the lack of a restricted strong convexity condition, whereas we investigate this issue in detail and actually provide a lower bound on the compatibility constant when the vector $(X^T, H^T)^T$ is jointly Gaussian (see Section \ref{sec_Compatibility}).

\section{Experiments}\label{sec_Experiments}
\subsection{Practical Considerations}
We implement Algorithm \ref{alg_HDAM} from Section \ref{sec_Method}.
For our implementation, we choose $b_j(\cdot)$ to be the vector of $K$ B-spline basis functions with knots at the empirical quantiles of $\mathbf X_{\cdot, j}$, $j=1,\ldots, p$. The method depends on the choice of the two tuning parameters $\lambda$ and $K$ which control sparsity and smoothness. In principle, one could also regard the trimming threshold for $Q=Q^\text{trim}$ as a tuning parameter, but as argued in \cite{CevidSpectralDeconfounding} and \cite{GuoDoublyDebiasedLasso}, it is usually sufficient to use the median singular value of $\mathbf X$ as trimming threshold. We use a 5-fold cross-validation scheme to choose the optimal $(K_0,\lambda_0)$ from a two-dimensional grid. Afterwards, we fix $K_0$ and select the optimal $\lambda$ for $K_0$ using cross-validation with a finer grid for $\lambda$, since we believe that the choice of $\lambda$ is more important than the choice of $K$. We do the cross-validation on the transformed data: we calculate the spectral transformation $Q$ on the full data $\mathbf X$ and choose $K$ and $\lambda$ to minimize the prediction error of $Q\mathbf Y$ by rows of $(QB^{(1)}, \ldots, Q B^{(p)})$. If we used cross-validation on the untransformed data, we would also fit the confounding effect $H^T\psi$, which would result in biased estimates. Doing the cross-validation on the transformed data has the disadvantage that the rows of the data $Q\mathbf Y$ and $(QB^{(1)}, \ldots, Q B^{(p)})$ are not independent anymore. However, it seems to perform reasonably well in practice.

\Revision{As a comparative method, we also implement an ad hoc method that explicitly estimates the confounding variables from $\mathbf X$, see also Section \ref{sec_CompFactorModels}. For this, we first estimate $q$ using the eigenvalue ratio method \cite{AhnEigenvalueRatioTest, LamFactorModeling}, which is also used in \cite{FanAreLatent} in a linear version of this procedure. Then, we use the eigenvectors of $\mathbf X \mathbf X^T$ associated with the $\hat q$ largest eigenvalues as an estimate $\hat{\mathbf H}$. The estimate $\hat{\mathbf H}$ is then added as an unpenalized linear term together with the basis functions into the group lasso objective. More details are given in Algorithm \ref{alg_AdHoc}.

\begin{algorithm}
\caption{Estimated factors method for high-dimensional additive model with confounding}\label{alg_AdHoc}
\begin{flushleft}
 \textbf{Input:} Data $\mathbf X\in \mathbb R^{n\times p}$, $\mathbf Y\in \mathbb R^n$, tuning parameters $\lambda$ and $K$, vectors $b_j(\cdot)$ of $K$ basis functions, $j=1,\ldots, p$.\\
 \textbf{Output:} Intercept $\hat\beta_0$ and functions $\hat f_j(\cdot)$, $j=1,\ldots, p$.
\end{flushleft}
\begin{algorithmic}
\State $r\gets \min(n, p)$
\State $D \gets \diag(d_1^2, \ldots, d_r^2, 0,\ldots, 0)$, with $d_1^2\geq \ldots\geq d_r^2$ eigenvalues of $\mathbf X \mathbf X^T$ 
\State $U\gets$ matrix of eigenvectors of $\mathbf X \mathbf X^T$ (i.e. $\mathbf X \mathbf X^T = U D U^T$)
\State $\hat q \gets \arg \max_{l = 1,\ldots, \lceil r/2 \rceil}\frac{d_l^2}{d_{l+1}^2}$ \Comment{Eigenvalue ratio method}
\State $\hat {\mathbf H} \gets \sqrt{n}(U_1,\ldots, U_{\hat q}) \in \mathbb R^{n\times \hat q}$ \Comment{First $\hat q$ columns of $U$}
\State $B^{(j)}\gets (b_j(x_{1,j}),\ldots, b_j(x_{n,j}))^T\in \mathbb R^{n\times K}$
\State Find $R_j\in \mathbb R^{K\times K}$ such that $R_j^T R_j=\frac{1}{n} (B^{(j)})^T B^{(j)}$ \Comment{Cholesky decomposition}
\State $\tilde B^{(j)} \gets B^{(j)} R_j^{-1}$
\State $(\hat \beta_0, \hat{\tilde\beta}_1,\ldots, \hat{\tilde \beta}_p, \hat\gamma)=\arg\min\{\|\mathbf Y-\beta_0 \mathbf 1_n -\sum_{j=1}^p \tilde B^{(j)}\tilde\beta_j- \hat{\mathbf H} \gamma\|_2^2/n+\lambda\sum_{j=1}^p\|\tilde \beta_j\|_2\}$ \Comment{{Group lasso}}
\State $\hat\beta_j \gets R_j^{-1}\hat{\tilde\beta}_j$, $j=1,\ldots, p$
\State $\hat f_j(\cdot)\gets b_j(\cdot)^T\hat\beta_j$
\end{algorithmic}
\end{algorithm}
}

\Revision{In the following we will refer to our proposed method as the \textit{deconfounded method} and to the ad hoc method explicitly estimating the confounders as the \textit{estimated factors method}. Additionally, we compare the performances to the classical method for fitting high-dimensional additive models. This method is implemented by setting $Q=I_{n}$ in Algorithm \ref{alg_HDAM} and we will refer to it as the \textit{naive method}.

The code to reproduce the analysis and the figures is available on GitHub.\footnote{\url{https://github.com/cyrillsch/Deconfounding_for_HDAM}}}

\subsection{Simulation Results}\label{sec_SimResults}
\Revision{We use the simulation setting below for model \eqref{eq_additive}. We consider two variants of the setup. In the variant \textit{equal confounding influence}, all the components $H_l$, $l=1,\ldots, q$ of the confounder have the same influence on $X$. In the setting \textit{decreasing confounding influence}, the influence of the component $H_l$, $l=1,\ldots, q$ is proportional to $1/l$. In Figure \ref{fig_SVDExample}, we plot the singular values of $\mathbf X$ generated according to the two settings. We expect the deconfounded method using the trim transform to perform equally well in both settings. We expect the estimated factors method to perform well in the setting \textit{equal confounding influence}, as the first $q$ singular values of $\mathbf X$ are clearly separated from the remaining singular values and hence it is easy to consistently estimate the dimension of the confounder $H$. For the setting \textit{decreasing confounding influence}, we do not expect the estimated factors method to perform particularly well as the first $q$ singular values of $\mathbf X$ are not clearly separated from the remaining singular values.} 

\begin{figure}
\centering
\includegraphics[width=0.91\textwidth]{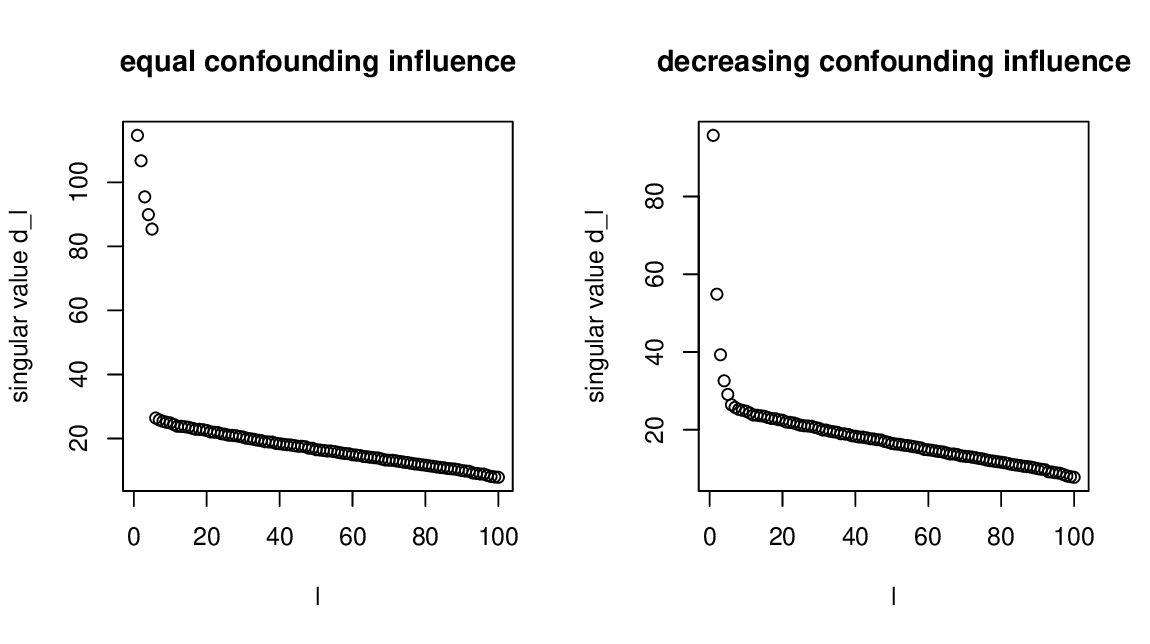}
\caption{Singular values of $\mathbf X$ for data generated according to the settings \textit{equal confounding influence} and \textit{decreasing confounding influence} with $n = 100$, $p = 300$ and $q = 5$.}
\label{fig_SVDExample}
\end{figure}

\begin{description}
	\item[Coefficients:] \Revision{For $l = 1,\ldots, q$, the entries of the $l$th row of $\Psi\in \mathbb R^{q\times p}$ are sampled i.i.d. $\textup{Unif}[-\alpha_l,\alpha_l]$. We use two different settings for $\alpha_l$: the setting \textit{equal confounding influence} has $\alpha_l = 1$, $l = 1, \ldots, q$. The setting \textit{decreasing confounding influence} has $\alpha_l = 1/l$, $l = 1, \ldots, q$.} The entries of $\psi\in\mathbb R^q$ are sampled i.i.d. $\textup{Unif}[0,2]$ \Revision{for both settings.}
	\item[Random variables:] The confounder $H\in \mathbb R^q$ is distributed according to $\mathcal N_q(0, I_{q})$. The unconfounded error $E\in \mathbb R^p$ is distributed according to $\mathcal N_p(0, \Sigma_E)$. The error $e$ is distributed according to $\mathcal N(0, 0.5^2)$.
	\item[Model:] The random vector $X\in \mathbb R^p$ and the random variable $Y\in \mathbb R$ are calculated from model \eqref{eq_additive} with the additive function $f^0(X)=\sum_{j=1}^4 f_j^0(X_j)$ with $f_1^0(x)=-\sin(2x)$, $f_2^0(x)=2-2\tanh(x+0.5)$, $f_3^0(x)=x$ and $f_4^0(x)= 4/(e^x+e^{-x})$.
\end{description}
In the following, we simulate data according to this setup with varying parameters $n$, $p$, and $\Sigma_E$. For each setting, we simulate $n_{rep}=100$ data sets. \Revision{We apply the deconfounded method (Algorithm \ref{alg_HDAM} with $Q=Q^{trim}$) and compare it to the ad hoc method (Algorithm \ref{alg_AdHoc}) and the naive method (Algorithm \ref{alg_HDAM} with $Q=I_{n}$).} We provide violin plots of the mean squared errors $\|\hat f-f^0\|_{L_2}^2$ (with respect to the respective distribution of $X$) and the size of the estimated active set.

\subsubsection{Varying $n$}\label{sec_VarN}
\Revision{In the following, we fix $p=300$ and $q=5$ and vary $n$ between $n=50$ and $n=800$. For each $n$, we simulate $100$ data sets according to the settings \textit{equal confounding influence} and \textit{decreasing confounding influence}. In Figures \ref{fig_VarNIndepEqualCI} and \ref{fig_VarNIndepDecreasingCI}, we see the resulting MSE of $\hat f$ on top and the size of the estimated active set on the bottom for the covariance matrix $\Sigma_E=I_{p}$. 

For the setting \textit{equal confounding influence}, we observe that both the deconfounded method and the estimated factors method clearly outperform the naive method in terms of MSE. Moreover, the estimated factors method seems to perform slightly better in terms of MSE than the deconfounded method. This is no surprise, since in the setting \textit{equal confounding influence}, the data is generated such that the confounders $\mathbf H$ can very well be estimated from $\mathbf X$. Moreover, all the methods seem to overestimate the size of the active set as the true size is only 4, however for the deconfounded method and the estimated factors method, this effect is much less severe.

For the setting \textit{decreasing confounding influence}, however, only the deconfounded method retains the good performance in terms of MSE and variable screening (size of estimated active set), whereas the estimated factors method shows a similar performance as the naive method. The reason is that in this setting, it is much harder for the estimated factors method to obtain a good estimate of the factors and their dimension and hence it is not successful at removing the confounding effect.

In Appendix \ref{sec_ToeplitzCov}, we consider the same simulation scenario but with $E$ having a Toeplitz covariance structure instead of having i.i.d. entries. However, the general picture is the same.}

\begin{figure}
\centering
\includegraphics[width=0.91\textwidth]{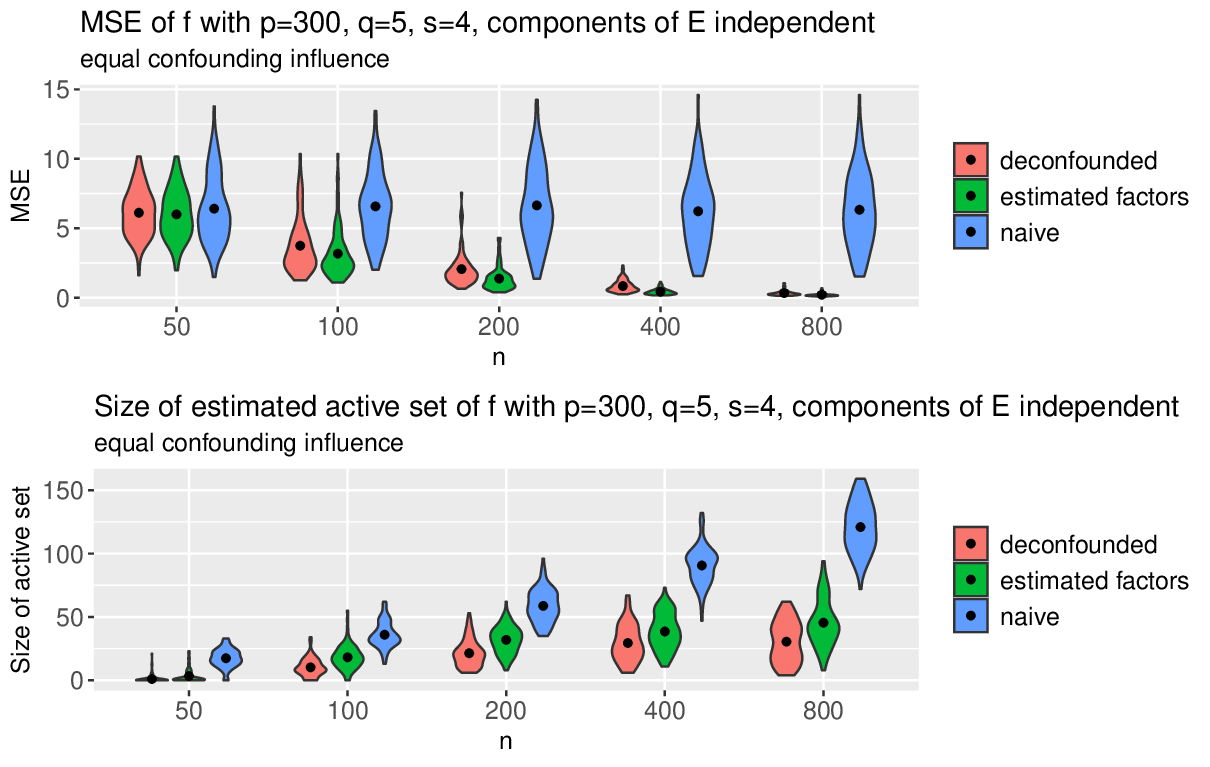}
\caption{MSE (top) and size of estimated active set (bottom) for $\Sigma_E=I_{p}$ and varying $n$ in the setting \textit{equal confounding influence}. 
}
\label{fig_VarNIndepEqualCI}
\end{figure}

\begin{figure}
\centering
\includegraphics[width=0.91\textwidth]{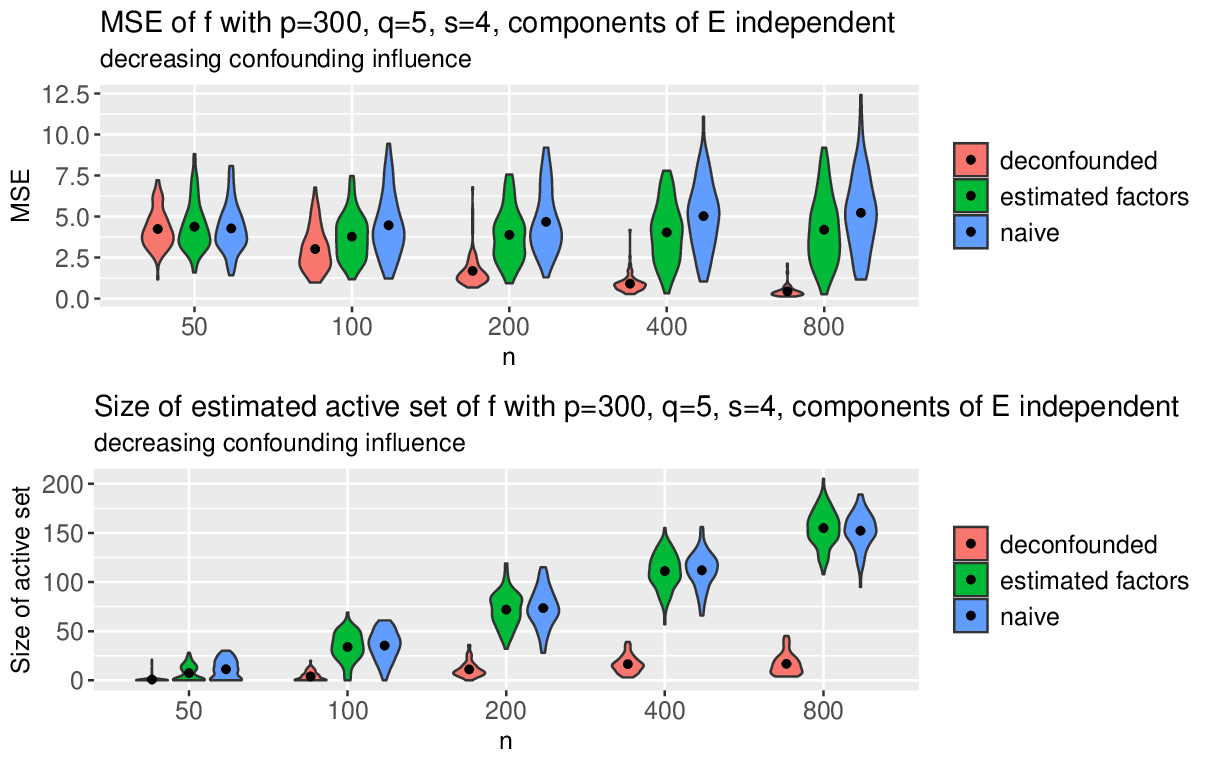}
\caption{MSE (top) and size of estimated active set (bottom) for $\Sigma_E=I_{p}$ and varying $n$ in the setting \textit{decreasing confounding influence}. 
}
\label{fig_VarNIndepDecreasingCI}
\end{figure}

\subsubsection{Varying $p$}\label{sec_VarP}
\Revision{
Here, we fix $n=300$ and $q=5$ and vary $p$ between $p=50$ and $p=800$. For each $p$, we simulate $100$ data sets according to the settings \textit{equal confounding influence} and \textit{decreasing confounding influence}. In Figures \ref{fig_VarPIndepEqualCI} and \ref{fig_VarPIndepDecreasingCI}, we see the resulting MSE of $\hat f$ on top and the size of the estimated active set on the bottom for the covariance matrix $\Sigma_E=I_{p}$. 

The picture is similar to before: in the setting \textit{equal confounding influence}, both the deconfounded method and the estimated factors method perform well in terms of MSE and of the size of the estimated active set, whereas in the setting \textit{decreasing confounding influence}, only the deconfounded method performs significantly better than the naive method. Again, the same simulations with more general covariance structure for $E$ are provided in Appendix \ref{sec_ToeplitzCov}.}
\begin{figure}
\centering
\includegraphics[width=0.91\textwidth]{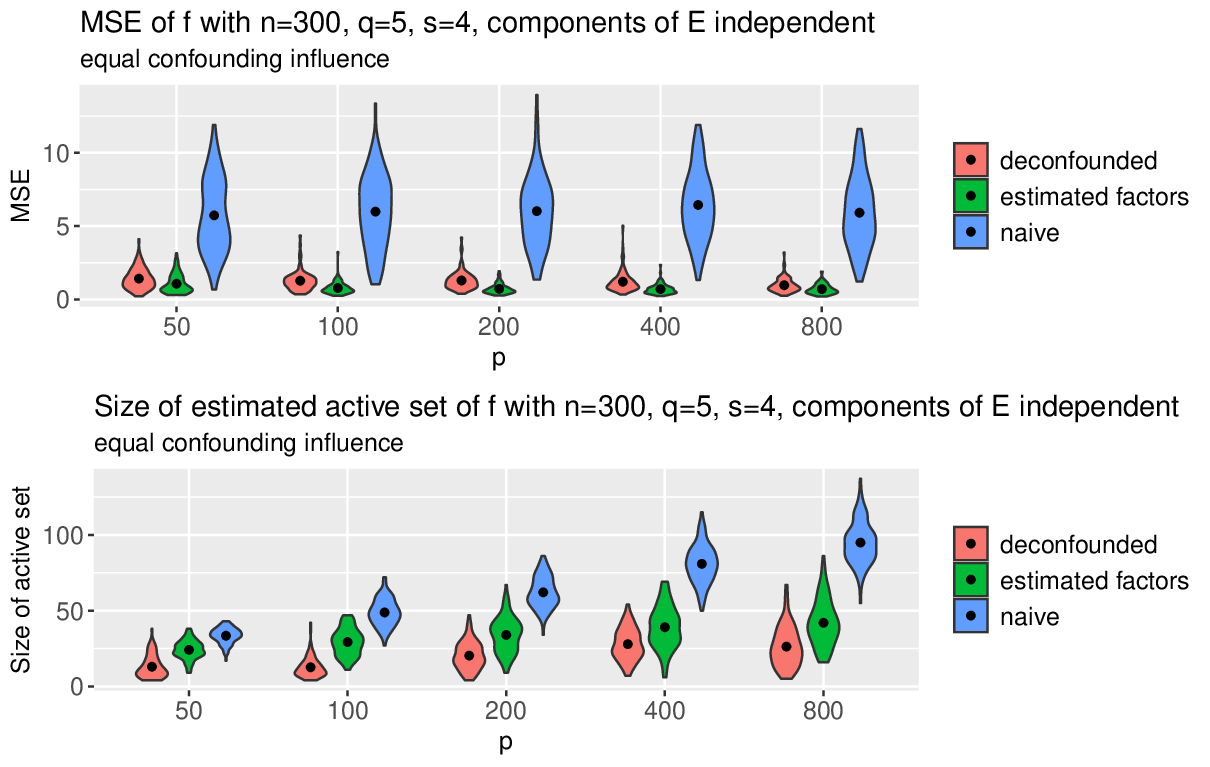}
\caption{MSE (top) and size of estimated active set (bottom) for $\Sigma_E=I_{p}$ and varying $p$ in the setting \textit{equal confounding influence}. 
}
\label{fig_VarPIndepEqualCI}
\end{figure}

\begin{figure}
\centering
\includegraphics[width=0.91\textwidth]{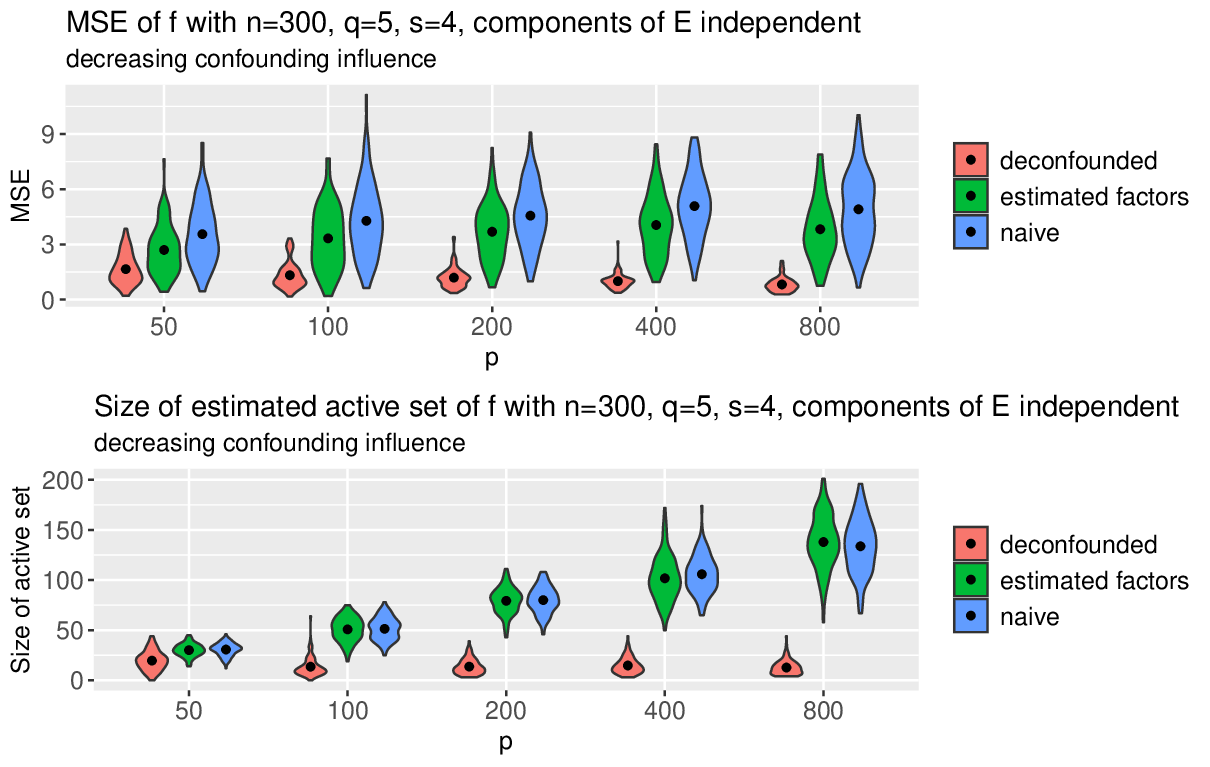}
\caption{MSE (top) and size of estimated active set (bottom) for $\Sigma_E=I_{p}$ and varying $p$ in the setting \textit{decreasing confounding influence}. 
}
\label{fig_VarPIndepDecreasingCI}
\end{figure}

\subsubsection{Varying the Strength of Confounding}
Here, we fix $n=400$, $p=500$, $q=5$ and $\Sigma_E=I_{p}$. We also use the previous settings but vary the strength of confounding on $Y$, i.e. the entries of $\psi\in \mathbb R^q$ are sampled i.i.d. $\textup{Unif}[0, \mathsf{cs}]$ with the confounding strength $\mathsf{cs}$ between $0$ and $3$. For each value of $\mathsf{cs}$, we simulate $100$ data sets. \Revision{In Figures \ref{fig_VarCSEqualCI} and \ref{fig_VarCSDecreasingCI}, we see the resulting MSE of $\hat f$ on top and the size of the estimated active set on the bottom for the settings \textit{equal confounding influence} and \textit{decreasing confounding influence}, respectively. Comparing the deconfounded method to the naive method, we observe that for very small confounding strength ($\mathsf{cs}\leq 0.5$), the deconfounded method performs slightly worse than the naive method in both the \textit{equal confounding influence} and \textit{decreasing confounding influence} settings. This is to be expected since by using the trim transformation we lose a bit of signal. However, as the confounding increases, the deconfounded method is much more robust than the naive method. Comparing the deconfounded method to the estimated factors method, we observe that as before, in the setting \textit{equal confounding influence}, the estimated factors method performs slightly better than the deconfounded method. However, in the setting \textit{decreasing confounding influence}, only the deconfounded method manages to remove the confounding effect, whereas the estimated factors method has comparable performance to the naive method.}

\begin{figure}
\centering
\includegraphics[width=0.91\textwidth]{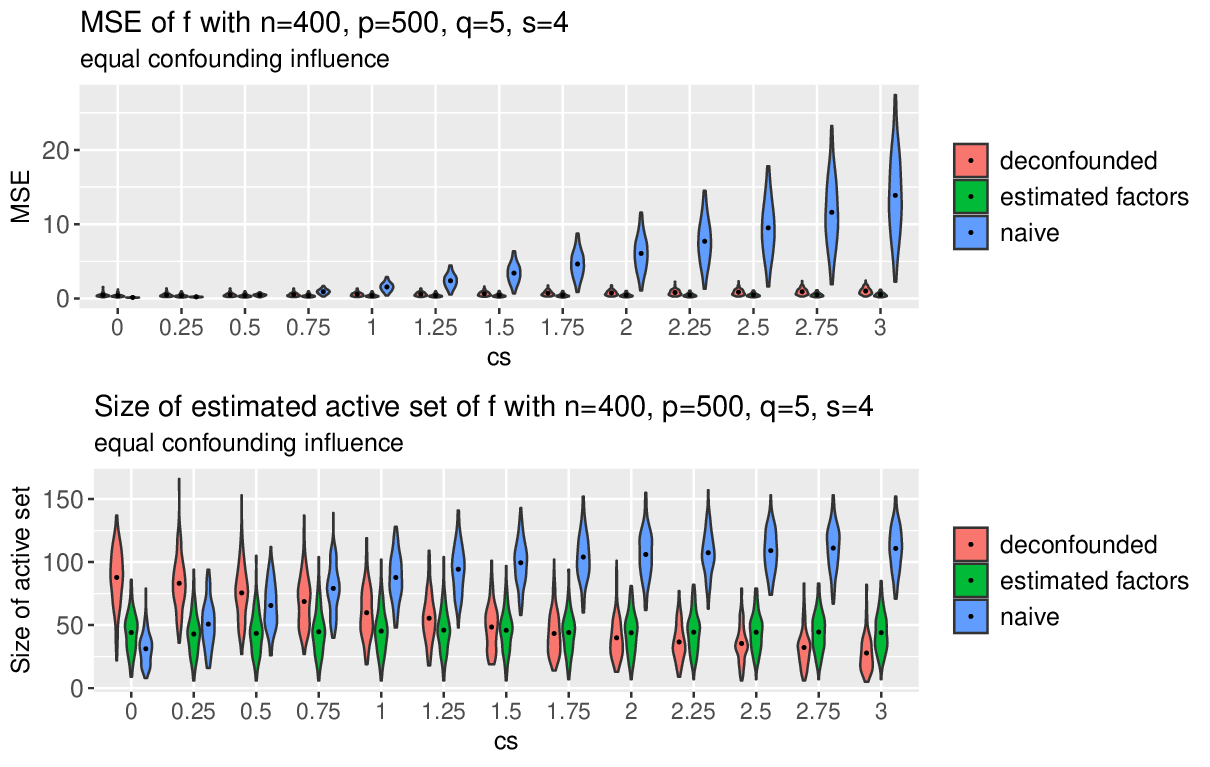}
\caption{MSE (top) and size of estimated active set (bottom) for varying confounding strength in the setting \textit{equal confounding influence}.}
\label{fig_VarCSEqualCI}
\end{figure}

\begin{figure}
\centering
\includegraphics[width=0.91\textwidth]{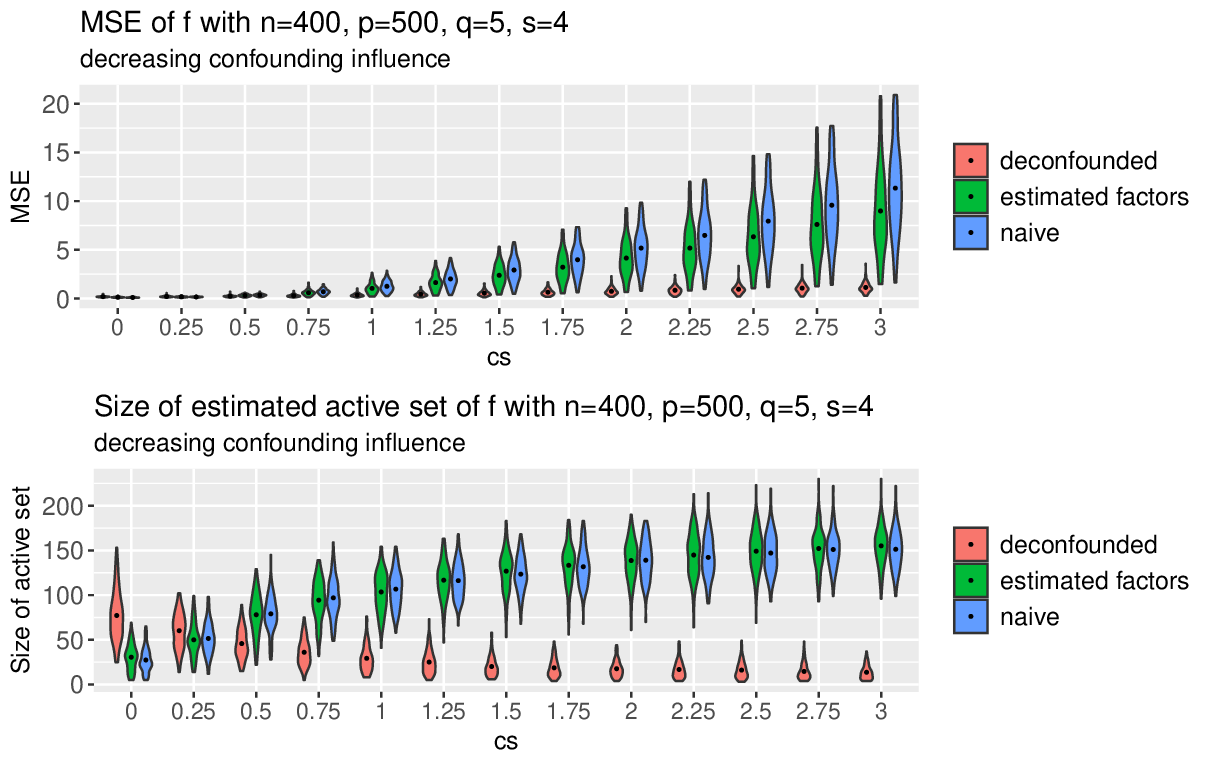}
\caption{MSE (top) and size of estimated active set (bottom) for varying confounding strength in the setting \textit{decreasing confounding influence}.}
\label{fig_VarCSDecreasingCI}
\end{figure}

\subsubsection{Summarizing the Simulation Results}
\Revision{
The simulations indicate that applying spectral deconfounding significantly improves the robustness of high-dimensional additive models both compared to the naive method and also compared to the estimated factors method. It is the only method considered here that shows good results across all the simulation settings considered, both in terms of prediction of $f^0$ and in terms of variable screening. Compared to the naive method, one loses a bit of performance when there is no confounding, but gains a lot if there is. Compared to the ad hoc method of estimated factors, one loses a bit of performance, when $\mathbf X$ has a clear factor structure and there is a clear gap in the spectrum of $\mathbf X \mathbf X^T$. However one gains a lot if the confounding effect is not that clearly separated from the noise.

In Appendix \ref{sec_AdditionalSimulations}, we show additional simulations and also consider misspecified settings.
}

\subsection{Real Data Analysis}\label{sec_RealDataResults}
We apply our method to the motif regression problem. We use a data set that has previously been analyzed by \cite{GuoDLLE}, whose results indicate that a (nonlinear) additive model might be appropriate. The data set originally comes from \cite{BeerPredictingGeneExpression} and has also been reexamined by \cite{YuanPredictingGeneExpression}. We use the same $\mathbf X$ and $\mathbf Y$ as in \cite{GuoDLLE}, that is, the rows of $\mathbf X\in \mathbb R^{2587\times 666}$ are the scores of 666 motifs and the entries of $\mathbf Y\in \mathbb R^{2587}$ are the gene expression values of the corresponding $n=2587$ genes under a particular condition. In Figure \ref{fig_MotifSVD}, we plot the singular values of $\mathbf X$, where we center the columns of $\mathbf X$ to have mean zero. We see that we have one very large spike and several smaller spikes in the singular values. \Revision{This indicates that confounding might be present, but it is not clear from the spectrum, what a good estimate $\hat q$ of the number of factors should be. This suggests that the deconfounded method might be more appropriate than the estimated factors method. We apply the deconfounded method, the estimated factors method, and the naive method on the data set. The fitted function for the deconfounded method has 95 active variables, the fitted function for the estimated factors method has 167 active variables, whereas the fitted function the naive method has 211 active variables. 85 variables are in both the active set of the deconfounded and of the estimated factors method and 92 variables are in both the active set of the deconfounded method and the naive method.}

In Figure \ref{fig_MotifFitted}, we plot the fitted functions $\hat f_j$ for the variables $X_j$ whose effects are the strongest (measured by the norm of the coefficient vector $\|\hat\beta_j\|_2$ of $\hat f_j(\cdot)=b_j(\cdot)^T\hat\beta_j$), when estimated using the deconfounded method. We see that these component functions are very similar for all three methods.
In Figure \ref{fig_MotifFittedDifference}, we plot the fitted functions $\hat f_j$ for the indices $j$ such that the effects of $\hat f_j$ estimated using the naive method are the strongest among the $j$ which are not in the active set estimated using the deconfounded method. \Revision{We see that there exist components $j$ such that the estimated functions $\hat f_j$ are zero for the deconfounded method but distinctively different from zero for both the estimated factors and the naive method. We also observe that still, the fitted functions $\hat f_j$ for the estimated factors and the naive method are very similar.} Finally, Figure \ref{fig_CompareCoefLength} displays the order of importance of the covariates: \Revision{it shows very clearly that very quickly, the top selected covariates from the deconfounded method do not exhibit strong correspondence to the top selected covariates from the estimated factors and the naive method and hence, the difference between the methods cannot be explained by a simple thresholding rule. In particular, we think that the estimated factors method did choose a too low $\hat q$ and hence was not able to remove much of the confounding.\footnote{In fact, from Figure \ref{fig_MotifSVD}, we can see that the eigenvalue ratio method from the estimated factors method chooses $\hat q = 1$, which seems to not remove all the confounding.}} In view of this, we believe that the variable importance and selection with the deconfounded method leads to much better results than the two other methods for this data set with spiked singular values as shown in Figure \ref{fig_MotifSVD}.

\begin{figure}
\centering
\includegraphics[width=0.7\textwidth]{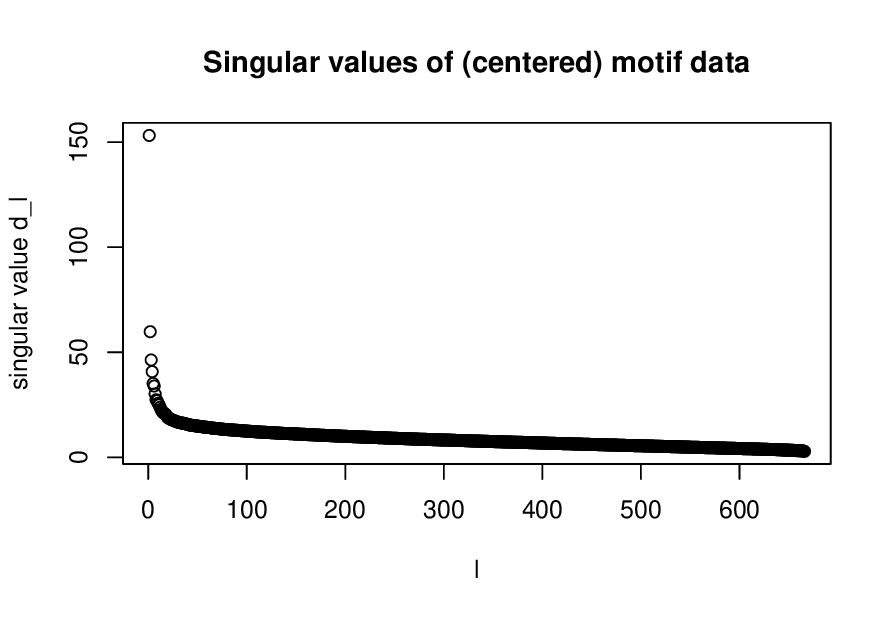}
\caption{Singular values of centered motif scores.}
\label{fig_MotifSVD}
\end{figure}

\begin{figure}
\centering
\includegraphics[width=0.91\textwidth]{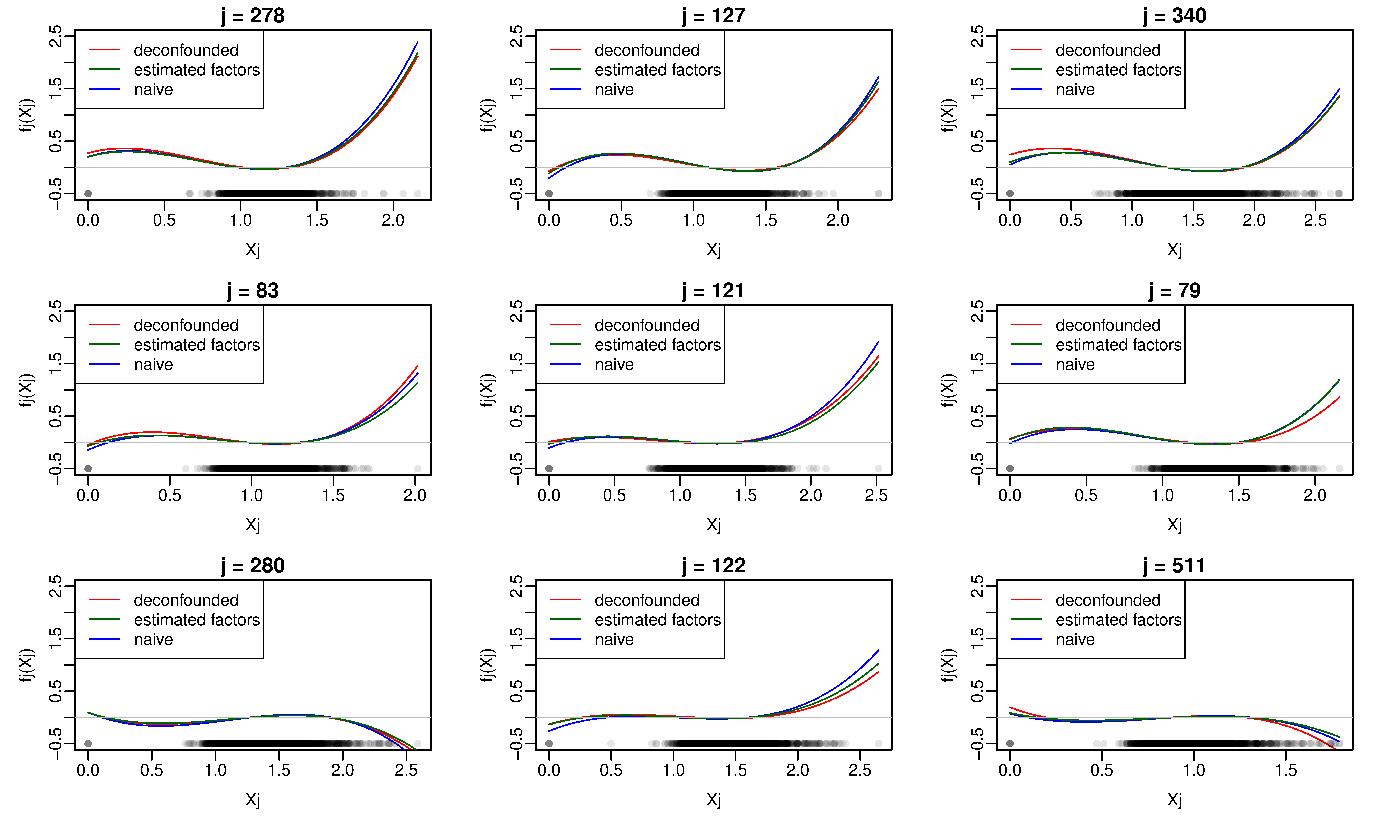}
\caption{Motif data set. Fitted functions $\hat f_j$ for the covariates $X_j$ with strongest effect estimated using the deconfounded method. The grey dots indicate the observed values of $X_j$.}
\label{fig_MotifFitted}
\end{figure}

\begin{figure}
\centering
\includegraphics[width=0.91\textwidth]{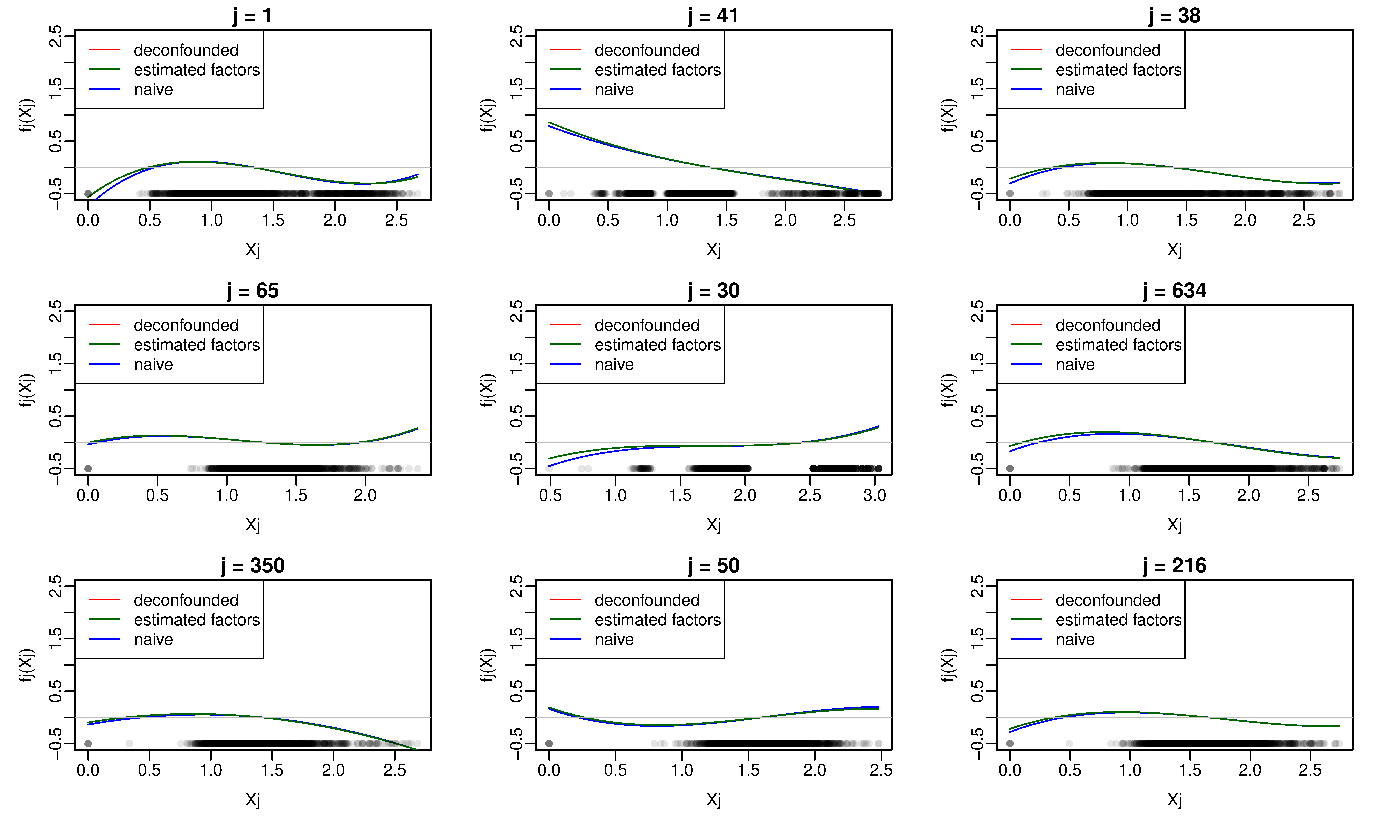}
\caption{Motif data set. Fitted functions $\hat f_j$ for the covariates $X_j$ with zero estimated effect by the deconfounded method but strongest estimated effect by the naive method. The grey dots indicate the observed values of $X_j$. 
}
\label{fig_MotifFittedDifference}
\end{figure}

\begin{figure}
\centering
\includegraphics[width=0.91\textwidth]{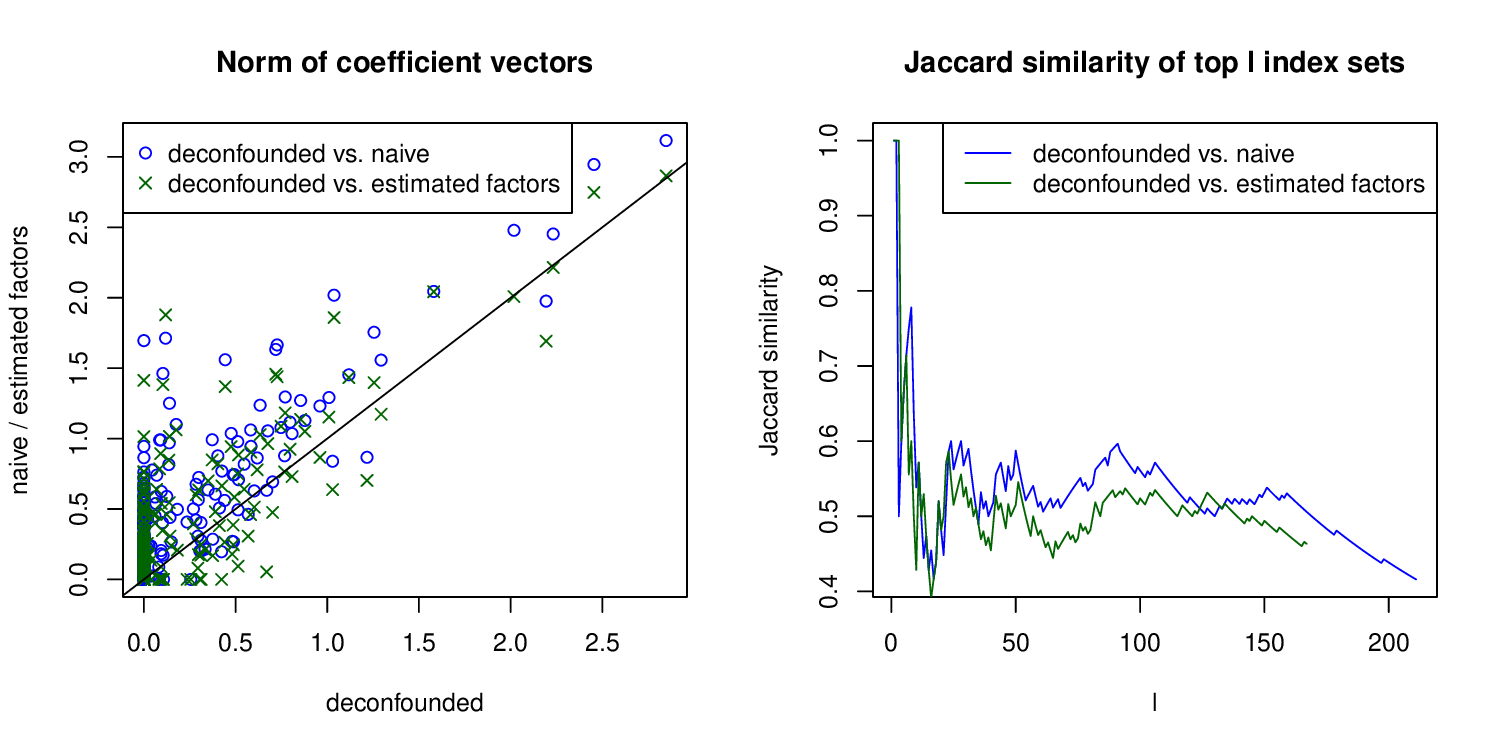}
\caption{Comparison of the strength of the fitted functions (measured by the norm of the coefficient vector $\|\hat\beta_j\|_2$ of $\hat f_j(\cdot)=b_j(\cdot)^T\hat\beta_j$) for the motif data set. Left: Strength of the fitted functions estimated using the deconfounded method vs. strength of the fitted functions estimated using the other methods. Right: Jaccard similarity $\frac{|\mathcal T_{dec}^{(l)}\cap \mathcal T_{naive / e.f.}^{(l)}|}{|\mathcal T_{dec}^{(l)}\cup \mathcal T_{naive / e.f.}^{(l)}|}$ vs. $l$ with $\mathcal T_{dec}^{(l)}\subset\{1,\ldots, p\}$ being the indices of the $l$ strongest fitted functions using the deconfounded method and $\mathcal T_{naive / e.f.}^{(l)}$ the indices of the $l$ strongest fitted functions using the naive method and the estimated factors method, respectively.}
\label{fig_CompareCoefLength}
\end{figure}

\section{Discussion}\label{sec_Discussion}
We developed novel theory and methodology for fitting high-dimensional additive models in presence of hidden confounding. With this, we established that spectral transformations introduced by \cite{CevidSpectralDeconfounding} can also be used in the context of nonlinear regression. Our rigorous theoretical development covers convergence rates as well as detailed justification of high-level assumptions such as the compatibility condition. We demonstrated good empirical performance of our procedure on a wide range of simulation scenarios as well as on real data. In case of no hidden confounding, the method is slightly worse than plain sparse additive model fitting. In presence of hidden confounding though, there is much to be gained. \Revision{Compared to an ad hoc approach of explicitly estimating the confounding dimension and the confounders, our approach is shown to be more robust in situations where the factor structure of $\mathbf X$ is not very clear. The reason is that our method does not depend on finding a clear gap in the spectrum of $\mathbf X$ but instead only needs the median singular value.

While our method is simple and easy to implement using standard group lasso software, the obtained convergence rate may be suboptimal for the high-dimensional additive model under hidden confounding. There might also be more sophisticated algorithms with better properties, for example varying smoothness across components or even adaptive smoothness \cite{TanDoublyPenalizedEstimation, SadhanalaAdditiveModelsTrendFiltering}. Nevertheless, our work indicates that the extension of using spectral transformations with such methods and even }with arbitrary machine learning algorithms could be possible. A general path for such extensions is to replace least squares type objectives $\arg\min_{f\in \mathcal F}\|\mathbf Y-\mathbf f(\mathbf X)\|_2^2/n$, where $\mathcal F$ is some function class, by their deconfounded version $\arg\min_{f\in\mathcal F}\|Q(\mathbf Y-\mathbf f(\mathbf X))\|_2^2/n$ as we did it for the function class of additive models. A rigorous and detailed theoretical understanding will be challenging, but some of our developed results may be useful for such an analysis.

\section*{Acknowledgements}
We are grateful to Cun-Hui Zhang for helpful discussions. Moreover, we want to thank Max Baum for doing preliminary simulations with a slightly different algorithm. Furthermore, we are grateful to Wei Yuan for sharing the pre-processed motif data. We thank Maximilian R\"ucker for kindly pointing out typos in an earlier version of this manuscript. We also thank the anonymous reviewers and the editor for constructive feedback.
CS and PB received funding from the European Research Council (ERC) under the European Union’s Horizon 2020 research and innovation programme (grant agreement No. 786461). The research of ZG was supported in part by the NSF-DMS 2015373 and NIH-R01GM140463 and R01LM013614; ZG also acknowledges financial support for visiting the Institute of Mathematical Research (FIM) at ETH Zurich.

\appendix

\section{Proofs of Theorem \ref{thm_BoundInSample} and Corollary \ref{cor_RateOutSample}}
\subsection{Proof of Theorem \ref{thm_BoundInSample}}\label{sec_ProofBoundInSample}

We first show that the functions $\hat f_j$ are empirically centered. This is an implication of Assumption \ref{ass_BasisFunctions} (partition of unity). For  all $\gamma_1,\ldots, \gamma_p\in \mathbb R$, we have the equality
$$\frac{1}{n}\left\|Q(\mathbf Y-\hat \beta_0\mathbf 1_n-\sum_{j=1}^pB^{(j)}\hat\beta_j)\right\|_2^2=\frac{1}{n}\left\|Q(\mathbf Y-(\hat\beta_0+\sum_{j=1}^p\gamma_j)\mathbf 1_n-\sum_{j=1}^pB^{(j)}(\hat \beta_j-\gamma_j\mathbf 1_K))\right\|_2^2$$
for the first term in the objective \eqref{eq_OptProblem}. Since $\hat\beta$ is the minimizer of \eqref{eq_OptProblem}, we must have that it minimizes the penalty term. Hence, for $j=1,\ldots, p$, we have
$\frac{\textup d}{\textup d\gamma_j}\rvert_{\gamma_j=0}\|B^{(j)}(\hat \beta_j+\gamma_j\mathbf 1_K)\|_2^2=0$. This implies that $\mathbf 1_K^T(B^{(j)})^TB^{(j)}\hat\beta_j=0$. Using again the partition of unity, we have that $B^{(j)}\hat\beta_j=0$. Hence, the estimated functions $\hat f_j$ are empirically centered, i.e. $\sum_{i=1}^n\hat f_j(x_{i,j})=0$.

For $j=1,\ldots p$, consider functions $f_j^c(\cdot)=b_j(\cdot)^T\beta_j^c$ that are empirically centered, that is $\sum_{i=1}^nf_j^c(x_{i,j})=0$. Also, define $\beta_0^c=\beta_0^0$. For $j\notin \mathcal T$, let $f_j^c=0$. In the end, we will set $f_j^c(\cdot)= f_j^\ast(\cdot)-\frac{1}{n}\sum_{i=1}^nf_j^\ast(x_{i,j})$ and $\beta_j^c=\beta_j^\ast-(\frac{1}{n}\sum_{i=1}^n f_j^\ast(x_{i,j}))\mathbf 1_K$.

We now follow the strategy of the proof of Proposition 5 in \cite{GuoDoublyDebiasedLasso}.
By the definition of $\hat \beta$, we have
\begin{align*}
&\frac{1}{n}\left\|Q(\mathbf Y-\hat \beta_0\mathbf 1_n - \sum_{j=1}^p B^{(j)}\hat\beta_j)\right\|_2^2+\frac{\lambda}{\sqrt n}\sum_{j=1}^p\|B^{(j)}\hat\beta_j\|_2\\
&\leq \frac{1}{n}\left\|Q(\mathbf Y-\beta_0^c\mathbf 1_n-\sum_{j=1}^p B^{(j)}\beta_j^c)\right\|_2^2+\frac{\lambda}{\sqrt n}\sum_{j=1}^p\|B^{(j)}\beta_j^c\|_2
\end{align*}
We use decomposition \eqref{eq_ModelWithB} to write
$$\mathbf Y-\hat \beta_0\mathbf 1_n -\sum_{j=1}^p B^{(j)} \hat \beta_j=\mathbf Xb+\bm \epsilon+(\beta_0^c-\hat\beta_0)\mathbf 1_n+\sum_{j=1}^p B^{(j)}(\beta_j^c-\hat \beta_j)+\sum_{j=1}^p(\mathbf f_j^0 -\mathbf f_j^c).$$
It follows that
\begin{align}
&\frac{1}{n}\|Q((\beta_0^c-\hat\beta_0)\mathbf 1_n +\sum_{j=1}^p B^{(j)}(\beta_j^c -\hat \beta_j))\|_2^2+\frac{\lambda}{\sqrt n}\sum_{j=1}^p \|B^{(j)}\hat\beta_j\|_2\nonumber\\
&\leq\frac{\lambda}{\sqrt n}\sum_{j=1}^p \|B^{(j)} \beta_j^c\|_2-\frac{2}{n}\left(\mathbf Xb+\bm{\epsilon} +\sum_{j=1}^p(\mathbf f_j^0-\mathbf f_j^c)\right)^T Q^2\left((\beta_0^c-\hat\beta_0)\mathbf 1_n+\sum_{j=1}^p B^{(j)}(\beta_j^c-\hat\beta_j)\right)\label{eq_BasicInequality}
\end{align}

We use a reparametrization: Let $R_j\in \mathbb R^{K\times K}$ such that $R_j^T R_j= \frac{1}{n} (B^{(j)})^T B^{(j)}$ and define $\tilde \beta_j = R_j\beta_j$ and $\tilde B^{(j)}=B^{(j)} R_j^{-1}$. Then $\tilde B^{(j)}\tilde \beta_j= B^{(j)}\beta_j$ and $(\tilde B^{(j)})^T \tilde B^{(j)}= n I_{K}$. Moreover,
$$\frac{1}{\sqrt n}\sum_{j=1}^p\left \| B^{(j)}\beta_j\right\|_2=\sum_{j=1}^p\|\tilde \beta_j\|_2.$$

Note that using the Cauchy-Schwarz inequality
\begin{align}
\left|\frac{2}{n}\bm\epsilon^T Q^2 \sum_{j=1}^p B^{(j)}(\beta_j^c- \hat\beta_j)\right|&=\left|\frac{2}{n}\bm\epsilon^T Q^2 \sum_{j=1}^p \tilde B^{(j)}(\tilde\beta_j^c- \tilde {\hat\beta}_j)\right|\nonumber\\
&\leq\frac{2}{n}\sum_{j=1}^p \|(\tilde B^{(j)})^TQ^2\bm\epsilon\|_2\|\tilde \beta_j^c- \tilde{\hat \beta}_j\|_2\nonumber\\
&\leq \frac{2}{n}\max_{j=1, \ldots, p} \|(\tilde B^{(j)})^T Q^2\bm\epsilon\|_2\sum_{j=1}^p\|\tilde \beta_j^c- \tilde{\hat \beta}_j\|_2\label{eq_EmpProc1}
\end{align}
and also
\begin{equation}\label{eq_EmpProc2}
\left|\frac{2}{n}\bm \epsilon^T Q^2(\beta_0^c-\hat\beta_0)\mathbf 1_n\right|= \frac{2}{n}|\bm\epsilon^T Q^2\mathbf 1_n||\beta_0^c-\hat \beta_0|
\end{equation}
For some constant $c>0$, let $\lambda_0=\lambda/(1+c)$ and $A_0= A/(1+c)$ and define the event
$$\mathcal A=\left\{\max\left(\frac{2}{n}|\bm\epsilon^T Q^2\mathbf 1_n|,\frac{2}{n}\max_{j=1,\ldots, p}\|(\tilde B^{(j)})^TQ^2\bm\epsilon\|_2\right)\leq \lambda_0\right\}.$$
The goal is to show that $\mathcal A$ has high probability for $n\to\infty$. Recall from decomposition \eqref{eq_ModelWithB} that $\bm \epsilon =\bm e+\Delta$ with $\Delta_i=h_i^T\psi- x_i^T b$ \Revision{and recall the notation $\|e_i\|_{\psi_2|\mathbf X}$ for the sub-Gaussian norm of $e_i$ conditional on $\mathbf X$ defined in Assumption \ref{ass_ConditionsModel0}. Observe that $\|e_i\|_{\psi_2|\mathbf X}=\|e\|_{\psi_2|X = x_i} \leq C_0$. Also note that $\|Q^2 \mathbf1_n\|_2^2\leq\|\mathbf 1_n\|_2^2 = n$, since $\|Q\|_{op}\leq 1$.} By Hoeffding's inequality (see for example Theorem 2.6.3 in \cite{VershyninHDProb}) applied conditionally on $\mathbf X$, there exists $c'>0$ such that
\begin{align}
\Prob\left(\frac{2}{n}|\bm e^TQ^2\mathbf 1_n|> A_0 C_0\sqrt{\frac{K\log p}{n}}|\mathbf X\right)&\leq 2\exp\left(\frac{-c'A_0^2C_0^2 K n\log p}{4\Revision{\max_{i=1,\ldots, n}\left(\|e_i\|_{\psi_2|\mathbf X}^2\right)}\|Q^2\mathbf 1_n\|_2^2}\right)\nonumber\\
&\leq 2\exp\left(\frac{-c'A_0^2 K \log p}{4}\right)\nonumber\\
&\leq 2 p^{-c'A_0^2K/4}.\label{eq_AFirstPart1}
\end{align}
Since the bound does not depend on $\mathbf X$, it also holds for the unconditional probability.
Define $t_n=\frac{1}{2}A_0C_0\sqrt{nK\log p}$. For $t_n^2\geq\|(\tilde B^{(j)})^TQ^2\|_F^2(C_0^2+\sigma_e^2)$, we have by the union bound and Lemma \ref{lem_TailSubGaussian} below (applied conditionally on $\mathbf X$), that there exists $c''>0$ such that
\begin{align*}
\Prob\left(\frac{2}{n}\max_{j=1,\ldots, p} \|(\tilde B^{(j)})^TQ^2\mathbf e\|_2> A_0C_0\sqrt{\frac{K\log p}{n}}|\mathbf X\right)&\leq  \sum_{j=1}^p\Prob\left(\|(\tilde B^{(j)})^TQ^2\mathbf e\|_2> t_n|\mathbf X\right)\\
&\leq\sum_{j=1}^p 2\exp\left[-c''\left(\frac{t_n^2}{C_0^2\|(\tilde B^{(j)})^TQ^2\|_F^2}-\frac{\sigma_e^2}{C_0^2}\right)\right]
\end{align*}
Since the singular values of $Q$ are bounded by $1$, we have by von Neumann's trace inequality \citep{MirskyTrace},
$$\|(\tilde B^{(j)})^T Q^2\|_F^2=\tr((\tilde B^{(j)})^T Q^4 \tilde B^{(j)})= \tr(Q^4 \tilde B^{(j)} (\tilde B^{(j)})^T)\leq \tr(\tilde B^{(j)}(\tilde B^{(j)})^T)= \tr((\tilde B^{(j)})^T\tilde B^{(j)})= nK.$$
Hence for $t_n^2\geq nK(C_0^2+\sigma_e^2)$ and plugging in the definition of $t_n$,
\begin{align}
\Prob\left(\frac{2}{n}\max_{j=1,\ldots, p} \|(\tilde B^{(j)})^TQ^2\mathbf e\|_2> A_0C_0\sqrt{\frac{K\log p}{n}}\right)&\leq  2p\exp\left[-c''\left(\frac{t_n^2}{nKC_0^2}-\frac{\sigma_e^2}{C_0^2}\right)\right]\nonumber\\
&= 2p \exp\left(c''\frac{\sigma_e^2}{C_0^2}\right)\exp\left(\frac{c''}{4}A_0^2 \log p\right)\nonumber\\
&= 2\exp\left(c''\frac{\sigma_e^2}{C_0^2}\right)p^{1-c''A_0^2/4}\label{eq_AFirstPart2}
\end{align}
Hence, we can choose $A_0^2>\max\left(\frac{4}{c''}, 4\frac{C_0^2+\sigma_e^2}{C_0^2\log p}\right)$.

On the other hand, since $(\tilde B^{(j)})^T\tilde B^{(j)} = n I_{K}$ and $\|Q\|_{op}\leq 1$,
$$\|(\tilde B^{(j)})^TQ^2\Delta\|_2\leq  \|(\tilde B^{(j)})^TQ^2\|_{op}\|\Delta\|_2\leq \|\tilde B^{(j)}\|_{op}\|\Delta\|_2\leq \sqrt{n}\|\Delta\|_2.$$
Together with Markov's inequality, we obtain for $t>0$
\begin{align*}
\Prob\left(\frac{2}{ n}\max_{j=1,\ldots, p}\|(\tilde B^{(j)})^T Q^2\Delta\|_2 > t\right)&\leq \Prob\left(\frac{2}{\sqrt {n} }\|\Delta\|_2 > t\right)\\
&\leq \frac{4\E[\|\Delta\|_2^2]}{n t^2}.
\end{align*}
From the definition \eqref{eq_DefB} of $b$ and using Assumption \ref{ass_ConditionsModel0}, we get that
\begin{align*}
\frac{1}{n}\E[\|\Delta\|_2^2]&=\E[\Delta_i^2]\\
&=\E[(\psi^TH-b^TX)^2]\\
&=\E[(\psi^T(H-\Psi\E[XX^T]^{-1} X))^2]\\
&=\psi^T(I_q-\Psi\E[XX^T]^{-1}\Psi^T)\psi
\end{align*}
By Lemma \ref{lem_ApproxB} below, we arrive at
$$\frac{1}{n}\E[\|\Delta\|_2^2]\lesssim  \frac{\|\psi\|_2^2}{1+\lambda_q^2(\Psi)}.$$
Hence by condition \eqref{eq_Lambda} on $\lambda_2$,
$$\Prob\left(\frac{2}{ n}\max_{j=1,\ldots, p}\|(\tilde B^{(j)})^T Q^2\Delta\|_2 > \frac{1}{1+c} \lambda_2\right)\lesssim  \frac{\|\psi\|_2^2}{1+\lambda_q^2(\Psi)}/\lambda_2^2=o(1).$$
Similarly, also
$$\Prob\left(\frac{2}{n}|\Delta^TQ^2\mathbf 1_n|>\frac{1}{1+c}\lambda_2 \right)\lesssim  \frac{\|\psi\|_2^2}{1+\lambda_q^2(\Psi)}/\lambda_2^2=o(1).$$
From this, (\ref{eq_AFirstPart1}) and (\ref{eq_AFirstPart2}), we get that $\Prob(\mathcal A)> 1- o(1)$. In the following, we establish \eqref{eq_RateInSample} on the event $\mathcal A$.
Together with \eqref{eq_EmpProc1} and \eqref{eq_EmpProc2}, we get from \eqref{eq_BasicInequality} that on the event $\mathcal A$, we have
\begin{align*}
&\frac{1}{n}\|Q((\beta_0^c-\hat\beta_0)\mathbf 1_n +\sum_{j=1}^p B^{(j)}(\beta_j^c -\hat \beta_j))\|_2^2+\lambda \sum_{j=1}^p \|\tilde{\hat\beta}_j\|_2\\
&\leq\lambda\sum_{j=1}^p \|\tilde \beta_j^c\|_2+\lambda_0\left(|\beta_0^c-\hat \beta_0|+\sum_{j=1}^p \|\tilde \beta_j^c-\tilde{\hat\beta}_j\|_2\right)+U_n
\end{align*}
With 
$$U_n=\left|\frac{2}{n}\left(\mathbf X b+\sum_{j=1}^p(\mathbf f_j^0-\mathbf f_j^c)\right)^TQ^2 \left((\beta_0^c-\hat\beta_0)\mathbf 1_n +\sum_{j=1}^p B^{(j)}(\beta_j^c -\hat \beta_j)\right)\right|.$$
Recall that $\beta_j^c=0$ for all $j\in \mathcal T^c$. By the triangle inequality,
$$\sum_{j\in \mathcal T}\|\tilde \beta_j^c\|_2-\sum_{j\in \mathcal T}\|\tilde{\hat\beta}_j\|_2\leq \sum_{j\in \mathcal T}\|\tilde \beta_j^c- \tilde{\hat \beta}_j\|_2, \text{ and } \sum_{j\in \mathcal T^c}\|\tilde{\hat\beta}_j\|_2=\sum_{j\in \mathcal T^c}\|\tilde \beta_j^c- \tilde{\hat\beta}_j\|_2.$$
It follows that
\begin{align*}
&\frac{1}{n}\|Q((\beta_0^c-\hat\beta_0)\mathbf 1_n+\sum_{j=1}^p B^{(j)}(\beta_j^c -\hat \beta_j))\|_2^2+(\lambda-\lambda_0)\sum_{j\in \mathcal T^c}\|\tilde{\hat\beta}_j\|_2\\
&\leq \lambda\left(\sum_{j\in \mathcal T}\|\tilde\beta_j^c\|-\sum_{j\in \mathcal T}\|\tilde{\hat\beta}_j\|_2\right)+\lambda_0\left(|\beta_0^c-\hat \beta_0|+\sum_{j\in \mathcal T} \|\tilde \beta_j^c-\tilde{\hat\beta}_j\|_2\right)+U_n\\
&\leq (\lambda+\lambda_0)\left(|\beta_0^c-\hat \beta_0|+\sum_{j\in \mathcal T} \|\tilde \beta_j^c-\tilde{\hat\beta}_j\|_2\right) + U_n.
\end{align*}
We consider two cases:
\begin{description}
	\item[Case 1:] $$(\lambda+\lambda_0)\left(|\beta_0^c-\hat \beta_0|+\sum_{j\in \mathcal T} \|\tilde \beta_j^c-\tilde{\hat\beta}_j\|_2\right) \geq U_n,$$
	\item[Case 2:] \begin{equation}\label{eq_DefCase2}
	(\lambda+\lambda_0)\left(|\beta_0^c-\hat \beta_0|+\sum_{j\in \mathcal T} \|\tilde \beta_j^c-\tilde{\hat\beta}_j\|_2\right) < U_n.
	\end{equation}
\end{description}
In Case 1, we have
\begin{equation}\label{eq_Case1}
\frac{1}{n}\|Q((\beta_0^c-\hat\beta_0)\mathbf 1_n+\sum_{j=1}^p B^{(j)}(\beta_j^c -\hat \beta_j))\|_2^2+(\lambda-\lambda_0)\sum_{j\in \mathcal T^c}\|\tilde{\hat\beta}_j\|_2\leq 2(\lambda+\lambda_0)\left(|\beta_0^c-\hat \beta_0|+\sum_{j\in \mathcal T} \|\tilde \beta_j^c-\tilde{\hat\beta}_j\|_2\right)
\end{equation}
and in particular
\begin{equation}\label{eq_CCFeas}
\sum_{j\in \mathcal T^c}\|\tilde{\hat\beta}_j\|_2\leq \frac{4+2c}{c}\left(|\beta_0^c-\hat \beta_0|+\sum_{j\in \mathcal T} \|\tilde \beta_j^c-\tilde{\hat\beta}_j\|_2\right).
\end{equation}
By the definition of $\tilde{\hat \beta}_j$, it follows that $\|\tilde{\hat \beta}_j\|_2=\frac{1}{\sqrt n}\|B^{(j)}\beta_j\|_2=\frac{1}{\sqrt n}\|\hat {\mathbf f}_j\|_2$ and similarly $\|\tilde\beta_j^c-\tilde{\hat\beta}_j\|_2=\frac{1}{\sqrt n}\|\mathbf f_j^c-\hat{\mathbf f}_j\|_2$. Hence, we can rewrite \eqref{eq_CCFeas} as
\begin{equation}\label{eq_CCFeas1}
\sum_{j\in \mathcal T^c}\frac{1}{\sqrt n}\|\mathbf f_j^c -\hat{\mathbf f}_j\|_2\leq \frac{4+2c}{c}\left(|\beta_0^c-\hat\beta_0|+\sum_{j\in \mathcal T}\frac{1}{\sqrt n}\|\mathbf f_j^c -\hat{\mathbf f}_j\|_2\right).
\end{equation}
This means that for $M=(4+2c)/c$, the function $f^c- \hat f=(\beta_0^c-\hat\beta_0)+\sum_{j=1}^p(f_j^c-\hat f_j)$ lies in the set $\mathcal F_{M, \mathcal T}^n$ defined in \eqref{eq_DefFMTn} (recall from the beginning of the proof that $\hat f_j$ and $f_j^c$ are empirically centered for all $j=1,\ldots, p$). By the definition \eqref{eq_DefCC} of the compatibility constant $\tau_n$, we have that
$$\frac{1}{n}\|Q((\beta_j^c-\hat\beta_j)\mathbf 1_n+\sum_{j=1}^p B^{(j)}(\beta_j^c- \hat \beta_j))\|_2^2= \frac{1}{n}\|Q(\mathbf f^c -\hat{\mathbf f})\|_2^2 \geq \tau_n \left((\beta_0^c-\hat\beta_0)^2+\sum_{j=1}^p \frac{1}{n}\|\mathbf f_j^c - \hat{\mathbf f}_j\|_2^2\right).$$
Together with the Cauchy-Schwarz inequality and \eqref{eq_Case1}, we have
\begin{align*}
\tau_n\left(|\beta_0^c-\hat\beta_0|+\sum_{j\in \mathcal T}\frac{1}{\sqrt n}\|\mathbf f_j^c-\hat{\mathbf f}_j\|_2\right)^2 &\leq (s+1)\tau_n \left((\beta_0^c-\hat\beta_0)^2+\sum_{j=1}^p\frac{1}{ n}\|\mathbf f_j^c-\hat{\mathbf f}_j\|_2^2\right)\\
&\leq (s+1)\frac{1}{n}\|Q((\beta_j^c-\hat\beta_j)\mathbf 1_n+\sum_{j=1}^p B^{(j)}(\beta_j^c- \hat \beta_j))\|_2^2\\
&\leq 2(s+1)(\lambda+\lambda_0)\left(|\beta_0^c-\hat \beta_0|+\sum_{j\in \mathcal T} \|\tilde \beta_j^c-\tilde{\hat\beta}_j\|_2\right)\\
&=2(s+1)(\lambda+\lambda_0)\left(|\beta_0^c-\hat \beta_0|+\sum_{j\in \mathcal T} \frac{1}{\sqrt n}\|\mathbf f_j^c-\hat{\mathbf f}_j\|_2\right)
\end{align*}
and hence,
$$|\beta_0^c-\hat\beta_0|+\sum_{j\in \mathcal T}\frac{1}{\sqrt n}\|\mathbf f_j^c-\hat{\mathbf f}_j\|_2\leq \frac{2(s+1)(\lambda+\lambda_0)}{\tau_n}.$$
Together with \eqref{eq_CCFeas1}, we arrive at
\begin{equation}\label{eq_Case1Result}
|\beta_0^c-\hat\beta_0|+\sum_{j=1}^p \frac{1}{\sqrt n}\|\mathbf f_j^c-\hat{\mathbf f}_j\|_2\lesssim \frac{s\lambda}{\tau_n}.
\end{equation}

In Case 2, we have
\begin{equation}\label{eq_Case2}
\frac{1}{n}\|Q((\beta_0^c-\hat\beta_0)\mathbf 1_n+\sum_{j=1}^p B^{(j)}(\beta_j^c -\hat \beta_j))\|_2^2+(\lambda-\lambda_0)\sum_{j\in \mathcal T^c}\|\tilde{\hat\beta}_j\|_2\leq 2 U_n.
\end{equation}
By the Cauchy-Schwarz inequality,
\begin{equation}\label{eq_Case2A}
2 U_n\leq \frac{4}{n}\|Q\mathbf Xb+Q\sum_{j=1}^p(\mathbf f_j^0-\mathbf f_j^c)\|_2\|Q((\beta_0^c-\hat\beta_0)\mathbf 1_n+\sum_{j=1}^pB^{(j)}(\beta_j^c-\hat\beta_j))\|_2
\end{equation}
In particular, it follows from \eqref{eq_Case2} and \eqref{eq_Case2A} that
$$\frac{1}{\sqrt n}\|Q((\beta_0^c-\hat\beta_0)\mathbf 1_n+\sum_{j=1}^pB^{(j)}(\beta_j^c-\hat\beta_j))\|_2\leq\frac{4}{\sqrt n}\|Q\mathbf Xb+Q\sum_{j=1}^p(\mathbf f_j^0-\mathbf f_j^c)\|_2.$$
Plugging this back into (\ref{eq_Case2A}), yields
$$2U_n\leq \frac{16}{n} \|Q\mathbf Xb+Q\sum_{j=1}^p(\mathbf f_j^0-\mathbf f_j^c)\|_2^2.$$
From (\ref{eq_DefCase2}) and (\ref{eq_Case2}), we have
$$|\beta_0^c-\hat\beta_0| + \sum_{j=1}^p\|\tilde \beta_j^c- \tilde{\hat \beta}_j\|_2\leq \frac{U_n}{ (\lambda+\lambda_0)}+\frac{2 U_n}{(\lambda-\lambda_0)}.$$

Hence, using again that $\|\tilde \beta_j^c- \tilde{\hat \beta}_j\|_2=\frac{1}{\sqrt n}\|\mathbf f_j^c - \hat{\mathbf f}_j\|_2$,
\begin{equation}\label{eq_Case2Result}
|\beta_0^c - \hat\beta_0| + \sum_{j=1}^p \frac{1}{\sqrt n}\|\mathbf f_j^c - \hat{\mathbf f}_j\|_2\lesssim\frac{1}{\lambda}\left(\frac{1}{n}\|Q\mathbf Xb\|_2^2+\frac{1}{n}\|Q\sum_{j=1}^p (\mathbf f_j^0- \mathbf f_j^c)\|_2^2\right)\
\end{equation}

Since either Case 1 or Case 2 holds, \eqref{eq_Case2Result} and \eqref{eq_Case1Result} together imply that on the event $\mathcal A$,
\begin{equation}\label{eq_RateFc}
|\beta_0^c-\hat\beta_0|+\sum_{j=1}^p\frac{1}{\sqrt n}\|\mathbf f_j^c-\hat{\mathbf f}_j\|_2\lesssim \frac{s \lambda}{\tau_n}+\frac{1}{\lambda}\frac{\|Q\mathbf Xb\|_2^2}{n}+\frac{1}{\lambda}\frac{\|Q\sum_{j=1}^p(\mathbf f_j^0-\mathbf f_j^c)\|_2^2}{n}.
\end{equation}

We now return to the beginning and define $f_j^c(\cdot)=f_j^\ast(\cdot)-\frac{1}{n}\sum_{i=1}^nf_j^\ast(x_{i,j})$. By Assumption \ref{ass_BasisFunctions} (partition of unity), we have $f_j^c(\cdot)=b_j(\cdot)^T\beta_j^c$ with $\beta_j^c =\beta_j^\ast-(\frac{1}{n}\sum_{i=1}^nf_j^\ast(x_{i,j}))\mathbf 1_k$.
Note that
\begin{align}
\frac{1}{\sqrt n}\|\mathbf f_j^c-\mathbf f_j^\ast\|_2&=\frac{1}{\sqrt n}\|\mathbf 1_n\sum_{i=1}^n\frac{1}{n} f_j^\ast(x_{i,j})\|_2\nonumber\\
&=|\frac{1}{n}\sum_{i=1}^n f_j^\ast(x_{i,j})|\nonumber\\
&\leq|\frac{1}{n}\sum_{i=1}^n (f_j^\ast(x_{i,j})-f_j^0(x_{i,j}))|+|\frac{1}{n}\sum_{i=1}^nf_j^0(x_{i,j})|.\label{eq_DecompCStar}
\end{align}
By the Cauchy-Schwarz inequality, 
\begin{equation}\label{eq_DecompCStar1}
|\frac{1}{n}\sum_{i=1}^n(f_j^\ast(x_{i,j})-f_j^0(x_{i,j}))|\leq \sqrt{\frac{1}{n}\sum_{i=1}^n (f_j^\ast(x_{i,j})-f_j^0(x_{i,j}))^2}=\frac{1}{\sqrt n}\|\mathbf f_j^\ast-\mathbf f_j^0\|_2.
\end{equation}
By the triangle inequality,
\begin{equation}\label{eq_DiffStarHat}
|\beta_0^0-\hat\beta_0|+\sum_{j=1}^p\frac{1}{\sqrt n}\|\mathbf f_j^\ast-\hat{\mathbf f}_j\|_2\leq |\beta_0^c-\hat\beta_0|+\sum_{j=1}^p\frac{1}{\sqrt n}\|\mathbf f_j^c-\hat{\mathbf f}_j\|_2+\sum_{j\in \mathcal T}\frac{1}{\sqrt n}\|\mathbf f_j^\ast-\mathbf f_j^c\|_2.
\end{equation}
Since $\|Q\|_{op}=1$,
$$\frac{1}{n}\|Q\sum_{j=1}^p(\mathbf f_j^0-\mathbf f_j^c)\|_2^2\leq \frac{1}{n}\|\sum_{j=1}^p(\mathbf f_j^0-\mathbf f_j^c)\|_2^2\leq \frac{1}{n}\left(\sum_{j\in \mathcal T}\|\mathbf f_j^0-\mathbf f_j^\ast\|_2+\sum_{j\in \mathcal T}\|\mathbf f_j^\ast-\mathbf f_j^c\|_2\right)^2.$$
Together with \eqref{eq_RateFc}, \eqref{eq_DecompCStar}, \eqref{eq_DecompCStar1} and \eqref{eq_DiffStarHat}, we obtain that on the event $\mathcal A$,
\begin{align*}
|\beta_0^0-\hat\beta_0|+\sum_{j=1}^p\frac{1}{\sqrt n}\|\mathbf f_j^\ast-\hat{\mathbf f}_j\|_2 &\lesssim\frac{s K\lambda}{\tau_n}+\frac{1}{\lambda}\frac{\|Q\mathbf Xb\|_2^2}{n}\\
&+\sum_{j\in \mathcal T}\frac{1}{\sqrt n}\|\mathbf f_j^\ast-\mathbf f_j^0\|_2+\sum_{j\in \mathcal T}|\frac{1}{n}\sum_{i=1}^n f_j^0(x_{i,j})|\\
&+\frac{1}{\lambda}\left(\sum_{j\in \mathcal T}\frac{1}{\sqrt n}\|\mathbf f_j^\ast-\mathbf f_j^0\|_2+\sum_{j\in \mathcal T}|\frac{1}{n}\sum_{i=1}^n f_j^0(x_{i,j})|\right)^2
\end{align*}
This concludes the proof.

\subsection{Some Lemmas}

\begin{lemma}\label{lem_TailSubGaussian}
Let the random vector $\mathbf e\in \mathbb R^n$ have \Revision{independent entries with variance $\E[e_i^2]\leq \sigma_e^2, \, i = 1,\ldots, n$ and sub-Gaussian norm $\|e_i\|_{\psi_2}\leq C_0,\, i = 1,\ldots, n$ with $\sigma_e^2$ and $C_0$ independent of $i$.} Let $A\in \mathbb R^{k\times n}$ be a matrix. Then for any $t^2\geq\|A\|_F^2(C_0^2+\sigma_e^2)$, we have
$$\Prob\left(\|A\mathbf e\|_2\geq t\right)\leq 2\exp\left[-c\left(\frac{t^2}{C_0^2\|A\|_F^2}-\frac{\sigma_e^2}{C_0^2}\right)\right].$$
\begin{proof}
We first observe that \Revision{$\E[\mathbf e^T A^TA\mathbf e]=\tr(A^T A\E[\mathbf e\mathbf e^T])\leq\sigma_e^2\|A\|_F^2$.}
Using the Hanson-Wright inequality (see for example \cite{RudelsonHansonWright}), we have
\begin{align*}
\Prob(\|A\mathbf e\|_2\geq t )&=\Prob(\mathbf e^T A^TA\mathbf e\geq t^2)\\
&=\Revision{\Prob(\mathbf e^T A^T A\mathbf e- \E[\mathbf e^TA^TA\mathbf e]\geq t^2-\E[\mathbf e^TA^TA\mathbf e])}\\
&\leq\Prob(\mathbf e^T A^T A\mathbf e- \E[\mathbf e^TA^TA\mathbf e]\geq t^2-\sigma_e^2\|A\|_F^2)\\
&\leq 2\exp\left[-c \min\left(\frac{(t^2-\sigma_e^2\|A\|_F^2)^2}{C_0^4\|A^TA\|_F^2},\frac{t^2-\sigma_e^2\|A\|_F^2}{C_0^2\|A^TA\|_{op}}\right)\right]
\end{align*}
Since, $\|A^TA\|_{op}\leq\|A\|_F^2$ and $\|A^TA\|_F^2\leq\|A\|_F^4$, we obtain
$$\Prob\left(\|A\mathbf e\|_2\geq t\right)\leq 2\exp\left[-c\min\left(\left(\frac{t^2}{C_0^2\|A\|_F^2}-\frac{\sigma_e^2}{C_0^2}\right)^2,\frac{t^2}{C_0^2\|A\|_F^2}-\frac{\sigma_e^2}{C_0^2}\right)\right].$$
Since $t^2\geq\|A\|_F^2(C_0^2+\sigma_e^2)$, we have $\frac{t^2}{C_0^2\|A\|_F^2}-\frac{\sigma_e^2}{C_0^2}\geq 1$, which gives the result.
\end{proof}
\end{lemma}

The following result is a slight variant of Lemma 2 in \cite{GuoDoublyDebiasedLasso}.
\begin{lemma}\label{lem_ApproxB}
Under Assumption \ref{ass_ConditionsModel0}, assertion (1) and Assumption \ref{ass_ConditionsModel1}, we have that
$$|\psi^T(I_q-\Psi\E[XX^T]^{-1}\Psi^T)\psi | \lesssim \frac{\|\psi\|_2^2}{1+\lambda_q^2(\Psi)}.$$
\end{lemma}
\begin{proof}
By Assumption \ref{ass_ConditionsModel0}, assertion (1) and the Woodbury identity \cite{GolubMatrixComputations}, we have that
\begin{align*}
|\psi^T(I_q-\Psi\E[XX^T]^{-1}\Psi^T)\psi |&=|\psi^T(I_q-\Psi(\Psi^T\Psi+\Sigma_E)^{-1}\Psi^T)\psi |\\
&=\psi^T(I_q+\Psi\Sigma_E^{-1}\Psi^T)^{-1}\psi\\
&\leq\|\psi\|_2^2/\lambda_{\min}(I_q+\Psi\Sigma_E^{-1}\Psi^T).
\end{align*}
Moreover,
\begin{align*}
\lambda_{\min}(I_q+\Psi\Sigma_E^{-1}\Psi^T)&=\inf_{y\neq 0}\frac{\|y\|_2^2+y^T\Psi\Sigma_E^{-1}\Psi^Ty}{\|y\|_2^2}\\
&\geq 1+\inf_{z\neq 0}\frac{z^T\Sigma_E^{-1} z}{\|z\|_2^2}\inf_{y\neq 0}\frac{y^T\Psi\Psi^T y}{\|y\|_2^2}\\
&=1+\lambda_{\min}(\Sigma_E^{-1})\lambda_	q (\Psi)^2.
\end{align*}
Using $\lambda_{\min}(\Sigma_E^{-1})\geq c$ (Assumption \ref{ass_ConditionsModel1}) yields the result.
\end{proof}
\subsection{Proof of Corollary \ref{cor_RateOutSample}}\label{sec_ProofCorOutSample}
For a function $f_j(\cdot)=b_j(\cdot)^T\beta_j$, we have on one hand
$$\frac{1}{n}\|\mathbf f_j\|_2^2=\frac{1}{n}\beta_j^T(B^{(j)})^TB^{(j)}\beta_j\geq \lambda_{\min}\left(\frac{1}{n}(B^{(j)})^TB^{(j)}\right)\|\beta_j\|_2^2$$
and on the other hand
$$\|f_j\|_{L_2}^2=\E[f_j(X_j)^2]=\beta_j^T\E[b_j(X_j)b_j(X_j)^T]\beta_j\leq \|\beta_j\|_2^2\lambda_{\max}\left(\E[b_j(X_j)b_j(X_j)^T]\right).$$
It follows that on the event $\mathcal B$ from Assumption \ref{ass_BasisOutSample} with $\Prob(\mathcal B)=1-o(1)$ we have $\|f_j\|_{L_2}\leq \sqrt C \frac{1}{\sqrt n}\|\mathbf f_j\|_2$. Since this holds for all $j=1,\ldots, p$ and independently of $\beta_j$, we establish \eqref{eq_RateCorOutSample} on the event $\mathcal{B}\cap\mathcal A$ for the event $\mathcal A$ with $\Prob(\mathcal A) = 1-o(1)$ from the proof of Theorem \ref{thm_BoundInSample}.

\section{Proofs for Section \ref{sec_Compatibility}}
\subsection{Proof of Theorem \ref{thm_BoundCC}}\label{sec_ProofBoundCC}
\begin{remark}
    For this proof, the Gaussianity assumption (Assumption \ref{ass_Gaussian}) can be relaxed to sub-Gaussian (with additional restrictions at some places to be able to apply Lemma 7 in \cite{GuoDoublyDebiasedLasso}).
\end{remark}

We first define a second spectral transformation $Q^{\text{PCA}}$ similar to $Q^{\text{trim}}$. Instead of shrinking the top half of the singular values of $\mathbf X$ to the median singular value, $Q^\text{PCA}$ shrinks the first $q$ singular values of $\mathbf X$ to $0$ and leaves the others as they are. More formally, as in Section \ref{sec_Method}, let $\mathbf X\mathbf X^T= UDU^T$ be the eigenvalue decomposition of $\mathbf X\mathbf X^T$. Let $\bar d_l=1\{l>q\}$ and $Q^\text{PCA}= U\diag(\bar d_1, \ldots, \bar d_r, 1, \dots 1) U^T$. Note that $q$ is not known in practice. However, we only use $Q^\text{PCA}$ as a theoretical construct. For $\mathbf X \mathbf X^T=UDU^T$, define $\hat{\mathbf H} = \sqrt n U_{1:q}$ to be the scaled first $q$ columns of $U$. $\hat{\mathbf H}$ is the solution of the following least squares problem, see for example \cite{FanAreLatent}:
$$(\hat{\mathbf H}, \hat \Psi)=\arg\min_{\mathbf H_0\in \mathbb R^{n\times q}, \Psi_0\in \mathbb R^{q\times p}}\|\mathbf X-\mathbf H_0\Psi_0\|_F^2 \text{ subject to }\frac{1}{n}\mathbf{H}_0^T\mathbf H_0= I_q\text{ and } \Psi_0\Psi_0^T \text{ is diagonal}.$$
Observe that $Q^\text{PCA}=I_n-U_{1:q}U_{1:q}^T=I_n-\frac{1}{n}\hat{\mathbf H}\hat{\mathbf H}^T$. Since $\frac{1}{n}\hat{\mathbf H}^T\hat{ \mathbf H}=I_	q$, we have that $Q^\text{PCA}$ is the projection on the orthogonal complement of the space spanned by the columns of $\hat{\mathbf H}$.
Up to rotation, $\hat{\mathbf H}$ is an approximation of $\mathbf H$.

\begin{lemma}\label{lem_HatH}
Under the assumptions of Theorem \ref{thm_BoundCC}, there exists a matrix $O\in \mathbb R^{q\times q}$ such that
\begin{enumerate}
	\item $\frac{1}{\sqrt n}\|\mathbf H O- \hat{\mathbf H}\|_{op}= o_P\left(\frac{1}{s}\right)$,
	\item $\|I_q-O O^T\|_{op}=o_P\left(\frac{1}{s}\right)$.
\end{enumerate}
\end{lemma}
The proof of Lemma \ref{lem_HatH} is presented in Section \ref{sec_ProofWeakFactorModel}.
Define $\tau_n^\text{PCA}$ according to \eqref{eq_DefCC} but with $Q = Q^\text{PCA}$ instead of $Q = Q^\text{trim}$. We first show that 
\begin{equation}\label{eq_CompTrimPCA}
\tau_n\gtrsim \tau_n^\text{PCA}
\end{equation}
with high probability. For this, recall the definition of $Q^\text{trim}= U\diag(\tilde d_1, \ldots, \tilde d_r, 1, \ldots, 1) U^T$ with $\tilde d_l=\min(d_{\lfloor\rho r\rfloor}/d_l, 1)$ for some $\rho \in (0,1)$. Hence, if $q< \lfloor\rho r\rfloor$,
$$\inf_{z\in \mathbb R^n}\frac{\|Q^\text{trim}z\|_2^2}{\|Q^\text{PCA}z\|_2^2}=\inf_{z\in \mathbb R^n}\frac{\sum_{l=1}^r \tilde d_l^2 z_l^2+\sum_{l=r+1}^nz_l^2}{\sum_{l=1}^r \bar d_l^2 z_l^2+\sum_{l=r+1}^nz_l^2}=\min_{l=1,\ldots, r}\frac{\tilde d_l^2}{\bar d_l^2}= \frac{d_{\lfloor\rho r\rfloor}^2}{d_{q+1}^2}.$$
It follows that for $q<\lfloor \rho r\rfloor$,  $\tau_n\geq \frac{d_{\lfloor\rho r\rfloor}^2}{d_{q+1}^2} \tau_n^\text{PCA}$. By Proposition 3 in \cite{GuoDoublyDebiasedLasso}, we have that with high probability $d_{q+1}^2\lesssim \max(n, p).$ \Revision{By Assumption \ref{ass_Gaussian}, $X$ is a Gaussian random vector and hence, the random vector $\E[XX^T]^{-1} X$ has independent entries. We can apply Lemma 7 from \cite{GuoDoublyDebiasedLasso} to obtain that with high probability $d_{\lfloor\rho r\rfloor}^2\gtrsim \max(n, p)$.\footnote{If one wants to weaken the Gaussianity assumption as written in Remark \ref{rmk_GaussWeak}, one needs additional assumptions to apply Lemma 7 from \cite{GuoDoublyDebiasedLasso}, in particular $p/n\to c^\ast \in [0,\infty)$.}} Hence, on an event $\mathcal C$ with $\Prob(\mathcal C) = 1-o(1)$, we have that $\tau_n\gtrsim \tau_n^\text{PCA}$. It remains to prove
\begin{equation}\label{eq_ControlTauPCA}
    \tau_n^\text{PCA}\gtrsim \tau_0
\end{equation}
with high probability.

\subsubsection{Proof of (\ref{eq_ControlTauPCA})}
For ease of notation, we omit the $\inf_{\mathcal T\subset[p],|\mathcal T|\leq s}$ in the following, but work with a fixed $\mathcal T$. One can just replace all $\inf_{f\in\mathcal F_{M,\mathcal T}^n}$ by $\inf_{\mathcal T\subset[p],|\mathcal T|\leq s}\inf_{f\in\mathcal F_{M,\mathcal T}^n}$ (and similarly for the supremum) to obtain the full result.
 Recall that $Q^\text{PCA}=I_n-\frac{1}{n}\hat{\mathbf H}\hat{\mathbf H}^T$. Hence,
\begin{align}
&\inf_{f^w\in \mathcal F^n_{M, \mathcal T}}\frac{\frac{1}{n}\|Q^\text{PCA}\mathbf f^w\|_2^2}{w_0^2+\sum_{j=1}^p\frac{1}{n}\|\mathbf f^w_j\|_2^2}\nonumber=\inf_{f^w\in \mathcal F^n_{M, \mathcal T}}\frac{\frac{1}{n}\|\mathbf f^w-\frac{1}{n}\hat {\mathbf H}\hat{\mathbf H}^T \mathbf f^w\|_2^2}{w_0^2+\sum_{j=1}^p\frac{1}{n}\|\mathbf f^w_j\|_2^2}\nonumber\\
&\geq \inf_{f^w\in \mathcal F^n_{M, \mathcal T}}\frac{\frac{1}{n}\|\mathbf f^w-\frac{1}{n}\mathbf H \mathbf H^T \mathbf f^w\|_2^2}{w_0^2+\sum_{j=1}^p\frac{1}{n}\|\mathbf f^w_j\|_2^2}-\sup_{f^w\in \mathcal F^n_{M, \mathcal T}}\frac{\left|\frac{1}{n}\|\mathbf f^w-\frac{1}{n}\mathbf H \mathbf H^T \mathbf f^w\|_2^2-\frac{1}{n}\|\mathbf f^w-\frac{1}{n}\hat{\mathbf H} \hat{\mathbf H}^T \mathbf f^w\|_2^2\right|}{w_0^2+\sum_{j=1}^p\frac{1}{n}\|\mathbf f^w_j\|_2^2}\label{eq_SampleCCBound}
\end{align}
We first prove the following lemma.
\begin{lemma}\label{lem_TechnicalCC}
$$\sup_{f^w\in \mathcal F^n_{M, \mathcal T}}\frac{\left|\frac{1}{n}\|\mathbf f^w-\frac{1}{n}\mathbf H \mathbf H^T \mathbf f^w\|_2^2-\frac{1}{n}\|\mathbf f^w-\frac{1}{n}\hat{\mathbf H} \hat{\mathbf H}^T \mathbf f^w\|_2^2\right|}{w_0^2+\sum_{j=1}^p\frac{1}{n}\|\mathbf f^w_j\|_2^2}=o_P(1).$$
\end{lemma}
\begin{proof}
Since $\frac{1}{n}\hat{\mathbf H}^T\hat{\mathbf H}= I_q$, we have with $\mathbf f=\mathbf f^w$,
\begin{align}
&\left|\frac{1}{n}\|\mathbf f-\frac{1}{n}\mathbf H \mathbf H^T \mathbf f\|_2^2-\frac{1}{n}\|\mathbf f-\frac{1}{n}\hat{\mathbf H} \hat{\mathbf H}^T \mathbf f\|_2^2\right|=\left|-\frac{2}{n^2}\mathbf f^T\mathbf H \mathbf H^T\mathbf f + \frac{1}{n^2} \mathbf f^T \mathbf H\left(\frac{1}{n}\mathbf H^T \mathbf H\right) \mathbf H^T \mathbf f+\frac{1}{n^2}\mathbf f^T\hat{\mathbf H} \hat{\mathbf H}^T \mathbf f\right|\nonumber\\
&\leq\left|\frac{1}{n^2}\mathbf f^T \mathbf H\left(\frac{1}{n}\mathbf H^T \mathbf H- I_q\right)\mathbf H^T\mathbf f\right|+\left|\frac{1}{n}\mathbf f^T\left(\frac{1}{n}\hat{\mathbf H}\hat{\mathbf H}^T -\frac{1}{n} \mathbf H \mathbf H^T\right)\mathbf f\right|\nonumber\\
&\leq \frac{1}{n}\|\mathbf f\|_2^2\frac{1}{ n}\| \mathbf H\|_{op}^2\|\frac{1}{n}\mathbf H^T \mathbf H-I_q\|_{op}+\frac{1}{n}\|\mathbf f\|_2^2\|\frac{1}{n} \hat {\mathbf H}\hat {\mathbf H}^T-\frac{1}{n} \mathbf{H} \mathbf H^T\|_{op}.\label{eq_BoundDiffHHat}
\end{align}
Observe that
\begin{align*}
&\|\frac{1}{n}\hat{\mathbf H}\hat{\mathbf H}^T-\frac{1}{n} \mathbf H \mathbf H^T\|_{op}\leq\frac{1}{n}\|\hat{\mathbf H}\hat{\mathbf H}^T-\mathbf H O O^T \mathbf H^ T\|_{op}+\frac{1}{n}\| \mathbf H O O^T \mathbf H^T- \mathbf H \mathbf H^T\|_{op}\\
&\leq\frac{1}{n}\|\hat{\mathbf H}\|_{op}\|\hat{\mathbf H}- \mathbf HO\|_{op}+\frac{1}{n}\|\mathbf H\|_{op}\|O\|_{op}\|\hat {\mathbf H}-\mathbf H O\|_{op}+\frac{1}{n}\|\mathbf H\|_{op}^2\|OO^T-I_q\|_{op}.
\end{align*}
Since the rows of $\mathbf H$ are i.i.d. sub-Gaussian isotropic random vectors in $\mathbb R^q$, we have $\frac{1}{n}\|\mathbf H\|_{op}^2=\lambda_{\text{max}}(\frac{1}{n} \mathbf H^T \mathbf H)=O_P(1)$ and $\|\frac{1}{n} \mathbf H^T \mathbf H- I_q\|_{op} = O_P(\frac{\sqrt q}{\sqrt n})=o_P(\frac{1}{s})$, see for example Theorem 4.6.1 in \cite{VershyninHDProb}. Moreover, we have $\|\frac{1}{\sqrt n}\hat{\mathbf H}\|_{op}^2=\lambda_{\max}(\frac{1}{n} \hat{\mathbf H}^T\hat{\mathbf H})=1$ and $\|O\|_{op}=1+o_P(1)$ by Lemma \ref{lem_HatH}. Hence, we obtain from  \eqref{eq_BoundDiffHHat} and Lemma \ref{lem_HatH}
$$\sup_{f^w\in \mathcal F^n_{M, \mathcal T}}\frac{\left|\frac{1}{n}\|\mathbf f^w-\frac{1}{n}\mathbf H \mathbf H^T \mathbf f^w\|_2^2-\frac{1}{n}\|\mathbf f^w-\frac{1}{n}\hat{\mathbf H} \hat{\mathbf H}^T \mathbf f^w\|_2^2\right|}{w_0^2+\sum_{j=1}^p\frac{1}{n}\|\mathbf f^w_j\|_2^2}\leq \frac{\frac{1}{n}\|\mathbf f^w\|_2^2}{w_0^2+\sum_{j=1}^p\frac{1}{n}\|\mathbf f_j^w\|_2^2}o_P(\frac{1}{s}).$$
For $f^w\in \mathcal F^n_{M, \mathcal T}$, we apply the triangle inequality and the Cauchy-Schwarz inequality to get
\begin{align*}
\frac{1}{n}\|\mathbf f^w\|_2^2&\leq\left(|w_0|+\sum_{j=1}^p \frac{1}{\sqrt n}\|\mathbf f_j^w\|_2\right)^2\\
&\leq (1+M)^2\left(|w_0|+\sum_{j\in\mathcal T} \frac{1}{\sqrt n}\|\mathbf f_j^w\|_2\right)^2\\
&\leq (s+1)(1+M)^2\left(w_0^2+\sum_{j\in \mathcal T}\frac{1}{n}\|\mathbf f^w_j\|_2^2\right)\\
&\leq (s+1)(1+M)^2\left(w_0^2+\sum_{j=1}^p\frac{1}{n}\|\mathbf f^w_j\|_2^2\right),
\end{align*}
which gives the result.
\end{proof}
We now reduce the first term \eqref{eq_SampleCCBound} to its population version $\tau_0$. Note that the functions in $\mathcal F_{M,\mathcal T}^n$ are empirically centered, whereas the functions in $\mathcal F_{\textup{add}}$ are centered with respect to the expectation. Hence, we need additional centering. We use the following Lemma.
\begin{lemma}\label{lem_TechnicalCC2}
$$\inf_{f^w\in \mathcal F_{M,\mathcal T}^n}\inf _{j=1, \ldots, p}\frac{\E\left[\left(f_j^w(X_j)-\E[f_j^w(X_j)]\right)^2\right]}{\frac{1}{n}\|\mathbf f_j^w\|_2^2}= 1+o_P(1)$$
\end{lemma}
\begin{proof}
Define $\hat \Sigma_j=\frac{1}{n}(B^{(j)})^TB^{(j)}$ and $\Sigma_j=\E[b_j(X)b_j(X)^T]$ and observe
\begin{align}
\left|\inf_{f^w\in \mathcal F^n_{M, \mathcal T}} \inf_{j=1, \ldots, p}\frac{\E[f_j^w(X_j)^2]}{\frac{1}{n}\|\mathbf f_j^w\|_2^2}-1\right|&\leq \sup_{f^w\in \mathcal F^n_{M, \mathcal T}} \sup_{j=1, \ldots, p}\frac{|\E[f_j^w(X_j)^2]-\frac{1}{n}\|\mathbf f_j^w\|_2^2|}{\frac{1}{n}\|\mathbf f_j^w\|_2^2}\nonumber\\
&\leq \sup_{f^w\in \mathcal F^n_{M, \mathcal T}} \sup_{j=1, \ldots, p} \frac{\|w_j\|_1^2\|\hat\Sigma_j-\Sigma_j\|_\infty}{\|w_j\|_2^2\lambda_{\min}(\hat \Sigma_j)}\nonumber\\
&\leq \frac{K}{\min_{j=1, \ldots, p}\lambda_{\min}(\hat\Sigma_j)}\left\|\frac{1}{n}\mathbf B^T\mathbf B-\E[\mathbf b(X)\mathbf b(X)^T]\right\|_\infty\nonumber
\end{align}
with $\mathbf b(X)^T=(b_1(X_1)^T, \ldots, b_p(X_P))^T)^T\in \mathbb R^{Kp}$ and the matrix $\mathbf B\in \mathbb R^{n\times Kp}$ having rows $\mathbf b(x_i)\in\mathbb R^{Kp}$. By Assumption \ref{ass_CondBasis}, assertion (1), we can apply Problem 14.3 in \cite{BuehlmannHDStats} and obtain 
$$\left\|\frac{1}{n}\mathbf B^T\mathbf B-\E[\mathbf b(X)\mathbf b(X)^T]\right\|_\infty=O_P\left(\sqrt{\frac{\log(Kp)}{n}}\right).$$
By Assumption \ref{ass_DimAndPsi}, assertion (1), it follows that
\begin{equation}\label{eq_LemTechnicalCC2Part1}
\inf_{f^w\in \mathcal F^n_{M, \mathcal T}} \inf_{j=1, \ldots, p}\frac{\E[f_j^w(X_j)^2]}{\frac{1}{n}\|\mathbf f_j^w\|_2^2}=1+o_P(1).
\end{equation}
Since $f^w\in \mathcal F_{M, \mathcal T}^n$ is empirically centered, we have
\begin{align}
\E[f_j^w(X_j)]^2&=\left(\E[f_j^w(X_j)]-\frac{1}{n}\sum_{i=1}^n f_j^w(x_{i,j})\right)^2\nonumber \\
&=\left((\E[b_j(X_j)]-\frac{1}{n}\sum_{i=1}^n b_j(x_{i,j}))^Tw_j\right)^2\nonumber \\
&\leq \|w_j\|_1^2\|\E[b_j(X_j)]-\frac{1}{n}\sum_{i=1}^n b_j(x_{i,j}\|_\infty^2\nonumber \\
&\leq  K \|w_j\|_2^2\|\E[\mathbf b(X)]-\frac{1}{n}\sum_{i=1}^n \mathbf b(x_{i,\cdot})\|_\infty^2.\label{eq_ExpFjw}
\end{align}
 Using Lemma 14.16 in \cite{BuehlmannHDStats}, it follows that $\|\E[\mathbf b(X)]-\frac{1}{n}\sum_{i=1}^n \mathbf b(x_{i,\cdot})\|_\infty=O_P(\sqrt{\frac{\log (Kp)}{n}})$ and hence, we obtain that
$$\sup_{f^w\in \mathcal F_{M, \mathcal T}^n}\sup_{j=1, \ldots, p} \frac{\E[f_j^w(X_j)]^2}{\frac{1}{n}\|\mathbf f_j^w\|_2^2}\leq \frac{K}{\min_{j=1, \ldots, p}\lambda_{\min}(\hat \Sigma_j)} O_P\left(\frac{\log (Kp)}{n}\right)= o_P(1)$$
by Assumption \ref{ass_CondBasis}, assertion (2). Together with \eqref{eq_LemTechnicalCC2Part1}, it follows that 
\begin{align*}
    \inf_{f^w\in \mathcal F_{M,\mathcal T}^n}\inf _{j=1, \ldots, p}\frac{\E\left[\left(f_j^w(X_j)-\E[f_j^w(X_j)]\right)^2\right]}{\frac{1}{n}\|\mathbf f_j^w\|_2^2}&=\inf_{f^w\in \mathcal F_{M,\mathcal T}^n}\inf _{j=1, \ldots, p}\frac{\E[f_j^w(X_j)^2]-\E[f_j(X_j)]^2}{\frac{1}{n}\|\mathbf f_j^w\|_2^2}\\
    &= 1+o_P(1),
\end{align*}
which concludes the proof.
\end{proof}
We continue with \eqref{eq_SampleCCBound}. Define $a_w=\mathbf H^T \mathbf f^w\in \mathbb R^q$.
For every $f^w\in \mathcal F_{M,\mathcal T}^n$ and $a\in \mathbb R^q$, we can write
\begin{equation}\label{eq_FactorizeCC}
\frac{\frac{1}{n}\|\mathbf f^w-\mathbf H a\|_2^2}{w_0^2+\sum_{j=1}^p\frac{1}{n}\|\mathbf f_j^w\|_2^2}=A_{f^w, a}\cdot B_{f^w, a}\cdot C_{f^w}
\end{equation}
with
\begin{align*}
A_{f^w, a}&=\frac{\frac{1}{n}\|\mathbf f^w-\mathbf H a\|_2^2}{\E\left[(f^w(X)-H^Ta-\sum_{j=1}^p\E[f_j(X_j)])^2\right]}\\
B_{f^w, a}&=\frac{\E\left[(f^w(X)-H^Ta-\sum_{j=1}^p\E[f_j(X_j)])^2\right]}{w_0^2+\sum_{j=1}^p\E\left[(f_j^w(X_j)-\E[f_j^w(X_j)])^2\right]}\\
C_{f^w}&=\frac{w_0^2+\sum_{j=1}^p\E\left[(f_j^w(X_j)-\E[f_j^w(X_j)])^2\right]}{w_0^2+\sum_{j=1}^p\frac{1}{n}\|\mathbf f_j^w\|_2^2}
\end{align*}
From Lemma \ref{lem_TechnicalCC2}, it follows that
\begin{equation}\label{eq_ControlCfw}
\inf_{f^w\in \mathcal F_{M,\mathcal T}^n}C_{f_w}\leq\min\left(1,\inf_{f^w\in \mathcal F_{M,\mathcal T}^n}\inf _{j=1, \ldots, p}\frac{\E\left[\left(f_j^w(X_j)-\E[f_j^w(X_j)]\right)^2\right]}{\frac{1}{n}\|\mathbf f_j^w\|_2^2} \right)= 1+o_P(1).
\end{equation}
For $B_{f_w, a}$, note that $f^w-\sum_{j=1}^p\E[f_j(X_j)]=w_0+\sum_{j=1}^pf_j^w(X_j)-\E[f_j^w(X_j)]\in \mathcal F_{\textup{add}}$. Hence, 
\begin{equation}\label{eq_ControlBfw}
\inf_{f^w\in \mathcal F_{M, \mathcal T}^n, a\in \mathbb R^q} B_{f^w, a}\geq \tau_0.
\end{equation}
For $A_{f^w, a}$, note that for all $f^w\in \mathcal F_{M, \mathcal T}^n$ and $a\in \mathbb R^q$, 
\begin{equation}\label{eq_FactorizeAfw}
|A_{f^w, a}-1|=\frac{\left|\frac{1}{n}\|\mathbf f^w-\mathbf H a\|_2^2-\E\left[(f^w(X)-H^Ta-\sum_{j=1}^p\E[f_j^w(X_j)])^2\right]\right|}{\E\left[(f^w(X)-H^Ta-\sum_{j=1}^p\E[f_j^w(X_j)])^2\right]}=D_{f^w, a}\cdot E_{f^w}\cdot F_{f^w, a}
\end{equation}
with
\begin{align*}
D_{f^w, a}&=\frac{\left|\frac{1}{n}\|\mathbf f^w-\mathbf H a\|_2^2-\E\left[(f^w(X)-H^Ta-\sum_{j=1}^p\E[f_j^w(X_j)])^2\right]\right|}{w_0^2+\sum_{j=1}^p\frac{1}{n}\|\mathbf f_j^w\|_2^2}\\
E_{f^w}&=\frac{w_0^2+\sum_{j=1}^p\frac{1}{n}\|\mathbf f_j\|_2^2}{w_0^2+\sum_{j=1}^p \E\left[(f_j^w(X_j)-\E[f_j^w(X_j)])^2\right]}\\
F_{f^w, a}&=\frac{w_0^2+\sum_{j=1}^p \E\left[(f_j^w(X_j)-\E[f_j^w(X_j)])^2\right]}{\E\left[(f^w(X)-H^Ta-\sum_{j=1}^p\E[f_j^w(X_j)])^2\right]}
\end{align*}
From Lemma \ref{lem_TechnicalCC2}, it follows that 
\begin{equation}\label{eq_ControlEfw}
\sup_{f^w\in \mathcal F_{M,\mathcal T}^n}E_{f^w}=1+o_P(1).
\end{equation}
Moreover, from the definition of $\tau_0$, we have that
\begin{equation}\label{eq_ControlFfw}
\sup_{f^w\in \mathcal F_{M,\mathcal T}^n, a\in \mathbb R^q}F_{f^w, a}\leq \frac{1}{\tau_0}.
\end{equation}
For $D_{f^w, a}$, we can write
\begin{equation}\label{eq_SumDfw}
D_{f^w, a}\leq D_{f^w, a}'+ D_{f^w}''
\end{equation}
\begin{align*}
D_{f^w, a}'&=\frac{\left|\frac{1}{n}\|\mathbf f^w-\mathbf H a\|_2^2-\E\left[(f^w(X)-H^Ta)^2\right]\right|}{w_0^2+\sum_{j=1}^p\frac{1}{n}\|\mathbf f_j^w\|_2^2}\\
D_{f^w}''&=\frac{\left|2\E\left[\sum_{j=1}^pf_j^w(X_j)\right]^2+2 w_0\E[\sum_{j=1}^p f_j^w(X_j)]\right|}{w_0^2+\sum_{j=1}^p\frac{1}{n}\|\mathbf f_j^w\|_2^2}
\end{align*}
Define the matrix $\bar{\mathbf B}\in \mathbb R^{n\times (pK+q+1)}$ with rows $\bar{\mathbf B}_{i,\cdot}\coloneqq \bar{\mathbf b}(x_{i,\cdot}, h_{i,\cdot})^T\coloneqq (1, b_1(x_{i,1})^T, \ldots, b_p(x_{i, p})^T, h_{i,\cdot}^T)^T$ and define the vector $\bar{\mathbf w}=(w_0, w_1^T, \ldots, w_p^T, -a^T)^T\in \mathbb R^{Kp+q+1}$. Observe that $f^w(X)-H^Ta=\bar{\mathbf b}(X, H)^T\bar{\mathbf w}$ and hence
\begin{align}
\left|\frac{1}{n}\|\mathbf f^w-\mathbf H a\|_2^2-\E\left[(f^w(X)-H^Ta)^2\right]\right|&=\left|\bar{\mathbf w}^T(\frac{1}{n}\bar{\mathbf B}^T\bar{\mathbf B}-\E[\bar{\mathbf b}(X,H)\bar{\mathbf b}(X,H)^T])\bar{\mathbf w}\right|\nonumber\\
&\leq\|\bar{\mathbf w}\|_1^2\|\frac{1}{n}\bar{\mathbf B}^T\bar{\mathbf B}-\E[\bar{\mathbf b}(X,H)\bar{\mathbf b}(X,H)^T]\|_\infty\label{eq_BoundDfw1}
\end{align}
By Problem 14.3 in \cite{BuehlmannHDStats}, $\|\frac{1}{n}\bar{\mathbf B}^T\bar{\mathbf B}-\E[\bar{\mathbf b}(X,H)\bar{\mathbf b}(X,H)^T]\|_\infty=O_P\left(\sqrt{\frac{\log(Kp)}{n}}\right)$. Using that $\frac{1}{n}\|\mathbf f_j^w\|_2^2\geq\|w_j\|_2^2\lambda_{\min}(\hat\Sigma_j)$, the definition of $f^w\in \mathcal F_{M,\mathcal T}^n$ and the Cauchy-Schwarz inequality, we have
\begin{align}
\|\bar{\mathbf w}\|_1^2 &=\left(|w_0|+\sum_{j=1}^p\|w_j\|_1+\|a\|_1\right)^2 \nonumber\\
&\leq 3 |w_0|^2+3\|a\|_1^2+3\left(\sum_{j=1}^p  \|w_j\|_1\right)^2 \nonumber\\
&\leq 3 |w_0|^2+3q\|a\|_2^2+3K\left(\sum_{j=1}^p  \|w_j\|_2\right)^2 \nonumber\\
&\leq 3 |w_0|^2+3q\|a\|_2^2+\frac{3K}{\min_{j=1, \ldots, p}\lambda_{\min}(\hat\Sigma_j)}\left(\sum_{j=1}^p  \frac{1}{\sqrt n} \|\mathbf f_j^w\|_2\right)^2 \nonumber\\
&\Revision{\leq 3 |w_0|^2+3q\|a\|_2^2+\frac{3K(1+s)(1+M)^2}{\min_{j=1, \ldots, p}\lambda_{\min}(\hat\Sigma_j)}\left(w_0^2 + \sum_{j=1}^p  \frac{1}{ n} \|\mathbf f_j^w\|_2^2\right)} \label{eq_BoundWBarL1}
\end{align}
With \eqref{eq_BoundDfw1}, we have that for all $f^w\in \mathcal F_{M,\mathcal T}^n$ and $a\in \mathbb R^q$, we have
\begin{equation}\label{eq_DefLfw}
\Revision{D_{f^w, a}'\leq L_{f^w, a}U_n \text{ with } L_{f^w, a}=\frac{3 |w_0|^2+3q\|a\|_2^2+\frac{3K(1+s)(1+M)^2}{\min_{j=1, \ldots, p}\lambda_{\min}(\hat\Sigma_j)}\left(w_0^2 + \sum_{j=1}^p  \frac{1}{ n} \|\mathbf f_j^w\|_2^2\right)}{w_0^2+\sum_{j=1}^p \frac{1}{n}\|\mathbf f^w_j\|_2^2}}
\end{equation} and $U_n=O_P\left(\sqrt{{\log(Kp)}/{n}}\right)$ independent of $\mathbf f^w$ and $a$.
Note that from the definition prior to \eqref{eq_FactorizeCC}, we only need to control $\sup_{f^w\in \mathcal F_{M,\mathcal T}^n} L_{f^w, a_w}$ with $a_w=\mathbf H^T \mathbf f^w$ and not $\sup_{f^w\in \mathcal F_{M,\mathcal T}^n, a\in \mathbb R^q} L_{f^w, a}$.  Using arguments as before, $\|\frac{1}{\sqrt n}\mathbf f^w\|_2^2\leq (1+M)^2 (s+1)\left(w_0^2+\sum_{j=1}^p\|\mathbf f_j^w\|_2^2\right)$. Hence, $\|a_w\|_2^2\leq\|\frac{1}{\sqrt n}\mathbf H\|_{op}^2(s+1)(M+1)^2\left(w_0^2+\sum_{j=1}^p\|\mathbf f_j^w\|_2^2\right)$. It follows that \Revision{$L_{f^w, a^w}\leq 3 + 3(M+1)^2(s+1)\left(q\|\frac{1}{\sqrt n}\mathbf H\|_{op}^2 + \frac{K}{\min_{j=1, \ldots, p}\lambda_{\min}(\hat\Sigma_j)}\right)$}. By Theorem 4.6.1 in \cite{VershyninHDProb}, we have $\|\frac{1}{\sqrt n}\mathbf H\|_{op}^2=O_P(1)$. From Assumption \ref{ass_DimAndPsi}, assertion (1), and Assumption \ref{ass_CondBasis}, assertion (2), it follows that
\begin{equation}\label{eq_ControlDfw1}
\sup_{f^w\in \mathcal F_{M,\mathcal T}^n}L_{f^w, a_w}U_n=o_P(1).
\end{equation}

From \eqref{eq_ExpFjw}, we have that $|\E[f_j^w(X_j)]|\leq \|w_j\|_1 \|\E[\mathbf b(X)]-\frac{1}{n}\sum_{i=1}^n\mathbf b(x_{i,\cdot})\|_\infty$. Similarly as before, $\|\E[\mathbf b(X)]-\frac{1}{n}\sum_{i=1}^n\mathbf b(x_{i,\cdot})\|_\infty=O_P\left(\sqrt{\frac{\log (Kp)}{n}}\right)$. Hence, also $\|\E[\mathbf b(X)]-\frac{1}{n}\sum_{i=1}^n\mathbf b(x_{i,\cdot})\|_\infty^2 = O_P\left(\sqrt{\frac{\log (Kp)}{n}}\right)$. Using \eqref{eq_BoundWBarL1} with $a=0$, it follows that
\begin{equation}\label{eq_ControlDfw2}
\sup_{f^w\in \mathcal F_{M,\mathcal T}^n}D_{f^w}''\leq 4\frac{Ks(1+M)^2}{\min_{j=1, \ldots, p}\lambda_{\min}(\hat \Sigma_j)}O_P\left(\sqrt{\frac{\log (Kp)}{n}}\right)=o_P(1)
\end{equation}
by Assumption \ref{ass_CondBasis}, assertion (2).

We can now put things together. By \eqref{eq_FactorizeCC}, \eqref{eq_FactorizeAfw}, \eqref{eq_SumDfw} and \eqref{eq_DefLfw}, we have for all $f^w\in \mathcal F_{M,\mathcal T}^n$ and $a\in R^q$,
$$\frac{\frac{1}{n}\|\mathbf f^w-\mathbf H a\|_2^2}{w_0^2+\sum_{j=1}^p\frac{1}{n}\|\mathbf f_j^w\|_2^2}\geq B_{f^w, a}\cdot C_{f^w}\cdot (1-|E_{f^w} F_{f^w, a}(D_{f^w}''+L_{f^w, a}U_n)|).$$
By \eqref{eq_SampleCCBound} and Lemma \ref{lem_TechnicalCC}, we only need to bound this expression for $a=a_w$. Putting together \eqref{eq_ControlDfw2}, \eqref{eq_ControlDfw1}, \eqref{eq_ControlFfw}, \eqref{eq_ControlEfw}, \eqref{eq_ControlBfw} and \eqref{eq_ControlCfw} yields $\tau_n^\text{PCA}\geq \tau_0+o_P(1)$. Together with \eqref{eq_CompTrimPCA}, this concludes the proof.

\subsubsection{Proof of Lemma \ref{lem_HatH}}\label{sec_ProofWeakFactorModel}
As in \cite{GuoDoublyDebiasedLasso}, equation (73), let the matrix $\Lambda^2\in \mathbb R^{q\times q}$ be the diagonal matrix with the largest $q$ eigenvalues of the matrix $\frac{1}{np}\mathbf X\mathbf X^T$ as entries and define
\begin{equation}\label{eq_DefO}
O = \frac{1}{np}\Psi\Psi^T\mathbf H^T\hat {\mathbf H}\Lambda^{-2}\in  \mathbb R^{q\times q}.
\end{equation}

We follow the strategy of Section B.2. in \cite{GuoDoublyDebiasedLasso}. For some $C>0$ large enough, and some $c>0$ small enough, define the events (where $h_t, \hat h_t\in \mathbb R^q$ and $e_t\in \mathbb R^p$ are the $t$th row of $\mathbf H$, $\hat{\mathbf H}$ and $\mathbf E$, respectively),
\begin{align*}
\mathcal A_1&=\left\{\max_{1\leq t\leq n}\|h_t\|_2\leq C\sqrt{q\log(nq)}\right\}\\
\mathcal A_2&=\left\{\max_{1\leq i\leq n}\|\Psi e_i/p\|_2\leq C \frac{\sqrt q}{\sqrt p}\sqrt{\log(nq)}\max_{l,j}|\Psi_{l,j}| \right\}\\
\mathcal A_3&=\left\{\max_{1\leq i\leq n} e_i^Te_i/p\leq C\log(np)\right\}\\
\mathcal A_4 &=\left\{\max_{1\leq t\neq i\leq n} |e_i^T e_t/p|\leq C\frac{\sqrt{\log p}\sqrt{\log(np)}}{\sqrt p}\right\}\\
\mathcal A_5&=\left\{\lambda_{\min}(\Lambda)\geq c\frac{\lambda_q(\Psi)}{\sqrt p}\right\}\\
\mathcal A_6&=\left\{\|\mathbf H\|_{op}\leq C\sqrt n,\, \|\mathbf E\|_{op}\leq C(\sqrt n + \sqrt p)\right\}\\
\mathcal A_7&=\left\{\|\mathbf H^T\mathbf H/n - I_q\|_{op}\leq C\sqrt{\frac{q+\log p}{n}}\right\}\\
\mathcal A_8 &=\left\{\|O\|_{op}\leq C\right\}.
\end{align*}
We show that $\mathcal A = \cap_{l=1}^8 \mathcal A_l$ satisfies $\Prob\left(\mathcal A\right)\geq pr(n, p)=1-n^{-c}-p^{-c}-\exp(-cn)-\exp(-cp)$ for some $c>0$. For this, most of the work was already done in the proof of Lemma 8 in \cite{GuoDoublyDebiasedLasso}. For $\mathcal A_1$, observe that $\max_t \|h_t\|_2\leq \sqrt q \max_{t, j}|\mathbf H_{t,j}|$. Since the random variables $\{\mathbf H_{t,j}: 1\leq t\leq n, 1\leq j\leq q\}$ are sub-Gaussian with bounded parameters, we obtain using the union bound
\begin{align*}
    \Prob(\mathcal A_1^c)&\leq \Prob\left(\max_{t, j}|\mathbf H_{t,j}|>C\sqrt{\log(np)}\right)\\
    &\leq \sum_{i,j}\Prob\left(|\mathbf H_{t,j}|>C\sqrt{\log(np)}\right)\\
    &\leq 2nq\exp\left(\frac{-\log(nq)C^2}{C_0^2}\right)\\
    &\leq 2(nq)^{1-C^2/C_0^2}
\end{align*}
for some constant $C_0$ depending on the sub-Gaussian norms of the entries of $H$. Hence, it suffices to take $C>C_0$.

For $\mathcal A_2$, note that
$$\max_{1\leq i\leq n}\|\Psi e_i/p\|_2\leq\max_{1\leq i\leq n}\max_{1\leq l\leq q}\frac{\sqrt q}{p}|\Psi_{l,\cdot}^T e_i|\leq \max_{1\leq i\leq n}\max_{1\leq l\leq q}\frac{\sqrt q}{p}\frac{|\Psi_{l,\cdot}^T e_i|}{\|\Psi_{l,\cdot}\|_2}\max_{1\leq l\leq q}\|\Psi_{l,\cdot}\|_2.$$
Since $\{e_i\}_{1\leq i\leq n}$ are i.i.d. sub-Gaussian vectors, the random variables $\left\{\frac{\Psi_{l,\cdot}^Te_i}{\|\Psi_{l,\cdot}\|_2}:1\leq i\leq n, 1\leq l\leq q \right\}$ are sub-Gaussian with bounded parameters. Hence, by the same argument as before, we have $\Prob(\mathcal A_2)\geq 1-n^{-c}$ for some $c>0$. 

One can show $\Prob(\mathcal A_3\cap\mathcal A_4)\geq pr(n, p)$ by using exactly the same reasoning as for the control of $\mathcal G_5\cap\mathcal G_6$ in the proof of Lemma 8 in Section B.5 of \cite{GuoDoublyDebiasedLasso}. The event $\mathcal A_5\cap \mathcal A_8$ is a superset of the event $\mathcal G_{10}$ in the proof given there, such that one can apply the reasoning from there. The event $\mathcal A_6$ corresponds to the event $\mathcal G_8$ and the event $\mathcal A_7$ corresponds to the event $\mathcal G_1$, such that we can again apply the arguments given there. In total, we indeed obtain $\Prob(\mathcal A)\geq pr(n, p)$.

As in the proof of Lemma 9 in \cite{GuoDoublyDebiasedLasso} (eq. (81)-(85) and following) and noting that $\frac{1}{n}\sum_{i=1}^n\|\hat h_i\|_2^2= q$ as explained there,  we have that
\begin{equation}\label{eq_MaxHO}
\max_{1\leq t\leq n}\|\hat h_t-O^Th_t\|_2\leq \|\Lambda^{-2}\|_{op}\left(2\sqrt{q}\max_{1\leq t\leq n}\|h_t\|_2\max_{1\leq i\leq n}\|\Psi e_i/p\|_2+\sqrt{q}\max_{1\leq t\leq n}\sqrt{\frac{1}{n}\sum_{i=1}^n|\frac{1}{p} e_i^Te_t|^2}\right).
\end{equation}
On the event $\mathcal A_1\cap \mathcal A_2$, we have that
\begin{equation}\label{eq_MaxHO1}
\sqrt q \max_{1\leq t\leq n}\|h_t\|_2\max_{1\leq i\leq n}\|\Psi e_i/p\|_2\lesssim \sqrt q \sqrt{q\log(nq)}\frac{\sqrt q\sqrt{\log(nq)}}{\sqrt p}\max_{l,j}|\Psi_{l,j}|\lesssim \frac{q^{3/2}(\log N)^{3/2}}{\sqrt p}
\end{equation}
using that $\max_{l,j}|\Psi_{l,j}|\lesssim \sqrt{\log(pq)}$ by Assumption \ref{ass_DimAndPsi}, assertion (4). Moreover
$$\sqrt q\sqrt{\frac{1}{n}\sum_{i=1}^n |\frac{1}{p}e_i^t e_t|^2}=\sqrt q \sqrt{\frac{1}{n}\sum_{t\neq i}|\frac{1}{p} e_i^T e_t|^2+\frac{1}{n}|\frac{1}{p}e_t^T e_t|^2}.$$ Hence, on the event $\mathcal A_3\cap \mathcal A_4$, we have that
$$\max_{1\leq t\leq n}\sqrt q\sqrt{\frac{1}{n}\sum_{i=1}^n |\frac{1}{p}e_i^t e_t|^2}\lesssim \sqrt q \sqrt{\frac{\log p\log(np)}{p}+\frac{\log(np)^2}{n}}\lesssim \frac{\sqrt q\log N}{\sqrt p}+\frac{\sqrt q\log N}{\sqrt n}.$$
In total, we get from this, \eqref{eq_MaxHO} and \eqref{eq_MaxHO1} that on $\mathcal A_1\cap \mathcal A_2\cap \mathcal A_3\cap \mathcal A_4\cap \mathcal A_5$,
$$\max_{1\leq t\leq n}\|\hat h_t-O^Th_t\|_2\lesssim\frac{p}{\lambda_q(\Psi)^2}\left(\frac{q^{3/2}(\log N)^{3/2}}{\sqrt p}+\frac{\sqrt q\log N}{\sqrt n}\right).$$
Note that
\begin{align*}
\|\hat {\mathbf H} -\mathbf H^T O\|_{op}&=\sup_{\|z\|_2=1}\|(\hat {\mathbf H} -\mathbf H^T O)z\|_2\\
&\leq \sup_{\|z\|_2=1}\sqrt{n}\max_{1\leq t\leq n}|(\hat h_t-O^Th_t)^Tz|\\
&\leq \sqrt n \max_{1\leq t\leq n}\|\hat h_t-O^Th_t\|_{2}.
\end{align*}
Hence, we have that
$$\frac{1}{\sqrt n}\|\hat {\mathbf H} -\mathbf H^T O\|_{op}\lesssim\frac{p}{\lambda_q(\Psi)^2}\left(\frac{q^{3/2}(\log N)^{3/2}}{\sqrt p}+\frac{\sqrt q\log N}{\sqrt n}\right)\ll \frac{1}{s}$$
by using Assumption \ref{ass_DimAndPsi}, assertion (3), and the first assertion of Lemma \ref{lem_HatH} follows.

For the second assertion, we first follow the proof of Lemma 11 in \cite{FanLargeCovarianceEstimation}. Observe that
$$\|O^T O- I_q\|_{op}\leq\|O^T O- \frac{1}{n}O^T\mathbf H^T\mathbf HO\|_{op}+\|\frac{1}{n}O^T \mathbf H^T\mathbf HO-I_q\|_{op}$$
On the set $\mathcal A_7\cap \mathcal A_8$, we have 
$$\|O^T O- \frac{1}{n}O^T\mathbf H^T\mathbf HO\|_{op}\leq\|O\|_{op}^2\|I_q-\frac{1}{n}\mathbf H^T\mathbf H\|_{op}\lesssim \sqrt{\frac{q+\log p}{n}}\ll \frac{1}{s}.$$
On the set $\mathcal A_6\cap \mathcal A_8$ and using $\|\hat{\mathbf H}\|=\sqrt n$, we have
\begin{align*}
\|\frac{1}{n}O^T\mathbf H^T\mathbf H O-I_q\|_{op}&\leq \|\frac{1}{n}O^T\mathbf H^T\mathbf H O-\frac{1}{n}O^T\mathbf H^T\hat {\mathbf H}\|_{op}+\|\frac{1}{n} O^T\mathbf H^T\hat{\mathbf H}-\frac{1}{n}\hat{\mathbf H}^T\hat{\mathbf H}\|_{op}\\
&\leq\frac{1}{n}\|O\|_{op}\|\mathbf H\|_{op}\|\mathbf H O-\hat{\mathbf H}\|_{op}+\frac{1}{n}\|\hat{\mathbf H}\|_{op}\|\mathbf H O-\hat{\mathbf H}\|_{op}\\
&\lesssim\frac{1}{\sqrt n}\|\mathbf H O-\hat{\mathbf H}\|_{op},
\end{align*}
By using the first assertion, it follows that
\begin{equation}\label{eq_BoundOTO}
    \|O^TO -I_q\|_{op}=o_P\left(\frac{1}{s}\right)
\end{equation}
Observe that
\begin{align*}
\|O O^T-I_q\|_{op}&=\|OO^T-O O^{-1}\|_{op}\leq\|O\|_{op}\|O^T-O^{-1}\|_{op},\\
\|O^T O-I_q\|_{op}&=\|O^TO-O^{-1} O\|_{op}\geq \lambda_{\min}(O)\|O^T-O^{-1}\|_{op}.
\end{align*}
Hence, we have
$$\|O O^T-I_q\|_{op}\leq\frac{\|O\|_{op}}{\lambda_{\min}(O)}\|O^T O-I_q\|_{op}$$

On the set $\mathcal A_8$, we have $\|O\|_{op}\leq C$. Moreover, $\lambda_{\min}(O)=\sqrt{\lambda_{\min}(O^T O)}$. By Weyl's inequality for singular values, we have that
$$|\lambda_{\min}(O^T O)-1|=|\lambda_{\min}(O^T O)-\lambda_{\min}(I_q)|\leq \|O^T O- I_q\|_{op}.$$
Hence, 
$$\lambda_{\min}(O)=\sqrt{\lambda_{\min}(O^T O) -1 + 1}\geq\sqrt{1-\|O^T O- I_q\|_{op}}.$$
It follows that
$$\|I_1-OO^T\|_{op}\lesssim \frac{\|I_q-O^TO\|_{op}}{\sqrt{1-\|I_q-O^TO\|_{op}}}.$$
Combining this with \eqref{eq_BoundOTO} completes the proof.

\subsection{Proof of Theorem \ref{thm_BoundPopCC}}\label{sec_ProofBoundPopCC}
We first remove the intercept $w_0$. Since for $f\in \mathcal F_{\textup{add}}$, $\E[f(X)-w_0-H^Ta]=0$, we have that
$$\frac{\E[(f(X)-H^Ta)^2]}{w_0^2+\sum_{j=1}^p\E[f_j(X_j)^2]}=\frac{w_0^2+\E[(f(X)-w_0-H^T a)^2]}{w_0^2+\sum_{j=1}^p\E[f_j(X_j)^2]}\geq \min \left(1, \frac{\E[(f(X)- w_0 -H^Ta)^2]}{\sum_{j=1}^p\E[f_j(X_j)^2]}\right).$$
Since $\lambda_{\min}(A_{\Psi,\Sigma_E})\leq 1$, we can work with $f(X)-w_0$ instead of $f(X)$. We can now follow the proof of Theorem 1 in \cite{GuoExtremeEigenvalues}, where a similar result without the confounder $H$ is proven. We first standardize $Z_j=X_j/\sqrt{\E[X_j^2]}$ and write $f_j(X_j)=g_j(Z_j)$ with $g_j$ being a rescaled version of $f_j$. Since the Hermite polynomials
$$\psi_m(x)=(m!)^{-1/2}(-1)^m e^{x^2/2}\frac{d^m}{dx^m}e^{-x^2/2}$$
form an orthonormal basis and $Z_j\sim \mathcal N(0,1)$, we can write for $j=1,\ldots, p$
$$g_j(Z_j)=\sum_{j=1}^\infty d_{j,m} \psi_m(Z_j),$$
where the infinite sum is to be understood in the $L_2$ sense. Moreover, we have that
\begin{equation}\label{eq_HermiteDecomp}
\E[f_j(X_j)^2]=\E[g_j(Z_j)^2]= \sum_{m=1}^\infty d_{j,m}^2.
\end{equation}
By equation (9) in \cite{LancasterSomeProperties}, we have for all $j,t=1,\ldots, p$, all $m, n\in \mathbb N$ and all $l=1,\ldots, q$ that
\begin{align*}
\E[\psi_m(Z_j)\psi_n(Z_t)]&=\E[Z_jZ_t]^m \delta_{m,n},\\
\E[\psi_m(Z_j)H_l]&=E[Z_j H_l]\delta_{m,1}.
\end{align*}
where $\delta_{m,n}$ is the Kronecker delta and we used that $\psi_1(x)=x$ for the second identity. It follows that
\begin{align*}
\E[g_j(Z_j)g_t(Z_t)]&=\sum_{m=1}^\infty d_{j,m}d_{t,m}\E[Z_jZ_t]^m\\
\E[g_j(Z_j)H_l]&=d_{j,1}\E[Z_j H_l].
\end{align*}
Hence, we can write for $f\in \mathcal F_{\textup{add}}$
\begin{align*}
&\E[(f(X)-w_0-H^Ta)^2]=\E\left[\left(\sum_{j=1}^p g_j(Z_j)-\sum_{l=1}^q H_l a_l\right)^2\right]\\
&=\sum_{j=1}^p\sum_{t=1}^p\E[g_j(Z_j)g_t(Z_t)]-2\sum_{j=1}^p\sum_{l=1}^qa_l\E[g_j(Z_j) H_l]+\sum_{l=1}^q\sum_{k=1}^qa_l a_k\E[H_lH_k]\\
&=\sum_{j=1}^p\sum_{t=1}^p\sum_{m=1}^\infty d_{j,m} d_{t, m}\E[Z_jZ_t]^m-2\sum_{j=1}^p\sum_{l=1}^qa_ld_{j, 1}\E[Z_t H_l]+\|a\|_2^2.
\end{align*}
If we minimize this over $a\in \mathbb R^q$, we get
\begin{align}
&\E[(f(X)-w_0-H^Ta)^2]\geq \sum_{j=1}^p\sum_{t=1}^p\sum_{m=1}^\infty d_{j,m} d_{t,m}\E[Z_j Z_t]^m-\sum_{l=1}^q\left(\sum_{j=1}^p d_{j,1}\E[Z_j H_l]\right)^2\nonumber \\
&=\sum_{m=2}^\infty\sum_{j=1}^p\sum_{t=1}^p d_{j,m}d_{t,m}\E[Z_jZ_t]^m +\sum_{j=1}^p\sum_{t=1}^pd_{j,1} d_{t,1}\left(\E[Z_j Z_t]-\sum_{l=1}^q \E[Z_jH_l]\E[Z_tH_l]\right)\label{eq_BoundPopCC1}
\end{align}
By the definition of $Z_j$, we have
$$\E[Z_j Z_t]=\frac{\Psi_j^T\Psi_t+(\Sigma_E)_{j,t}}{\sqrt{\|\Psi_j\|_2^2+(\Sigma_E)_{j,j}}\sqrt{\|\Psi_t\|_2^2+(\Sigma_E)_{t,t}}}$$
and
\begin{align*}
    \sum_{l=1}^q \E[Z_j H_l]\E[Z_t H_l]&=\sum_{l=1}^p \frac{\Psi_{l,j}\Psi_{l, t}}{\sqrt{\|\Psi_j\|_2^2+(\Sigma_E)_{j,j}}\sqrt{\|\Psi_t\|_2^2+(\Sigma_E)_{t,t}}}\\
    &=\frac{\Psi_j^T\Psi_t}{\sqrt{\|\Psi_j\|_2^2+(\Sigma_E)_{j,j}}\sqrt{\|\Psi_t\|_2^2+(\Sigma_E)_{t,t}}}
\end{align*}
Using the definition of the matrix $A=A_{\Psi, \Sigma_{E}}$ and Lemma \ref{lem_Lemma2GuoZhang} below, we get from \eqref{eq_BoundPopCC1} that
\begin{align*}
\E[(f(X)-w_0-H^Ta)^2]&\geq \sum_{m=2}^\infty\sum_{j=1}^p\sum_{t=1}^p d_{j,m}d_{t,m}\E[Z_jZ_t]^m +\sum_{j=1}^p\sum_{t=1}^pd_{j,1} d_{t,1}A_{j,t}\\
&\geq \sum_{m=2}^\infty\sum_{j=1}^p d_{j,m}^2\lambda_{\min}(\E[ZZ^T])+\sum_{j=1}^p d_{j, 1}^2\lambda_{\min}(A)\\
&\geq\min\left(\lambda_{\min}(\E[ZZ^T], \lambda_{\min}(A)\right)\sum_{j=1}^p \sum_{m=1}^\infty d_{j,m}^2\\
&=\min\left(\lambda_{\min}(\E[ZZ^T], \lambda_{\min}(A)\right)\sum_{j=1}^p\E[f_j(X_j)^2],
\end{align*}
where we used \eqref{eq_HermiteDecomp} in the last step. Finally, we observe that $\E[ZZ^T]=\Lambda_{\Psi, \Sigma_E}(\Psi^T\Psi+\Sigma_E)\Lambda_{\Psi, \Sigma_E}=A_{\Psi, \Sigma_E}+\Lambda_{\Psi, \Sigma_E}\Psi^T\Psi\Lambda_{\Psi, \Sigma_E}$. Since both $A_{\Psi, \Sigma_E}$ and $\Lambda_{\Psi, \Sigma_E}\Psi^T\Psi\Lambda_{\Psi, \Sigma_E}$ are positive semi definite, we have that $\lambda_{\min}(\E[ZZ^T])\geq \lambda_{\min}(A_{\Psi, \Sigma_E})$, which concludes the proof.

The following Lemma is a special case of Lemma 2 in \cite{GuoExtremeEigenvalues}.
\begin{lemma}\label{lem_Lemma2GuoZhang}
Let $Z\sim \mathcal N_p(0, \E[ZZ^T])$ be a Gaussian vector in $\mathbb R^p$ with $\E[Z_j^2]=1$ for all $j=1,\ldots, p$.
For all $m\geq 1$ and all $h\in \mathbb R^p$ with $\|h\|_2^2=1$, we have
$$\sum_{j=1}^p\sum_{t=1}^p \E[Z_jZ_t]^m h_j h_t\geq \lambda_{\min}(\E[ZZ^T]).$$
\end{lemma}

\section{Remaining Proofs}
\subsection{Proof of Lemma \ref{lem_QXB}}\label{sec_ProofQXB}
By Proposition 3 in \cite{GuoDoublyDebiasedLasso}, we have that with probability larger than $1-\exp(-cn)$ for some constant $c>0$, $\lambda_{q+1}(\frac{1}{n}X^TX)\lesssim \max(1, p/n)$. For $r=\min(n, p)$ large enough, we have $\lfloor \rho r\rfloor\geq q+1$, hence for both $Q=Q^\text{trim}$ and $Q=Q^\text{PCA}$, we have $\frac{1}{n}\|QX\|_{op}^2\lesssim \max(1, p/n)$. Hence,
$$\frac{1}{n}\|QXb\|_2^2\lesssim \|b\|_2^2\max(1, \frac{p}{n}).$$
To control $\|b\|_2$, we follow the proof of Lemma 2 in \cite{GuoDoublyDebiasedLasso}. From the definition of $b$ and the Woodbury identity \cite{GolubMatrixComputations}, we have
\begin{align*}
b=\E[XX^T]^{-1}\Psi^T\psi&=(\Psi^T\Psi+\Sigma_E)^{-1}\Psi^T\psi\\
&=\left(\Sigma_E^{-1}-\Sigma_E^{-1}\Psi^T(I_q	+\Psi\Sigma_E^{-1}\Psi^T)^{-1}\Psi\Sigma_E^{-1}\right)\Psi^T\psi\\
&=\Sigma_E^{-1}\Psi^T(I_q+\Psi\Sigma_E^{-1}\Psi^T)^{-1}\psi
\end{align*}
With $D_E=\Psi\Sigma_E^{-1/2}$, we have$\|b\|_2\leq \|\Sigma_E^{-1/2}\|_{op}\|D_E^T(I_q+D_E D_E^T)^{-1}\|_{op}\|\psi\|_2$ and
\begin{align*}
    &\|D_E^T(I_q+D_E D_E^T)^{-1}\|_{op}^2=\lambda_{\max}(D_E^T(I_q+D_E D_E^T)^{-2}D_E)\\
    &=\max_{1\leq l\leq q}\left(\frac{\lambda_l(D_E)}{1+\lambda_l(D_E)^2}\right)^2 \leq\max_{1\leq l\leq q}\frac{1}{\lambda_l(D_E)^2}=\frac{1}{\lambda_q(D_E)^2}
\end{align*}
Since $D_E=\Psi\Sigma_E^{-1/2}$ and by Assumption \ref{ass_ConditionsModel1}, $c\leq \lambda_{\min}(\Sigma_E^{-1})\leq \lambda_{\max}(\Sigma_E^{-1})\leq C$, we get the result.
\subsection{Proof of Lemma \ref{lem_ApproxError}}\label{sec_ProofApproxError}
From Theorem (6) in Chapter XII in \cite{DeBoorSplines} applied to the functions $f_j^0\circ F_j^{-1}$, we have that for all $j=1,\ldots, p$, there exists $\beta_j^\ast$ such that the functions $g_j^\ast=b_0(\cdot)^T \beta_j^\ast$ satisfy $\|f_j^0\circ F_j^{-1}-g_j^\ast\|_{\infty, [0,1]}\leq c h^2$ for some constant $c$ independent of $j=1,\ldots, p$. It follows that also the functions $f_j^\ast(\cdot)=b_j(\cdot)^T\beta_j^\ast=g_j(F_j(\cdot))$ satisfy $\|f_j^\ast-f_j^0\|_{\infty}\leq c h^2$. Hence, $\|f_j^\ast-f_j^0\|_{L_2}^2=\E[(f_j^\ast(X_j)-f_j^0(X_j))^2]\leq c^2 h^4$. Since $h=1/(K-3)$, the claim follows.

\Revision{
\subsection{Proof of Corollary \ref{cor_FinalRate}}\label{sec_ProofFinalRate}
We only need to show that \eqref{eq_BoundRnNew} reduces to \eqref{eq_FinalRate} under the conditions of Corollary \ref{cor_FinalRate}. Under these conditions, we have from Theorem \ref{thm_BoundCC} and Theorem \ref{thm_BoundPopCC} that $\tau_n\gtrsim 1$ with high probability. Since $n\lesssim p$, $\lambda_q(\Psi)\asymp \sqrt p$ and $\|\psi\|_2\lesssim 1$, we can choose $\lambda_2$ such that the first term in the definition \eqref{eq_Lambda} of $\lambda$ dominates. We obtain 
\begin{equation}\label{eq_IntermediateRn}
    r_n = O_P\left( s\sqrt{\frac{K \log p}{n}} +\frac{1}{\sqrt {nK\log p}}+ \frac{s}{K^2}+\frac{s}{\sqrt n}+\frac{s^2}{K^4}\sqrt{\frac{n}{K\log p}}+s^2\sqrt{\frac{1}{n K\log p}}\right).
\end{equation}
Plugging in $K\asymp (n/\log p)^{2/5}$, the fifth term dominates and the claim follows.

To prove the claim of Remark \ref{rmk_RemarkFinalRate}, we instead plug $K\asymp(ns/\log p)^{1/5}$ into \eqref{eq_IntermediateRn}.
}

\subsection{Proof of Lemma \ref{lem_ConditionBasis}}\label{sec_ProofConditionBasis}
Define the random variables $U_j=F_j^{-1}(X_j)$, which are now uniformly distributed in $[0,1]$.
We follow the proof of Lemma 6.1 and Lemma 6.2 in \cite{ZhouLocalAsymptoticsForRegressionSplines}. We apply the steps given there to the random variables $U_j$. The difference is that we need the $o(1)$ in the statements of the lemmas there uniformly in $j=1,\ldots, p$. Following the steps of the proof and using that we have equidistant knots and uniform distributions, we arrive at \eqref{eq_EVBasisPop} and \eqref{eq_EVBasisSam} with $S_n=C \sup_{j=1,\ldots, p}\sup_{y\in [0,1]}|Q_n^j(y)-Q(y)|$, where
$Q_n^j(y)=\frac{1}{n}\sum_{i=1}^n\mathbbm 1\{F_j^{-1}(x_{i,j})\leq y\}$ is the empirical distribution function of $U_j$ and $Q(y)=y$ is the distribution function of $U_j$. It remains to prove that $S_n=o_P(h)$. For this, let for $m=0,\ldots, M$, $y_m=m/M$. For $y\in[0,1]$, there exists $m\in\{0,\ldots, M-1\}$ such that $y\in[y_m, y_{m+1}]$. If $Q_n^j(y)-Q(y)\geq 0$, we have using that both $Q_n^j$ and $Q$ are non-decreasing,
\begin{align*}
|Q_n^j(y)-Q(y)|&\leq Q_n^j(y_{m+1})-Q(y_m)\\
&\leq |Q_n^j(y_{m+1})-Q(y_{m+1})|+|Q(y_{m+1})-Q(y_m)|\\
&=|Q_n^j(y_{m+1})-Q(y_{m+1})|+\frac{1}{M}
\end{align*}
since $Q(x)=x$ for all $x\in [0,1]$. Similarly, if $Q_n^j(y)-Q(y)< 0$, we have
\begin{align*}
|Q_n^j(y)-Q(y)|&\leq Q(y_{m+1})-Q_n^j(y_m)\\
&\leq |Q(y_{m+1})-Q(y_{m})|+|Q(y_{m})-Q_n^j(y_m)|\\
&= |Q(y_{m})-Q_n^j(y_m)|+\frac{1}{M}
\end{align*}
In any case, we have $|Q_n^j(y)-Q(y)|\leq \sup_{j=1,\ldots, p}\sup_{m=1,\ldots, M} |Q_n^j(y_m)-Q(y_m)|+1/M$. Hence, also
$$\sup_{j=1,\ldots, p}\sup_{y\in [0,1]}|Q_n^j(y)-Q(y)|\leq \sup_{j=1,\ldots, p}\sup_{m=1,\ldots, M} |Q_n^j(y_m)-Q(y_m)|+1/M.$$
Since the random variables $Q_n^j(y_m)-Q(y_m)=\frac{1}{n}\sum_{i=1}^n(\mathbbm 1\{F_j^{-1}(x_{i,j})\leq y_m\}- Q(y_m))$ are averages of i.i.d. uniformly bounded random variables with mean zero, we have that
$$\sup_{j=1,\ldots, p}\sup_{m=1,\ldots, M} |Q_n^j(y_m)-Q(y_m)|=O_P\left(\sqrt{\frac{\log(pM)}{n}}\right),$$
see for example Lemma 14.13 in \cite{BuehlmannHDStats}. Choosing $M=\sqrt n$ yields
$$\sup_{j=1,\ldots, p}\sup_{y\in [0,1]}|Q_n^j(y)-Q(y)|=O_P\left(\sqrt{\frac{\log p +\log n}{n}}\right).$$
Since $h=1/(K-3)$ and $K\sqrt{\frac{\log p+\log n}{n}}=o(1)$ by Assumption \ref{ass_ConstructionBasis}, we have that $S_n=o_P(h)$, which concludes the proof.

\Revision{
\section{Minimal Requirements for Consistency}\label{sec_AppendixConsistency}
\begin{corollary}\label{cor_Consistent}
    Under Assumptions \ref{ass_ConditionsModel0}-\ref{ass_ApproxError}, assume that the matrix $A_{\Psi, \Sigma_E}$ defined in \eqref{eq_DefAPsi} satisfies $\lambda_{\min}(A_{\Psi, \Sigma_E})\gtrsim 1$. Moreover, assume that either
    \begin{equation}\label{eq_CondConsistencyA}
        \lambda_q(\Psi)^2\gg\|\psi\|_2^2\sqrt{\frac{n}{K\log p}}\max\left(\frac{p}{n}, \sqrt{\frac{n}{K\log p}}\right)\quad \text{and} \quad s \ll \left(\frac{K\log p}{n}\right)^{1/4}\max(K^2, \sqrt n)
    \end{equation}
    or
    \begin{equation}\label{eq_CondConsistencyB}
    \begin{split}
        \|\psi\|_2^2\max(1, p^2/n^2)\lesssim \lambda_q(\Psi)^2\lesssim \|\psi\|_2^2\frac{n}{K\log p} \text{ and}\\
        s \ll \min\left(\frac{\lambda_q(\Psi)}{\|\psi\|_2}, \sqrt{\frac{\|\psi\|_2}{\lambda_q(\Psi)}} K^2, \sqrt{\frac{\|\psi\|_2}{\lambda_q(\Psi)}} \sqrt{n}\right).
    \end{split}    
    \end{equation}
    holds.
    Then, we can choose $\lambda_2$ in the definition \eqref{eq_Lambda} of $\lambda$ such that 
    $$|\beta_0^0-\hat\beta_0|+\sum_{j=1}^p\|f_j^0-\hat f_j\|_{L_2}=o_P(1).$$
    In particular, $\hat f$ is a consistent estimator of $f^0$.
\end{corollary}
\begin{proof}
    Under the conditions of Corollary \ref{cor_Consistent}, it follows from Theorem \ref{thm_BoundCC} and \ref{thm_BoundPopCC} that $\tau_n\gtrsim 1$ with high probability. From \eqref{eq_BoundRnNew}, it follows that we need to show
    \begin{align}
        \lambda_q(\Psi)^2&\gg \frac{\|\psi\|_2^2 \max(1, p/n)}{\lambda}\label{eq_Consistency1}\\
        s&\ll\min\left(\frac{1}{\lambda}, K^2, \sqrt{n}, \sqrt{\lambda}K^2, \sqrt{\lambda n}\right)\label{eq_Consistency2}
    \end{align}

    From the definition \eqref{eq_Lambda} of $\lambda$, we have $\lambda = \lambda_1 + \lambda_2$ with $\lambda_1\asymp \sqrt{\frac{K\log p}{n}}$ and $\lambda_2$ chosen in a way such that $\lambda_2\gg \frac{\|\psi\|_2}{\sqrt{1+\lambda_q(\Psi)^2}}$.
    
    If \eqref{eq_CondConsistencyA} holds, we know that $\sqrt{\frac{K\log p}{n}}\gg \frac{\|\psi\|_2}{\sqrt{1+\lambda_q^2(\Psi)}}$ and hence we can find $\lambda_2$ such that $\lambda\asymp \sqrt{\frac{K\log p}{n}}$. From assertion (1) of Assumption \ref{ass_DimAndPsi}, it follows that $\sqrt{n/(K\log p)}\gg 1$. Hence, \eqref{eq_Consistency1} follows. Assertion (1) of Assumption \ref{ass_DimAndPsi} implies that $s\ll \sqrt{\frac{n}{K\log p}} = \frac{1}{\lambda}$ and $\lambda = \sqrt{K\log p/n}\ll 1$. Hence, also \eqref{eq_Consistency2} follows from \eqref{eq_CondConsistencyA}.

    If \eqref{eq_CondConsistencyB} holds, we choose $\lambda_2$ such that $\lambda_2\gg\frac{\|\psi\|_2}{\lambda_q(\Psi)}$ and $s \ll 1/\lambda_2$. It follows that $\lambda\asymp \lambda_2$. Note that the first equation in \eqref{eq_CondConsistencyB} implies that $\lambda_q(\Psi)/\|\psi\|_2\gtrsim 1$ and hence, \eqref{eq_Consistency2} follows from the second equation in \eqref{eq_CondConsistencyB}. On the other hand, the first equation in \eqref{eq_CondConsistencyB} implies that
    $$\lambda_q(\Psi)^2\gtrsim \lambda_q(\Psi)\|\psi\|_2\max(1, p/n)=\frac{\|\psi\|_2^2\max(1, p/n)}{\|\psi\|_2/\lambda_q(\Psi)}\gg \frac{\|\psi\|_2^2\max(1, p/n)}{\lambda}.$$
    This is precisely \eqref{eq_Consistency1}, which completes the proof.
\end{proof}
}

\section{Additional Simulations}\label{sec_AdditionalSimulations}
\subsection{Toeplitz Covariance Matrix for the Error $E$}\label{sec_ToeplitzCov}
\subsubsection{Varying $n$}
\Revision{In Figures \ref{fig_VarNToe08EqualCI} and \ref{fig_VarNToe08DecreasingCI}, we see the same simulation scenarios as in Section \ref{sec_VarN}, but with Toeplitz covariance structure for $E$, concretely $\Sigma_E=\textup{Toeplitz}(0.8)$, where the matrix $\textup{Toeplitz}(\rho)\in \mathbb R^p$ has entries $(\rho^{|i-j|})_{i,j=1,\ldots, p}$. The picture is completely the same as before in the sense that in the setting \textit{equal confounding influence}, the deconfounded method and the estimated factors method both outperform the naive method, whereas in the setting \textit{decreasing confounding influence} only the deconfounded method shows good performance.}
\begin{figure}
\centering
\includegraphics[width=0.91\textwidth]{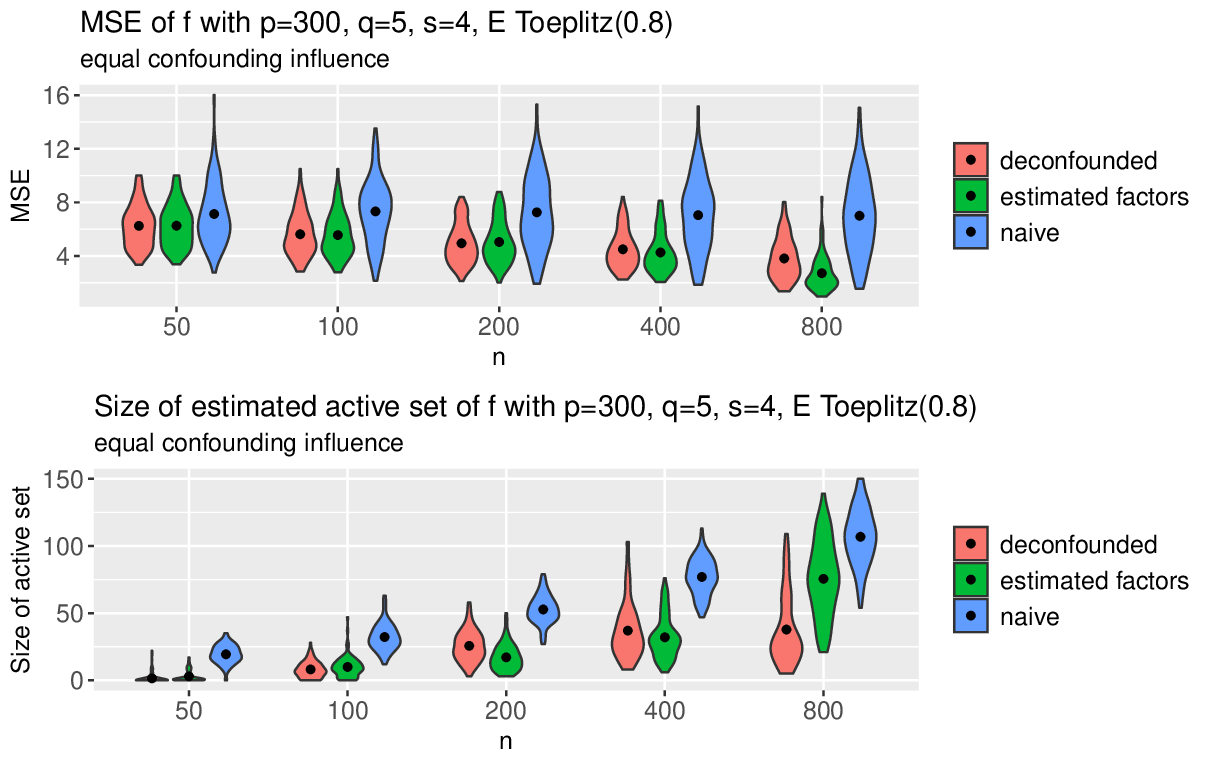}
\caption{MSE (top) and size of estimated active set (bottom) for $\Sigma_E=\textup{Toeplitz}(0.8)$ and varying $n$ in the setting \textit{equal confounding influence}. 
}
\label{fig_VarNToe08EqualCI}
\end{figure}

\begin{figure}
\centering
\includegraphics[width=0.91\textwidth]{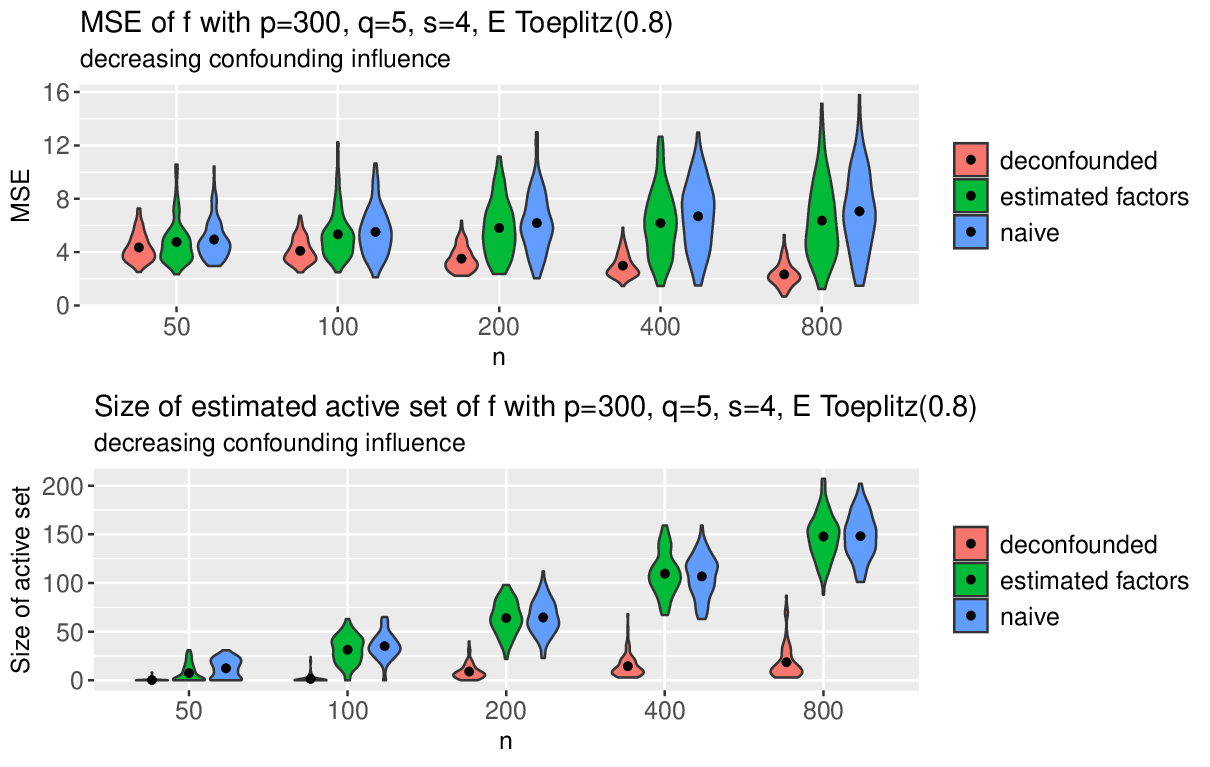}
\caption{MSE (top) and size of estimated active set (bottom) for $\Sigma_E=\textup{Toeplitz}(0.8)$ and varying $n$ in the setting \textit{decreasing confounding influence}. 
}
\label{fig_VarNToe08DecreasingCI}
\end{figure}

\subsubsection{Varying $p$}
\Revision{In Figures \ref{fig_VarPToe08EqualCI} and \ref{fig_VarPToe08DecreasingCI}, we see the same simulation scenarios as in Section \ref{sec_VarP}, but with Toeplitz covariance structure for $E$, concretely $\Sigma_E=\textup{Toeplitz}(0.8)$, where the matrix $\textup{Toeplitz}(\rho)\in \mathbb R^p$ has entries $(\rho^{|i-j|})_{i,j=1,\ldots, p}$. Again, the picture is the same as before.}

\begin{figure}
\centering
\includegraphics[width=0.91\textwidth]{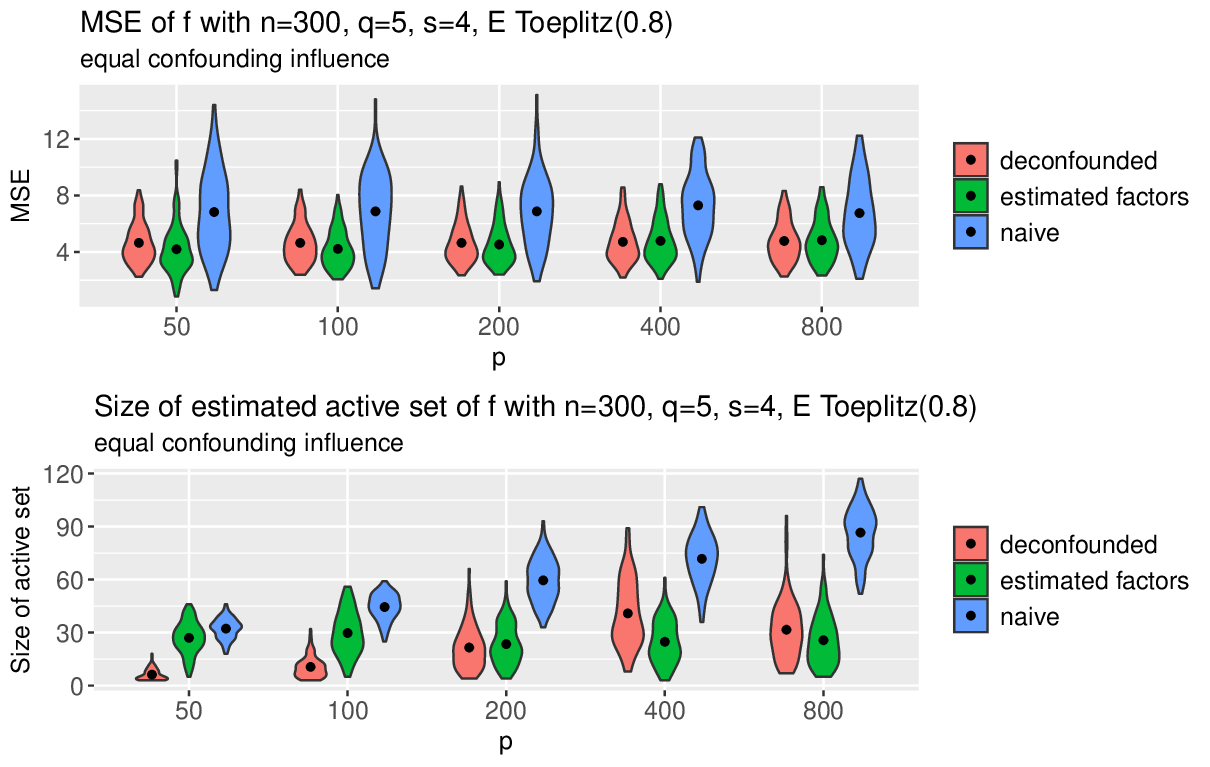}
\caption{MSE (top) and size of estimated active set (bottom) for $\Sigma_E=\textup{Toeplitz}(0.8)$ and varying $p$ in the setting \textit{equal confounding influence}. 
}
\label{fig_VarPToe08EqualCI}
\end{figure}

\begin{figure}
\centering
\includegraphics[width=0.91\textwidth]{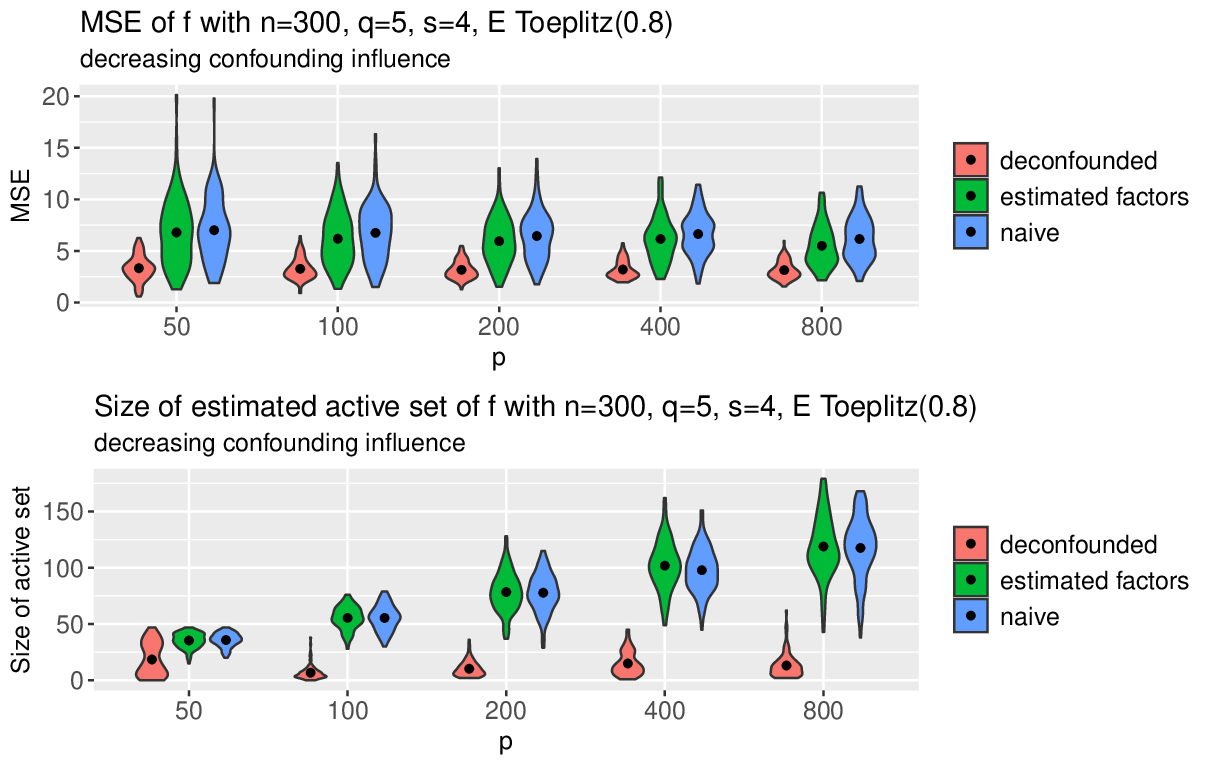}
\caption{MSE (top) and size of estimated active set (bottom) for $\Sigma_E=\textup{Toeplitz}(0.8)$ and varying $p$ in the setting \textit{decreasing confounding influence}. 
}
\label{fig_VarPToe08DecreasingCI}
\end{figure}

\subsection{Varying the Denseness of the Confounding}
\label{sec_VarCProp}
We investigate the effect of the denseness assumption by varying the proportion of covariates \Revision{$X_j$ affected by each confounder $H_l$.} For this, we fix $n=400$, $p=500$, $q=5$ and $\Sigma_E= I_{p}$. We keep the setting described in Section \ref{sec_SimResults} but the entries of the matrix $\Psi$ are now i.i.d. $\textup{Unif}[-1, 1]\cdot \textup{Bernoulli}(\mathsf{prop})$, where $\mathsf{prop}\in [0,1]$ is the proportion of covariates affected by each confounder. That is, a fraction of $1-\mathsf{prop}$ of the entries of $\Psi$ are set to $0$. For each value of $\mathsf{prop}$, we simulate 100 data sets. \Revision{The same plots as before can be found in Figures \ref{fig_VaryCPEqualCI} and \ref{fig_VaryCPDecreasingCI}}. When $\mathsf{prop} = 0$, this corresponds to $X=E$, that is, the confounding does not affect $X$. Hence, the contribution $\psi^T H$ is an error term independent of $X$. We observe that in this case, the deconfounded method performs slightly worse than the naive method, as there is still some signal removed by using a spectral transformation. On the other hand, we see from the plot that the deconfounded method outperforms the naive method even if the confounding only affects a small proportion of the covariates. \Revision{Comparing the deconfounded method to the estimated factors method, the picture is analogous to the previous simulations, i.e. in the setting \textit{equal confounding influence}, the estimated factors method performs slightly better in terms of MSE, but in the setting \textit{decreasing confounding influence}, the deconfounded method performs much better than the estimated factors method.} We conclude that deconfounding is useful also if the confounding is not very dense.

\begin{figure}
\centering
\includegraphics[width=0.91\textwidth]{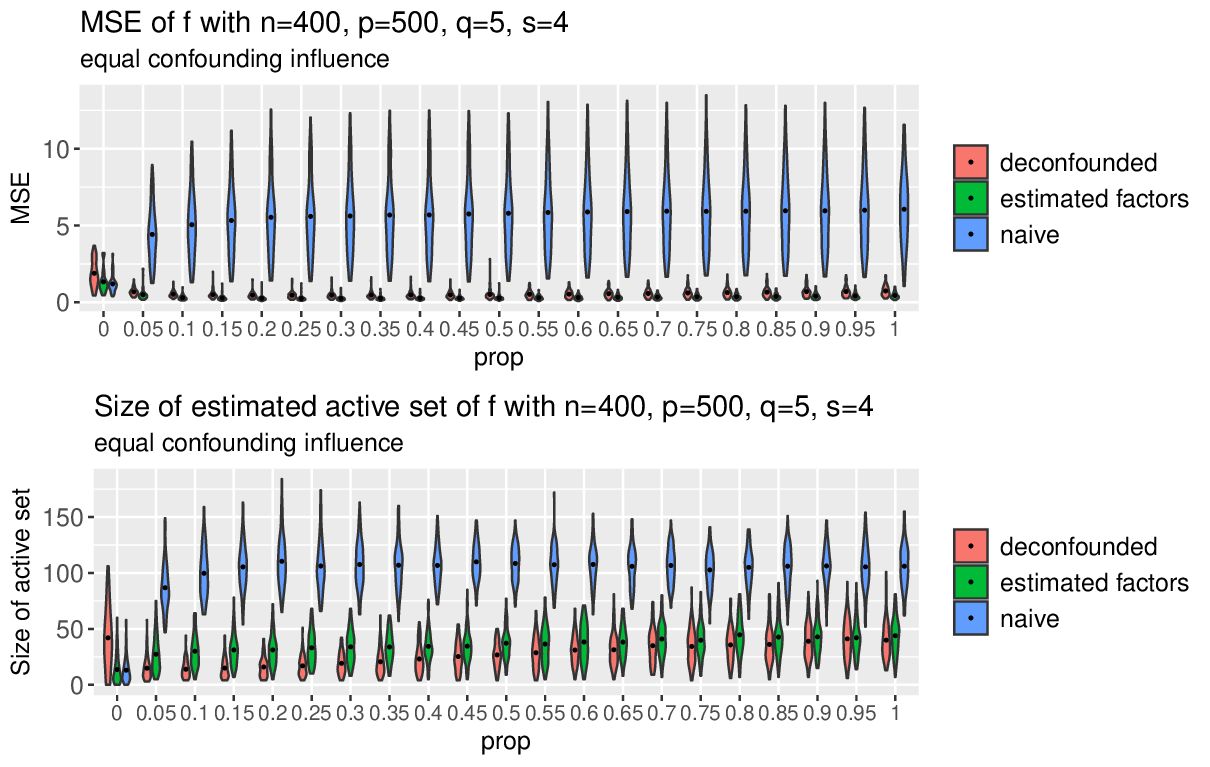}
\caption{MSE (top) and size of the estimated active set (bottom) for varying the denseness of the confounding in the setting \textit{equal confounding influence.}}
\label{fig_VaryCPEqualCI}
\end{figure}

\begin{figure}
\centering
\includegraphics[width=0.91\textwidth]{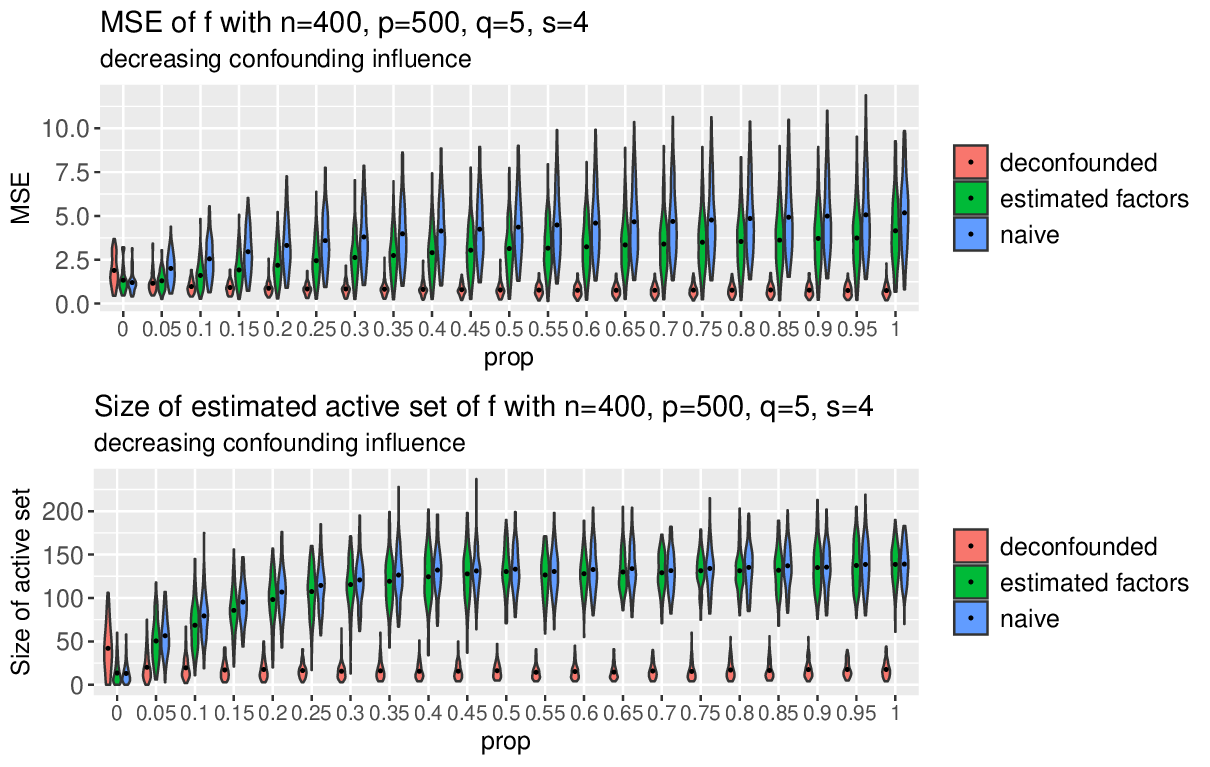}
\caption{MSE (top) and size of the estimated active set (bottom) for varying the denseness of the confounding in the setting \textit{decreasing confounding influence.}}
\label{fig_VaryCPDecreasingCI}
\end{figure}

\subsection{Nonlinear Confounding Effects}\label{sec_Nonlinear}
We now consider the following misspecified version of \eqref{eq_additive}, where the confounding acts potentially nonlinearly on both $X$ and $Y$.
$$Y=f^0(X)+\eta_\beta(H^T\psi)+e\text{ and }X_j=\eta_\alpha(\Psi_j^TH) + E_j,\, j=1,\ldots, p,$$
for some nonlinear functions $\eta_\alpha, \eta_\beta:\mathbb R\to\mathbb R$.
For our simulations, we use the family of functions $\eta_\alpha(t)=(1-\alpha)t+\alpha |t|, \, \alpha\in [0,1]$, that is $\eta_\alpha(t)$ interpolates between $t$ and $|t|$. Otherwise, we use the setup from Section \ref{sec_SimResults} \Revision{in both settings \textit{equal confounding influence} and \textit{decreasing confounding influence}}. As before, we fix $n=400$, $p=500$, $q=5$ and $\Sigma_E=I_{p}$. We vary $\alpha$ and $\beta$ on a grid of values in $[0,1]$ and simulate $100$ data sets for each setting and calculate the mean squared errors $\|\hat f-f^0\|_{L_2}^2$ for the deconfounded method, \Revision{ the naive method and the estimated factors method}.
\Revision{In Figure \ref{fig_NLRatioEqualCI}, we report the ratio of the average MSEs for the setting \textit{equal confounding influence}. The left panel shows the ratio of the average MSE of the deconfounded method and the average MSE of the naive method, where the averages are taken over the $100$ simulated data sets.} Values less than $1$ indicate a smaller average MSE for the deconfounded method, whereas values larger than $1$ indicate that the naive method has a smaller average MSE. We see that for a wide range of combinations of $\alpha$ and $\beta$, the results are in favor of the deconfounded method. We observe that deconfounding slightly worsens the performance of the algorithm only if $\alpha$ is close to $1$ and $\beta$ close to $0$ (i.e. the confounding acts very nonlinearly on $X$ and almost linearly on $Y$ or if $\alpha$ is close to $0$ and $\beta$ is close to $1$ (i.e. the confounding acts almost linearly on $X$ and very nonlinearly on $Y$). Intuitively, in such settings, the contribution of the confounding to $X$ is almost orthogonal to the contribution of the confounding to $Y$; hence, applying the trim transformation is not helpful in such settings. However, we see that for slightly to moderately nonlinear confounding effects in $X$ and $Y$, applying the deconfounded method always improves the performance compared to the naive method.
\Revision{The right panel of Figure \ref{fig_NLRatioEqualCI} shows the ratio of the average MSE of the deconfounded method and the average MSE of the estimated factors method. As in the previous simulations, we observe that the estimated factors method performs moderately better in terms of MSE than the deconfounded method, at least if both $\alpha$ and $\beta$ are close to $0$, i.e. the confounding is close to linear. This changes, when we consider the setting \textit{decreasing confounding influence} in Figure \ref{fig_NLRatioDecreasingCI}. We can see that one can gain a lot in terms of MSE by using the deconfounded method compared to both the naive and the estimated factors method. Only in the edge cases where either the confounding acts very nonlinearly either on $X$ or on $Y$, the naive method and the estimated factors method perform slightly better.}

\begin{figure}
\centering
\includegraphics[width=0.91\textwidth]{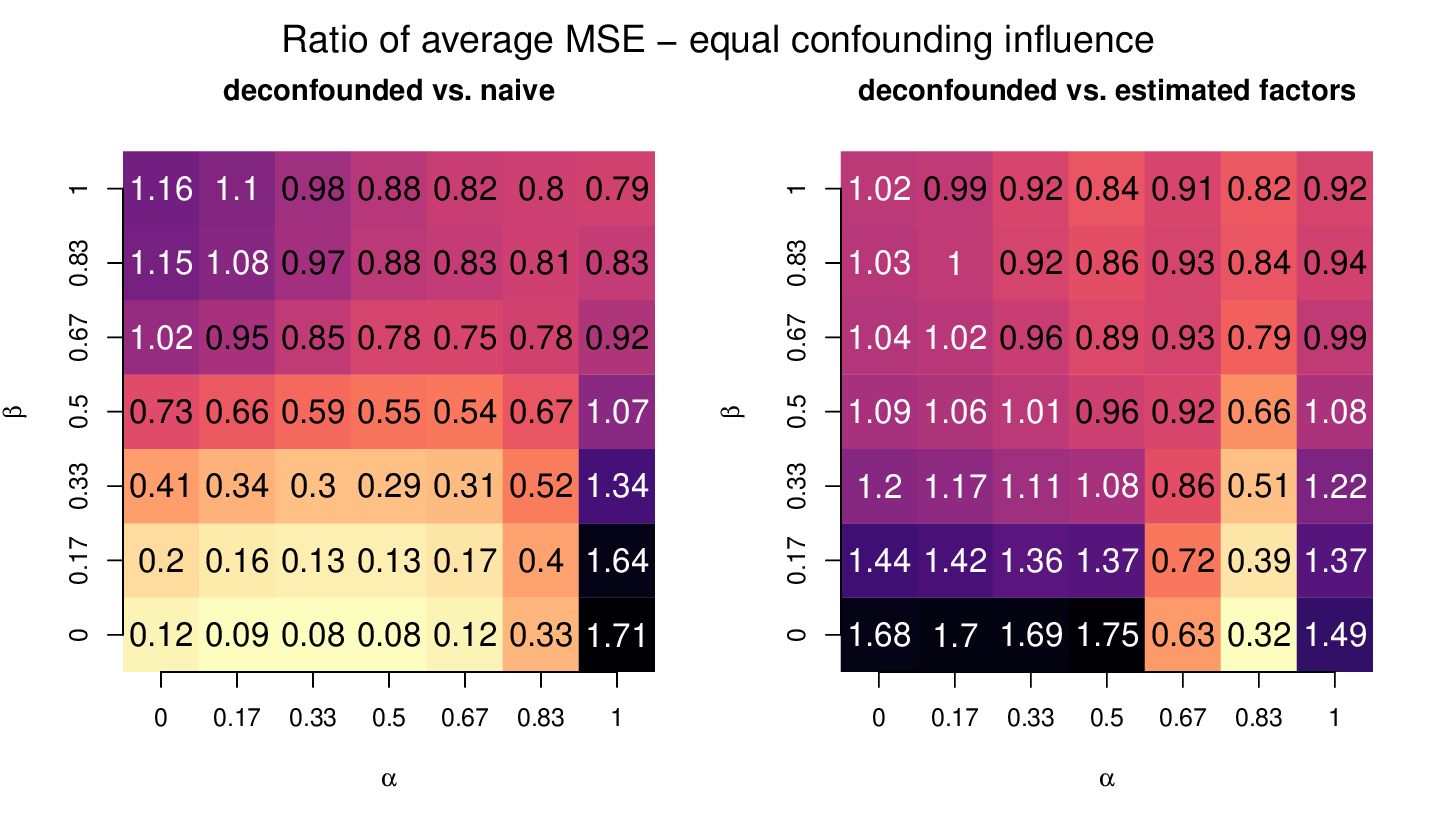}
\caption{Left: Ratio of average MSE for the deconfounded method and average MSE for the naive method. Right: Ratio of average MSE for the deconfounded method and average MSE for the estimated factors method. Values smaller than $1$ are in favor of the deconfounded method, whereas values larger than $1$ are in favor of the other method.}
\label{fig_NLRatioEqualCI}
\end{figure}

\begin{figure}
\centering
\includegraphics[width=0.91\textwidth]{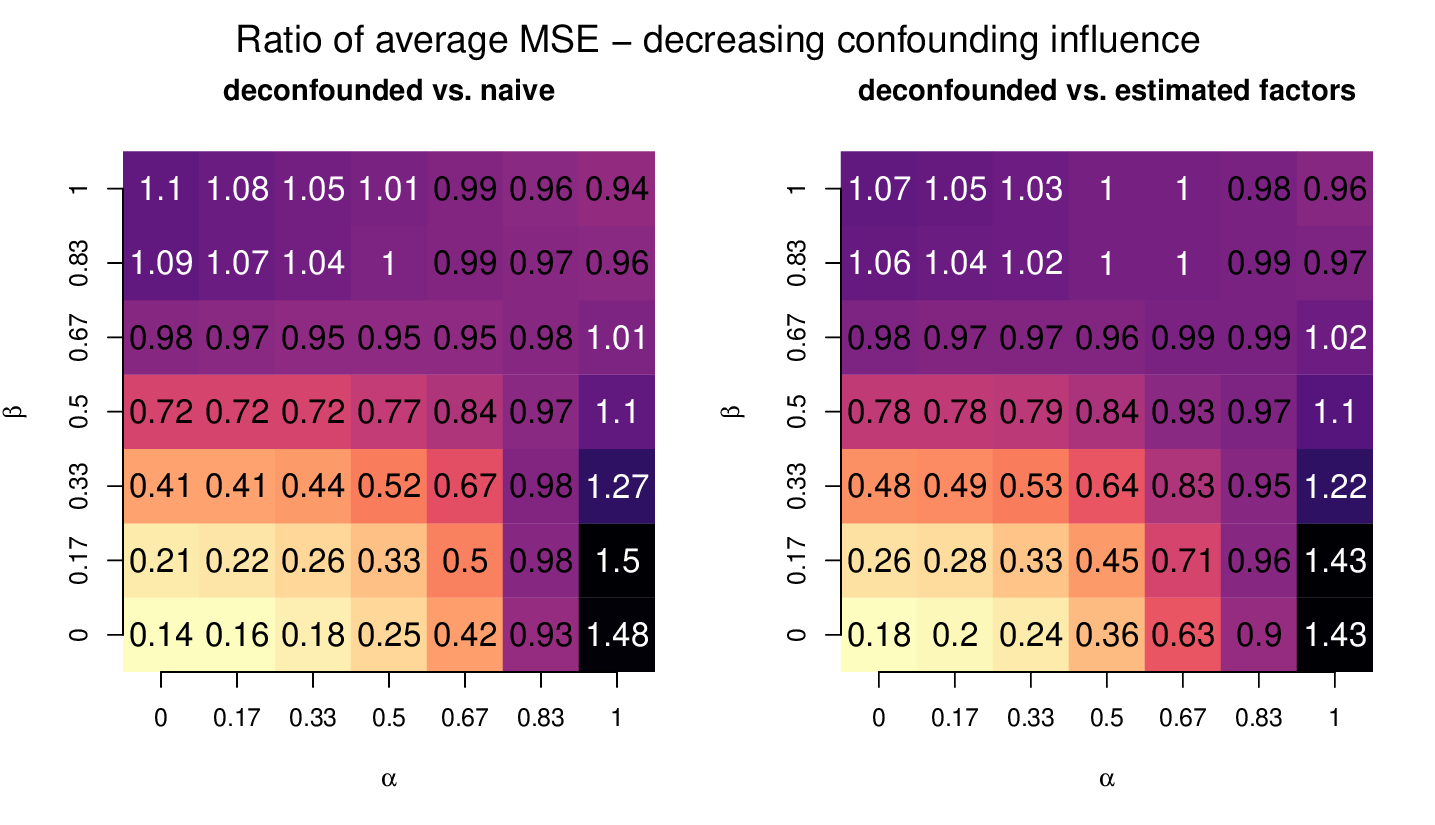}
\caption{Left: Ratio of average MSE for the deconfounded method and average MSE for the naive method. Right: Ratio of average MSE for the deconfounded method and average MSE for the estimated factors method. Values smaller than $1$ are in favor of the deconfounded method, whereas values larger than $1$ are in favor of the other method.}
\label{fig_NLRatioDecreasingCI}
\end{figure}

\bibliography{Literature_HDAM_OL}

\end{document}